\documentclass[11pt]{article}
\usepackage{epsfig}
\usepackage[margin=1in]{geometry}
\usepackage{amssymb, amsmath, mathrsfs, dsfont, amsthm, graphicx, hyperref, blindtext,todonotes,soul, caption, enumitem, color}

\newtheorem{lemma}{Lemma}[section]
\newtheorem{theorem}[lemma]{Theorem}
\newtheorem{corollary}[lemma]{Corollary}

\newtheorem{proposition}[lemma]{Proposition}
\newtheorem{definition}[lemma]{Definition}

\newtheorem{remark}[lemma]{Remark}
\let\lutzremark=\remark
\def\remark{\lutzremark\normalfont}

\newtheorem{notation}[lemma]{Notation}

\def\be{\begin{equation}}
\def\ee{\end{equation}}
\def\bea{\begin{eqnarray}}
\def\eea{\end{eqnarray}}
\def\bes{\begin{eqnarray*}}
\def\ees{\end{eqnarray*}}
\def\ba{\begin{array}}
\def\ea{\end{array}}
\def\bpm{\begin{pmatrix}}
\def\epm{\end{pmatrix}}

\def\nn{\nonumber}
\def\lb{\label}
\def\bs{\setminus}
\def\pt{\partial}

\def\R{{\bf R}}
\def\C{{\bf C}}
\def\Z{{\bf Z}}

\def\N{{\bf N}}
\def\U{{\bf U}}

\def\aa{{\alpha}}
\def\bb{{\beta}}
\def\ga{{\gamma}}
\def\Ga{{\Gamma}}

\def\th{{\theta}}

\def\om{{\omega}}

\def\ep{{\epsilon}}

\def\lm{{\lambda}}

\def\sg{{\sigma}}
\def\dm{{\diamond}}
\def\Sg{{\Sigma}}

\def\<{{\langle}}
\def\>{{\rangle}}
\def\d{{\mathrm{d}}}
\def\pt{\partial}

\def\cR{{\cal R}}
\def\cA{{\cal A}}
\def\B{{\cal B}}

\def\P{{\cal P}}

\def\diag{{\rm diag}}

\def\span{{\rm span}}

\def\Sp{{\rm Sp}}

\def\r{\right}
\def\l{\left}
\def\ii{\sqrt{-1}}
\def\td{\tilde}

\def\ol{\overline}
\def\td#1{\tilde{#1}}

\title{Linear Stability of Elliptic Relative Equilibria of
	Restricted Four-body Problem}

\author{Bowen Liu$^{1}$, \thanks{Partially supported by NSFC (No. 11131004, No. 11671215) and Sino-German (CSC-DAAD) Postdoc Scholarship Program (CSC No. 201800260010 and DAAD No. 91696544) funded by China Scholarship Council and Deutscher Akademischer Austausch Dienst.
		E-mail: bowen.liu@math.uni-augsburg.de}\quad
	Qinglong Zhou$^{2} $\thanks{Partially supported by the Natural Science Foundation of Zhejiang Province (No. Y19A010072) and the Fundamental Research Funds for the Central Universities (No. 2019QNA3002).
		E-mail: zhouqinglong@zju.edu.cn}, \\
	$^{1}$ Institut f\"ur Mathematik\\ Universität Augsburg, D-86135 Augsburg, Germany\\
	$^{2}$ Department of Mathematics\\Zhejiang University, Hangzhou 310027, Zhejiang, China\\
}
\date{}
\allowdisplaybreaks
\begin{document}

\maketitle

\begin{abstract}
In this paper, we consider the linear stability of the elliptic relative equilibria of the restricted 4-body problems where the three primaries form a Lagrangian triangle.
By reduction, the linearized Poincar\'e map is decomposed to the essential part, the Keplerian part and the elliptic Lagrangian part where the last two parts have been studied in literature.
The linear stability of the essential part depends on the masses parameters
$\alpha$, $\beta$ with $\alpha\geq \beta >0$ and the eccentricity $e\in[0,1)$.
Via $\om$-Maslov index theory and linear differential operator theory, we obtain the full bifurcation diagram of linearly stable and unstable regions with respect to $\alpha$, $\beta$ and $e$. Especially, two linearly stable sub-regions are found.
\end{abstract}



\renewcommand{\theequation}{\thesection.\arabic{equation}}
\renewcommand{\thefigure}{\thesection.\arabic{figure}}

\setcounter{equation}{0}
\setcounter{figure}{0}
\section{Introduction and main results}
\label{sec:1}
Let $q_1,q_2,\ldots,q_n\in \R^2$ be the position vectors of $n$ particles with masses $m_1$, $m_2$, $\ldots$, $m_n>0$ respectively.
By the law of universal
gravitation and Newton’s second law,
the system of equations is
\be  m_i\ddot{q}_i=\frac{\partial U}{\partial q_i}, \qquad {\rm for}\quad i=1, 2, \ldots, n, \lb{1.1}\ee
where $U(q)=U(q_1,q_2,\ldots,q_n)=\sum_{1\leq i<j\leq n}\frac{m_im_j}{|q_i-q_j|}$ is the
potential or force function by using the standard norm $|\cdot|$ of vector in $\R^2$.

Letting $p_i=m_i\dot{q}_i\in\R^2$ for $1\le i\le n$, then \eqref{1.1} is transformed to a Hamiltonian system
\be \dot{p}_i=-\frac{\partial H}{\partial q_i},\,\,\dot{q}_i
= \frac{\partial H}{\partial p_i},\qquad {\rm for}\quad i=1,2,\ldots, n,  \lb{1.2}\ee
with Hamiltonian function
\be H(p,q)=H(p_1,p_2,\ldots,p_n, q_1,q_2,\ldots,q_n)=\sum_{i=1}^n\frac{|p_i|^2}{2m_i}-U(q_1,q_2,\ldots,q_n).  \lb{1.3}\ee

A {\it central configuration} $(q_1,q_2,\ldots,q_n)=(a_1,a_2,\ldots,a_n)$ is a solution  of
\begin{equation}
-\lambda m_iq_i=\frac{\partial U}{\pt q_i}(q_1,q_2,\ldots,q_n),
\end{equation}
for some constant $\lambda$.
A direct computation shows that $\lambda=\frac{U(a)}{2I(a)}>0$,
where $I(a)=\frac{1}{2}\sum m_i|a_i|^2$ is the moment of inertia.
Readers may refer \cite{Moe1} and \cite{Win1} for the properties of central configuration.

It is well known that a planar central configuration of the $n$-body problem gives rise to solutions
where each particle moves on a specific Keplerian orbit while each particle follows a homographic motion.
Following Meyer and Schmidt \cite{MS}, we call these solutions as {\it elliptic relative equilibria}
(ERE for short).
Specially when $e=0$, the Keplerian elliptic
motion becomes circular and then all the bodies move around the center of masses along circular
orbits with the same frequency, which are called {\it relative equilibria} in literature.

In the three-body case, the linearly stability of any ERE was clearly studied recently (c.f.\cite{HLS} and \cite{ZL15ARMA}).
In fact, the stability of ERE depends on the eccentricity $e$ and a mass parameter $\beta$.
For the elliptic Lagrangian solution, the mass parameter is given by $\beta_L \in [0,9]$ by
\begin{equation}\lb{L:bb}
\beta_L=\frac{m_1m_2+m_2m_3+m_3m_1}{(m_1+m_2+m_3)^2}.
\end{equation}
For the elliptic Euler solution, the mass parameter is given by $\beta_E\in [0,7)$ in \cite{ZL15ARMA}.

In \cite{HLS} and \cite{ZL15ARMA},
the authors used Maslov-type index and operator theory
to study the linear stability,
and obtained a full description of the bifurcation diagram.
For the near-collision Euler solutions of 3-body problem,
the linear stability  was studied by Hu and Ou in \cite{HO1}.

To our knowledge,
for the general masses of $n$ bodies,
the elliptic Euler-Moulton solutions is the only case
which has been well studied in \cite{ZL2015CMDA}.
The linear stability
of the elliptic Euler-Moulton solutions depends on $(n-1)$ parameters,
namely the eccentricity $e\in [0,1)$ and the $n-2$ mass parameters
$\bb_1,\bb_2,\ldots,\bb_{n-2}$ which defined by (1.14) in \cite{ZL2015CMDA}.
For some special cases of $n$-body problems,
the linear stability of ERE, which is raised from
an $n$-gon or $(1+n)$-gon central configurations with $n$ equal masses,
was studied by Hu, Long and Ou in \cite{HLO} recently.

For the elliptic relative equilibria of
a general non-collinear central configurations,
even for $n=4$, the stability problem is quite open.
In \cite{Liu}, the first author studied the linear stability of the planar four body problem when the central configuration is a rhombus. It turns out that it is linearly unstable for all possible masses and all eccentricity.
In \cite{Zhou}, the second author studied the linear stability
of ERE of planar 4-body problem with two zero masses.
The most interesting case is when the two small masses tend to the same Lagrangian point $L_4$ (or $L_5$).
There are two cases: the ERE raised from the non-convex central configurations are always linearly unstable;
while for the ERE raised from the convex central configurations,
the linear stability are depends on the parameters.

In this paper, we consider the linear stability of ERE
of the planar restricted 4-body problems.
There are four point masses on the plane, one of which possesses zero mass.
This problems are also referred as the $(3+1)$-body problems.
The zero mass body is supposed to have no gravitational effect
on three primaries.
As a consequence, the central configuration equations of the restricted $4$-body problem can be decomposed to two parts,
one part is the equations of the three-body,
and the other is the action of the three primaries on the zero mass body.
The solution of the first part corresponds to the well-known Lagrange equilateral triangular configurations
and the Euler collinear configurations.
Both configurations exist for all values of the masses.
We will consider the central configurations with the three primaries forming an equilateral
triangle (see Figure \ref{RFBP}).

\begin{figure}[ht]
	\centering
	\includegraphics[height=9.5cm]{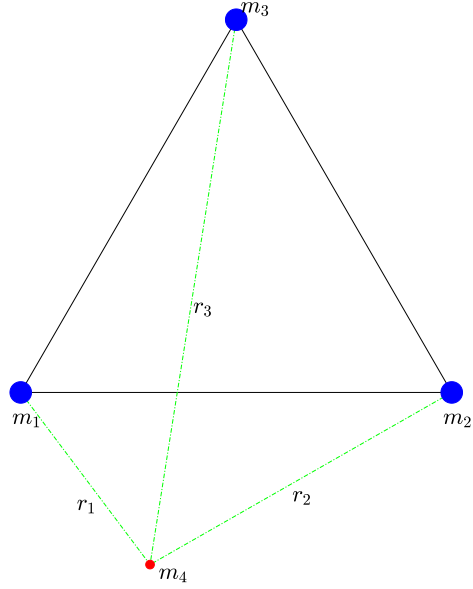}
	\vspace{-10mm}
	\caption{The planar restricted four-body problem.}
	\label{RFBP}
\end{figure}
\vspace{2mm}

In \cite{Leandro}, Leandro studied the central configurations of the restricted $4$-body problems.
Following his notations,
three primaries $m_1$, $m_2$, and $m_3$ form a standard equilateral triangle and
$r_i$s are the distance from $m_i$ to $m_4$ respectively.
Three types of possible regions for $m_4$ were found:
the exterior of the triangle and convex configurations region, the exterior of the triangle and non-convex configurations region, and the interior of the triangle region $\Pi_I$
(see Figure \ref{admissible.region}).
Note that the exterior regions of the triangle are symmetrical by rotations of $\frac{2\pi}{3}$ about the center of the triangle. Therefore, we restrict our study in following regions (shown in the Figure \ref{admissible.region}).
\bea
\Pi_C&=&\left\{(r_1,r_2,r_3)\in(\R^+)^3\;|\;r_1<1,r_2<1,r_3>1\right\}\cap\{F=0\},\\
\Pi_N&=&\left\{(r_1,r_2,r_3)\in(\R^+)^3\;|\;1+2\sqrt{3}x_1<0,1+2\sqrt{3}x_2<0,r_3<1\right\}\cap\{F=0\},\\
\Pi_I&=&\left\{(r_1,r_2,r_3)\in(\R^+)^3\;|\;1+2\sqrt{3}x_k>0,k=1,2,3\right\}\cap\{F=0\},
\eea
where $x_k$s are given by (3.7) of \cite{Leandro} and $F$ is given by (3.2) of \cite{Leandro}.

\begin{figure}[ht]
	\centering
	\includegraphics[height=11.5cm]{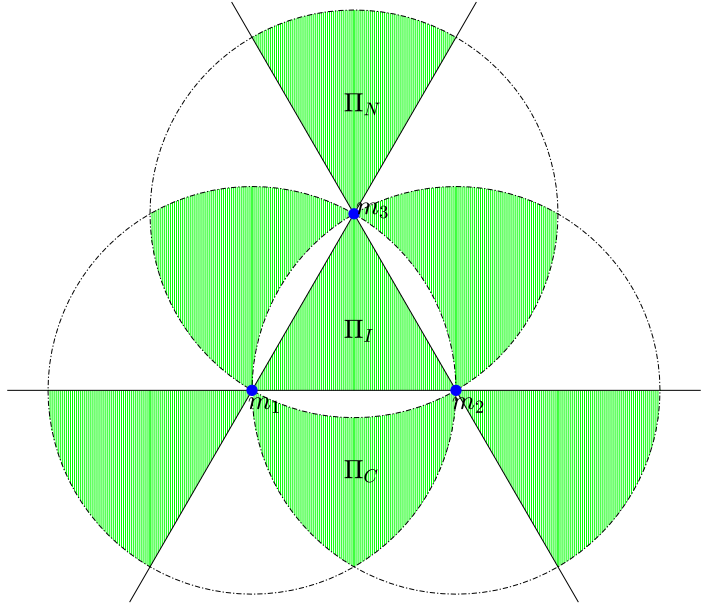}
	\vspace{-15mm}
	\caption{The regions of positive masses.}
	\label{admissible.region}
\end{figure}
\vspace{2mm}

By Proposition 4.2 and Corollary 4.3 of \cite{Leandro},
the author proved that there exists only one convex central configuration
which $m_4$ locates in $\Pi_C$ and only one non-convex central configuration
if $m_4$ locates in $\Pi_N$ respectively for all positive $m_1,m_2$ and $m_3$.

For the linearized Hamiltonian system at the ERE,
we have the following reduction:

\begin{theorem}\label{main.theorem.decomposition}
	For the planar $4$-body problem with given masses $m=(m_1,m_2,m_3,m_4)\in (\R^+)^4$, denote the
	ERE with eccentricity $e\in [0,1)$ for $m$ by $q_{m,e}(t)=(q_1(t), q_2(t), q_3(t),q_4(t))$. Then
	in the limiting case when $m_4$ tends to $0$,
	the linearized Hamiltonian system at $q_{m,e}$ is reduced to the summation of $3$ independent
	Hamiltonian systems, the first one is the linearized system of the Kepler $2$-body problems at the corresponding Kepler orbit,
  the second one is the linearized Hamiltonian system of  $3$-body problems at the Lagrangian ERE with eccentricity $e$ and the mass parameter $\bb_L$ of (\ref{L:bb}),
	and the third one is the essential part $\xi$
  of the linearized Hamiltonian system
	given by
	\begin{equation}\label{LHS}
	\xi'=J\left(\begin{matrix}1&  0&  0& 1\\
		0&  1& -1& 0\\
		0& -1& 1-\frac{\lambda_3}{1+e\cos t}& 0\\
		1&  0&   0& 1-\frac{\lambda_4}{1+e\cos t}
	\end{matrix}\right)\xi,
	\end{equation}
	where $t\in [0,2\pi]$, $\lambda_3 = 1 +\aa+3\bb$ and $\lambda_4 = 1+\aa-3\bb$ with
  \be\label{aa.bb}
  \aa \equiv \frac{1}{2}\sum_{i=1}^3\frac{m_i}{|q_{i,0}-z^*|^3}, \quad
  \bb \equiv \frac{1}{2}\l|\sum_{i=1}^3\frac{m_i(q_{i,0}-z^*)^2}{|q_{i,0}-z^*|^5}\r|,
  \ee
and for $1\leq i\leq 3$, $q_{i,0}$ and $z^*$ are the limit positions of $m_i$ and $m_4$ respectively when $m_4$ tends to zero.
\end{theorem}

\begin{remark}
	According to the discussion in Section 2, the parameters $\lambda_3$ and $\lambda_4$ only depend on
	$m_1$, $m_2$
	since $\sum_{i=1}^3 m_i = 1$.
	But the physical meaning of $\lambda_3$ and $\lambda_4$
	are not straightforward.
\end{remark}

Noting that the linear stability of the Keplerian and Lagrangian parts have been studied in \cite{HLS} and \cite{HS2},
we only need to study the linear stability of the essential part in this paper.
When $m_4$ locates in $\Pi_N$, we have following result.
\begin{theorem}\label{Th:stability.of.Pi_N}
  \begin{enumerate}[label=(\roman*)]
    \item If limit position $z^*$ of the zero mass $m_4$ locates in $\Pi_N$
    such that $q_{1,0}$, $q_{2,0}$, $q_{3,0}$ and $z^*$ form a central configuration,
  	the essential part of the system is linearly unstable.

    \item If $m_1= m_2 = m_4 = 0$, then $z^* = \frac{1}{2} + \ii (1+\frac{\sqrt{3}}{2})\in \pt \Pi_N$,
    the essential part of the system is spectrally stable but linearly unstable.
  \end{enumerate}
\end{theorem}

Denote by $\Sp(2n)$ the symplectic group of real $2n\times 2n$ matrices.
Following \cite{Lon2} and \cite{Lon4}, for any $\omega\in\U=\{z\in\C\;|\;|z|=1\}$, define a real function
$D_\om(M)=(-1)^{n-1}\overline{\om}^n det(M-\om I_{2n})$ for any $M$ in the symplectic group $\Sp(2n)$.
Then we can define $\Sp(2n)_{\om}^0 = \{M\in\Sp(2n)\,|\, D_{\om}(M)=0\}$ and
$\Sp(2n)_{\om}^{\ast} = \Sp(2n)\bs \Sp(2n)_{\om}^0$. The orientation of $\Sp(2n)_{\om}^0$ at any of its point
$M$ is defined to be the positive direction $\frac{\d}{\d t}Me^{t J}|_{t=0}$ of the path $Me^{t J}$ with $t>0$ small
enough. Let $\nu_{\om}(M)=\dim_{\C}\ker_{\C}(M-\om I_{2n})$. Let
$\mathcal{P}_{2\pi}(2n) = \{\ga\in C([0,2\pi],\Sp(2n))\;|\;\ga(0)=I\}$ and
$\ga_0(t)=\diag(2-\frac{t}{2\pi}, (2-\frac{t}{2\pi})^{-1})$ for $0\le t\le 2\pi$.

Given any two $2m_k\times 2m_k$ matrices of square block form
$M_k = (\begin{smallmatrix}
A_k & B_k\\ C_k & D_k
\end{smallmatrix} )$ with $k=1, 2$,
the symplectic sum of $M_1$ and $M_2$ is defined (cf. \cite{Lon2} and \cite{Lon4}) by
the following $2(m_1+m_2)\times 2(m_1+m_2)$ matrix $M_1\dm M_2$:
$$
M_1\dm M_2=\bpm A_1 &   0 & B_1 &   0\\
                 0   & A_2 &   0 & B_2\\
                 C_1 &   0 & D_1 &   0\\
                0   & C_2 &   0 & D_2 \epm.
$$
 For any two paths $\xi_j\in\P_{\tau}(2n_j)$
 with $j=0$ and $1$, let $\xi_0\dm\xi_1(t)= \xi_0(t)\dm\xi_1(t)$ for all $t\in [0,\tau]$.
As in \cite{Lon4}, for $\lm\in\R\bs\{0\}$, $a\in\R$, $\th\in (0,\pi)\cup (\pi,2\pi)$,
$b=(\begin{smallmatrix}b_1 & b_2\\
                 b_3 & b_4\end{smallmatrix})$ with $b_i\in\R$ for $i=1, \ldots, 4$, and $c_j\in\R$
for $j=1, 2$, some normal forms are given by
\bea
&& D(\lm)=\begin{pmatrix}\lm & 0\\
                         0  & \lm^{-1} \end{pmatrix}, \qquad
   R(\th)=\begin{pmatrix}\cos\th & -\sin\th\\
                        \sin\th  & \cos\th\end{pmatrix},  \nn\\
&& N_1(\lm, a)=\begin{pmatrix}\lm & a\\
                             0   & \lm\end{pmatrix}, \qquad
   N_2(e^{\sqrt{-1}\th},b) = \begin{pmatrix}R(\th) & b\\
                                           0      & R(\th)\end{pmatrix},  \nn\\
&& M_2(\lm,c)=\begin{pmatrix}
\lm &   1 &       c_1 &         0 \\
0 & \lm &       c_2 & (-\lm)c_2 \\
0 &   0 &  \lm^{-1} &         0 \\
0 &   0 & -\lm^{-2} &  \lm^{-1} \end{pmatrix}. \lb{eqn:m2}\eea
Here $N_2(e^{\sqrt{-1}\th},b)$ is {\bf trivial} if $(b_2-b_3)\sin\th>0$, or {\bf non-trivial}
if $(b_2-b_3)\sin\th<0$, in the sense of Definition 1.8.11 on p.41 of \cite{Lon4}. Note that
by Theorem 1.5.1 on pp.24-25 and (1.4.7)-(1.4.8) on p.18 of \cite{Lon4}, when $\lm=-1$ there hold
$c_2 \not= 0$ if and only if $\dim\ker(M_2(-1,c)+I)=1$ and
$c_2 = 0$ if and only if $\dim\ker(M_2(-1,c)+I)=2$.
For more details, readers may refer Section 1.4-1.8 of \cite{Lon4}.

For any $\xi\in \mathcal{P}_{2\pi}(2n)$ we define $\nu_\om(\ga)=\nu_\om(\xi(2\pi))$ and
$$  i_\om(\xi)=[\Sp(2n)_\om^0:\xi\ast\xi_n], \qquad {\rm if}\;\;\xi(2\pi)\not\in \Sp(2n)_{\om}^0,  $$
i.e., the usual homotopy intersection number, and the orientation of the joint path $\xi\ast\xi_n$ is
its positive time direction under homotopy with fixed end points, where the path $\xi_n(t) = (\begin{smallmatrix}
2-\frac{t}{\tau} & 0 \\ 0 &  (2-\frac{t}{\tau})^{-1}
\end{smallmatrix})$. When $\xi(2\pi)\in \Sp(2n)_{\om}^0$, the index $i_{\om}(\xi)$  follows \cite{Lon4}. The pair
$(i_{\om}(\xi), \nu_{\om}(\xi)) \in \Z\times \{0,1,\ldots,2n\}$ is called the index function of $\xi$
at $\om$. When $\nu_{\om}(\xi)=0$ or $\nu_{\om}(\xi)>0$, the path $\xi$ is called $\om$-{\it non-degenerate}
or $\om$-{\it degenerate} respectively.
For more details readers may refer to \cite{Lon4}.

When $e = 0$, the region of $(\aa,\bb,e)$ is divided to $\cR_i$ for $1 \leq  i\leq 4$ (see I, II, III and IV of Section \ref{sec:e=0}) and the sub-regions of $\cR_2$ and $\cR_3$ are defined by (\ref{eqn:R30*} -\ref{eqn:R3n+}) when $e= 0$. Then we have following results on $\om$-Maslov index.

\begin{figure}[htbp]
    \centering
    \includegraphics[width=0.618\textwidth]{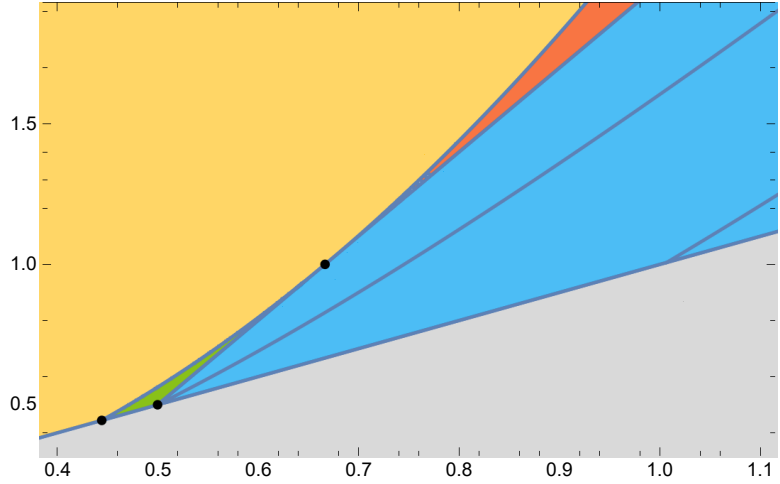}
    \caption{The yellow region is $\cR_1$; the green region is $\cR_2$; the blue region is $\cR_3$; the red region is $\cR_4$; and the gray region is the forbidden region which is $\aa<\bb$.
	The intersection points of the boundaries are presented as the black dots in the figure from left to right respectively $P_1 = (\frac{4}{9},\frac{4}{9})$, $P_2 =(\frac{1}{2},\frac{1}{2})$ and $P_3 =(1,\frac{2}{3})$. The two lines in the blue region $\cR_3$ are $\cR_{3,1}^*$ and $\cR_{3,\frac{3}{2}}^*$ which are the $1$-degenerate curve and $-1$-degenerate curve respectively.}
\label{fig:1}
\end{figure}

\begin{theorem}\lb{thm:gen.ind.nul}
  When $\aa \ge \bb> 0$, the essential part $\xi_{\aa,\bb,e}$ of the fundamental solutions of the linearized Hamiltonian system by (\ref{LHS}) satisfies following results on the $\om$-Maslov-type index and the nullity:
  \begin{enumerate}[label=(\roman*)]
    \item when $e = 0$, the $1$-index and nullity of $\xi_{\aa,\bb,0}$ satisfy that
    \bea
    i_1(\xi_{\aa,\bb,0}) &=&  \begin{cases}
    0,    & \mbox{if} \; (\aa,\bb) \in \cR_1\cup \cR_{2}\cup\cR_{3,0}^* \cup \cR_4;\\
	2n+1  & \mbox{if} \; (\aa,\bb) \in \cR_{3,n}^+\cup \cR^*_{3,n+\frac{1}{2}}\cup \cR_{3,n+1}^-, n \geq 0;\\
    \end{cases}	\lb{index.e=0}\\
    \nu_1(\xi_{\aa,\bb,0}) &=&  \begin{cases}
	3, & \mbox{if} \;  (\aa,\bb) = (\frac{1}{2},\frac{1}{2});\\
	2, & \mbox{if} \; (\aa,\bb) \in \cup_{n=1}^{\infty}\cR_{3,n}^*;\\
	1, & \mbox{if} \;(\aa,\bb) \in \cR_{3,0}^*; \\
	0, & \mbox{if} \; (\aa,\bb)\notin \cup_{n=0}^{\infty}\cR_{3,n}^*;
    \end{cases}\lb{eqn:null.e=0}
    \eea

    \item when $e = 0$, the $-1$-index and nullity of $\xi_{\aa,\bb,0}$ satisfy that
    \bea
    i_{-1}(\xi_{\aa,\bb,0}) &=& \begin{cases}
    0, & \mbox{if} \; (\aa,\bb)\in \cR_1 \cup \cR_{2,\frac{1}{2}}^- \cup\cR_{3,0}^+\cup \cR_4; \\
    2, & \mbox{if} \; (\aa,\bb)\in \cR_{2,\frac{1}{2}}^+; \\
	2n, & \mbox{if} \; (\aa,\bb)\in \cR_{3,n}^- \cup  \cR_{3,n}^* \cup \cR_{3,n}^+, n \geq 1;
    \end{cases} \lb{eqn:intro.-1.index}	\\
    \nu_{-1}(\xi_{\aa,\bb,0}) &=&
    \begin{cases}
	2, & \mbox{if} \;  (\aa,\bb)\in \cR_{2,\frac{1}{2}}^* \bigcup(\cup_{n = 1}^{\infty} \cR_{3,n+\frac{1}{2}}^*);\\
	0,  & \mbox{if} \;  (\aa,\bb)\notin \cR_{2,\frac{1}{2}}^* \bigcup(\cup_{n = 1}^{\infty} \cR_{3,n+\frac{1}{2}}^*); \lb{null.-1.index}
    \end{cases}
    \eea
  \item\lb{thm:hyper.region} when $\aa\geq 3\bb>0$, $e\in [0,1)$ and $\om \in \U$, the $\om$-index and nullity of $\xi_{\aa,\bb,e}$ satisfy
  \bea
    i_{\om}(\xi_{\aa,\bb,e}) = 0, \quad \nu_{\om}(\xi_{\aa,\bb,e}) = 0.
  \eea
    Therefore, $\xi_{\aa,\bb,e}(2\pi)$ is hyperbolic and linearly unstable;

  \item when $\aa> \bb$, $\aa> 3\bb-1$ and $e\in [0,1)$, the $1$-index and nullity of $\xi_{\aa,\bb,e}$ satisfy
  \be i_{1}(\xi_{\aa,\bb,e}) = 0, \quad \nu_{1}(\xi_{\aa,\bb,e}) = 0, \ee
  and when $0< \bb \leq \aa<3\bb-1$ and $e\in[0,1)$,
  $i_{1}(\xi_{\aa,\bb,e}) $ is positive and odd;

  \item for fixed $e\in [0,1)$ and $\om\in\U$, if we fixed $\bb_0\in(0,\infty)$,
  $i_{\om}(\xi_{\aa,\bb_0,e})$ is non-increasing in $\aa\in(0,\infty)$;
  if we fixed $\aa_0\in(0,\infty)$,
  $i_{\om}(\xi_{\aa,\bb_0,e})$ is non-decreasing in $\bb\in(0,\aa_0)$; and $ i_{\om}(\xi_{\aa,\bb_0,e})$ tends to infinity when $3\bb -\aa \to \infty$.
  \end{enumerate}
\end{theorem}

Since by (i) and (ii) of Theorem \ref{thm:gen.ind.nul}, we have the $\pm 1$-index and nullity when $e = 0$.
When $m_4$ locates on $\Pi_C$ or $\Pi_I$, for convenience of the discussion,
two new parameters $\tilde\alpha,\tilde\bb$ are introduced by
\bea
\bpm
\tilde{\aa} \\
\tilde{\bb}
\epm
=T
\bpm
\aa \\
\bb
\epm
=
\bpm
1 & -1\\
-1 & 3
\epm
\bpm
\aa \\
\bb
\epm
+
\bpm
0 \\
-1
\epm.\lb{eqn:trans.T}
\eea
Then $\tilde\alpha,\tilde\bb$ are two functions of $m_1,m_2$ by (\ref{aa.bb}) and (\ref{eqn:trans.T}).
By (iii) of Theorem \ref{thm:gen.ind.nul},
we only need to consider the region of $\td\aa\geq 0$, $\td\bb \geq -1$ and $e\in[0,1)$.
We further divide this region to $\cR_{NH}$ and $\cR_{EH}$ by
\bea
\cR_{NH}  &=&\{(\td\aa,\td\bb,e)| -1< \tilde{\bb} < 0,\tilde{\aa}>0,  e\in[0,1)\};\\
\cR_{EH} &=&\{(\td\aa,\td\bb,e)| \tilde{\bb} > 0, \tilde{\aa}>0, e\in[0,1)\}.
\eea

In $\cR_{NH}$, by (v) of Theorem \ref{thm:gen.ind.nul},
$\nu_{1}(\xi_{\aa,\bb,e}) = 0$ and there exist two $-1$-degeneracy surfaces $\B_s^*$ and $\B_m^*$ and the envelope of the $\om$-degenerate surface $B_k$.
By these surfaces, $\cR_{NH}$ is divided into $\B_s$, $\B_m$, $\B_k$ and $\B_h$ which are defined by (\ref{eqn:RE.Bm} - \ref{eqn:RE.-1.nul}).
For $(\td\aa,\td\bb,e) \in \cR_{NH}$, the norm form and the linear stability are given in Theorem \ref{thm:RE.norm.form} and Theorem \ref{thm:RE.lim}.

\begin{theorem}\lb{thm:RE.norm.form}
    The region $\cR_{NH}$ is divided into sub-regions $\B_s$, $\B_m$, $\B_k$ and $\B_h$ where $\B_h$ is the hyperbolic regions. When $(\td\aa,\td\bb,e) \in \B_s\cup \B_m \cup \B_k$,
	the normal form and linear stability of $\td\xi_{\td\aa,\td\bb,e}(2\pi)$
	satisfy following results.
	\begin{enumerate}[label=(\roman*)]
		\item if $(\td\aa,\td\bb,e) \in \B_m$, $\td\xi_{\td\aa,\td\bb,e}(2\pi) \approx R(\th_1)\diamond R(\th_2)$ for some $\th_1$ and $\th_2\in (\pi, 2\pi)$. Thus it is strongly linear stable;

		\item if $(\td\aa,\td\bb,e)\in \B_s$, $\xi_{\aa,\bb,e}(2\pi) \approx D(\lm)\diamond R(\th)$ for some $0>\lm\neq -1$ and $\th\in (\pi, 2\pi)$. Thus and it is elliptic–hyperbolic, and thus linearly unstable;

		\item if $(\td\aa,\td\bb,e)\in \B_k$. $\xi_{\aa,\bb,e}(2\pi) \approx R(\th_1)\diamond R(\th_2)$ for some $\th_1\in(0,\pi)$ and $\th_2\in (\pi, 2\pi)$ with $2\pi-\th_2 < \th_1$. Thus it is strongly linearly stable;
	\end{enumerate}
\end{theorem}

\begin{theorem}\lb{thm:RE.lim}
	When $(\td\aa,\td\bb,e) \in \B_s^*$, $\B_m^*$ or $\B_k^*$, the normal form and the linear stability of $\td\xi_{\td\aa,\td\bb,e}(2\pi)$ satisfy followings.
	\begin{enumerate}[label=(\roman*)]
		\item If $\td\bb_s(\td\aa,e) < \td\bb_m(\td\aa,e)$, we have $\td\xi_{\td\aa,\td\bb_{m}(\aa,e),e}(2\pi) \approx N_1(-1,-1) \diamond R(\th)$ for some $\th\in (\pi,2\pi)$. Thus it is spectrally stable and linearly unstable;

		\item if $\td\bb_k(\td\aa,e)<\td\bb_s(\td\aa,e) =\bb_m(\aa,e)$, we have $\td\xi_{\td\aa,\td\bb_{s}(\td\aa,e),e}(2\pi) \approx -I_2 \diamond R(\th)$ for some $\th\in (\pi,2\pi)$. Thus it is linearly stable, but not strongly linearly stable;

		\item if $\td\bb_k(\td\aa,e)<\td\bb_s(\td\aa,e) < \td\bb_m(\td\aa,e)$, we have $\td\xi_{\td\aa,\td\bb_{s}(\td\aa,e),e}(2\pi) \approx N_1(-1,-1)  \diamond R(\th)$ for some $\th\in (\pi,2\pi)$. Thus it is spectrally stable and linearly unstable;

		\item if $\td\bb_k(\td\aa,e)<\td\bb_s(\td\aa,e) \leq \td\bb_m(\td\aa,e)$, we have $\td\xi_{\td\aa,\td\bb_{k}(\aa,e),e}(2\pi) \approx N_2(e^{\sqrt{-1}\th},b)$ for some $\th\in (0,\pi)$ and $b= \begin{pmatrix}
		b_1 & b_2 \\ b_3 & b_4 \end{pmatrix}$
		satisfying $(b_2-b_3)\sin\th> 0$,
		that is, $N_2(e^{\sqrt{-1}\th},b)$ is
		trivial in the sense of Definition 1.8.11 in p.41 of \cite{Lon1}.
		Consequently the matrix $\td\xi_{\td\aa,\td\bb_{k}(\td\aa,e),e}(2\pi)$ is spectrally stable and linearly unstable;

		\item if $\td\bb_k(\td\aa,e) = \td\bb_s(\td\aa,e) \leq \td\bb_m(\td\aa,e)$,
		we have either $\td\xi_{\aa,\td\bb_{k}(\td\aa,e),e}(2\pi)  \approx N_1(-1,1)  \diamond D(\lm)$ for some $-1< \lm <0$
		and is linearly unstable; or $\td\xi_{\td\aa,\td\bb_{k}(\td\aa,e),e}(2\pi) \approx M_2(-1,c)$ with $c_1$, $c_2\in \R$ and $c_2 \neq 0$. Thus it is spectrally stable and linearly unstable;

		\item if $\td\bb_k(\td\aa,e) = \td\bb_s(\td\aa,e) = \td\bb_m(\td\aa,e)$, either $\td\xi_{\td\aa,\td\bb_{k}(\td\aa,e),e}(2\pi) \approx M_2(-1,c)$
		with $c_1$, $c_2\in \R$ and $c_2 = 0$ which possesses basic normal form $N_1(-1,1) \diamond N_1(-1, 1)$,
		or
		$\td\xi_{\td\aa,\td\bb_{k}(\td\aa,e),e} \approx N_1(-1, 1)\diamond N_1(-1, 1)$.
		Thus $\td\xi_{\td\aa,\td\bb_{k}(\td\aa,e),e}(2\pi)$ is spectrally stable and linearly unstable.
	\end{enumerate}
\end{theorem}

When $(\td\aa_0,\td\bb_0,e)\in\cR_{EH}$, we find there exist infinitely many $1$-degenerate and $-1$-degenerate surfaces and these surfaces separate $\cR_{EH}$ into sub-regions.

\begin{theorem}\label{thm:REH.norm}
    For the given $(\td\aa_0,\td\bb_0,e)\in\cR_{EH}$,
    there exist $1$-degenerate surface functions $(\td\aa,e) \to \td\bb_i(\td\aa,1, e)$ and $-1$-degenerate surface functions $(\td\aa,e) \to \td\bb_i(\td\aa,-1, e)$ with $\td\bb_{2n}(\td\aa,1, e) = \td\bb_{2n-1}(\td\aa,1, e)$. If $1$-degenerate surfaces $\Ga_n$ and $-1$-degenerate surfaces $\Sg_n^{\pm}$ are defined by
    \bea
    \Ga_n &=&\{(\td\aa,\td\bb_{2n}(\td\aa,1, e),e)|e\in[0,1), \td\aa\geq 0\}, \quad \forall n\in\N_0,  \\
    \Sg_n^- &=&\{(\td\aa,\td\bb_{2n-1}(\td\aa,-1, e),e)|e\in[0,1), \td\aa\geq 0\}, \quad \forall n\in\N,\\
    \Sg_n^+ &=&\{(\td\aa,\td\bb_{2n+1}(\td\aa,-1, e),e)|e\in[0,1), \td\aa\geq 0\}, \quad \forall n\in\N,
    \eea
    then we have that
    \begin{enumerate}[label=(\roman*)]
      \item $\Ga_n$ starting from curve $\cR_{3,n}^*$ when $e = 0$, $\Ga_n$ is perpendicular to the $\aa\bb$-plane and for each $e\in [0,1)$,
      $\nu_1(\xi_{\td\aa,\td\bb_{2n}(\td\aa,1, e),e}) = 2$.
      Furthermore, $\td\bb_{2n}(\td\aa,1, e)$ is the real analytic functions in $(\td\aa,e)$;

      \item starting from the line $\cR_{3,n+\frac{1}{2}}^*$ defined in (\ref{eqn:R3n+1/2*}) for $n\in\N$,
      two $-1$-degenerate surfaces $\Sg_n^{\pm}$ of $\xi_{\aa,\bb,e}(2\pi)$ are perpendicular
      to the $\aa\bb$-plane.
      Moreover, for each $e\in (0,1)$,
      if $\td\bb_{2n-1}(\td\aa,-1,e)\ne\bb_{2n}(\td\aa,-1,e)$ with $(\td\aa,\td\bb_{2n-1}(\td\aa,-1,e),e)\in\Sg_n^{-}$ and $(\td\aa,\td\bb_{2n}(\td\aa,-1,e),e)\in\Sg_n^{+}$,
      the two surfaces satisfy $\nu_1(\xi_{\td\aa,\td\bb_{2n-1}(\td\aa,-1,e),e})=\nu_1(\xi_{\td\aa,\td\bb_{2n}(\td\aa,-1,e),e})=1$;
      if $\td\bb_{2n-1}(\td\aa,-1,e)=\td\bb_{2n}(\td\aa,-1,e)\in\Sg_n^{+}\cap\Sg_n^{-}$,
      the two surfaces satisfy $\nu_1(\xi_{\td\aa, \td\bb_{2n-1}(\td\aa,-1,e),e})=2$.
      Furthermore, both $\bb_{2n-1}(-1,e)$ and
      $\bb_{2n}(-1,e)$ are real piece-wise analytic functions in $e\in [0,1)$;

      \item the $1$-degenerate surfaces and $-1$-degenerate surfaces can be ordered from left to right by
      \be
      \Ga_0, \Sigma_1^-, \Sigma_1^+, \Ga_1,\Sigma_2^-, \Sigma_2^+, \Ga_2,\dots, \Sigma_n^-, \Sigma_n^+,  \Ga_n,\dots. \lb{eqn:order}
      \ee

      Moreover, for $n_1,n_2\in\N$, $\Ga_{n_1}$ and $\Sg^{\pm}_{n_2}$ cannot intersect each other;
      if $n_1\ne n_2$, $\Ga_{n_1}$ and $\Ga_{n_2}$ cannot intersect each other,
      and $\Sg^{\pm}_{n_1}$ and $\Sg^{\pm}_{n_2}$ cannot intersect each other.
      More precisely, for each fixed $\td\aa \geq 0$ and $e\in [0,1)$, we have
      \begin{eqnarray}
      &&0 <\td\bb_1(\td\aa,-1,e)\le \td\bb_2(\td\aa,-1,e) < \td\bb_1(\td\aa,1,e) = \td\bb_2(\td\aa,1,e) < \td\bb_3(\td\aa,-1,e)<\dots \nn\\
      && \qquad \qquad <\td\bb_{2n-1}(\td\aa,-1,e)\le \td\bb_{2n}(\td\aa,-1,e)<\td\bb_{2n-1}(\td\aa,1,e)=\td\bb_{2n}(1,e)<\dots .
      \end{eqnarray}
    \end{enumerate}
\end{theorem}

Then the linear stability and the normal form of the $\xi_{\aa,\bb,e}(2\pi)$ are obtained when $(\td\aa,\td\bb,e) \in \cR_{EH}$.
 \begin{theorem}\lb{Th:stability.of.EH}
 For the given $\aa_0,\bb_0, e_0 \in \cR_{EH}$, the linear stability and the normal form $\td\xi_{\td\aa_0,\td\bb_0,e_0}(2\pi)$ satisfy following results: for $n\in\N_0$,
    \begin{enumerate}[label=(\roman*)]
        \item if $\td\bb_{2n}(\td\aa_0,1,e_0)<\td\bb_0<\td\bb_{2n+1}(\td\aa_0,-1,e_0)$,
        then $i_1(\td\xi_{\td\aa_0,\td\bb_0,e_0}(2\pi))=2n+1$, $\nu_1(\td\xi_{\td\aa_0,\td\bb_0,e_0}(2\pi))=0$, and
        $i_{-1}(\td\xi_{\td\aa_0, \td\bb_0,e_0}(2\pi))=2n$, $\nu_{-1}(\td\xi_{\td\aa_0, \td\bb_0,e_0}(2\pi))=0$. Therefore,
        $\td\xi_{\td\aa_0,\td\bb_0,e_0}(2\pi)\approx R(\theta)\diamond D(2)$ for some $\theta\in(0,\pi)$;

        \item if $\td\bb_0=\td\bb_{2n+1}(\td\aa_0,-1,e_0)=\td\bb_{2n+2}(\td\aa_0,-1,e_0)$,
        then  $i_1(\td\xi_{\td\aa_0, \td\bb_0,e_0}(2\pi))=2n+1$, $\nu_1(\td\xi_{\td\aa_0,\td\bb_0,e_0}(2\pi))=0$, and
        $i_{-1}(\td\xi_{\td\aa_0,\td\bb_0,e_0}(2\pi))=2n$, $\nu_{-1}(\td\xi_{\td\aa_0, \td\bb_0,e_0}(2\pi))=2$. Therefore,
        $\td\xi_{\td\aa_0, \td\bb_0,e_0}(2\pi)\approx -I_2\diamond D(2)$;

        \item if $\td\bb_{2n+1}(\td\aa_0,-1,e_0)\ne \td\bb_{2n+2}(\td\aa_0,-1,e_0)$ and $\td\bb_0=\td\bb_{2n+1}(\td\aa_0,-1,e_0)$,
        then  $i_1(\td\xi_{\td\aa_0,\td\bb_0,e_0}(2\pi))=2n+1$, $\nu_1(\td\xi_{\td\aa_0,\td\bb_0,e_0}(2\pi))=0$,
        and $i_{-1}(\td\xi_{\td\aa_0,\td\bb_0,e_0}(2\pi))=2n$, $\nu_{-1}(\td\xi_{\td\aa_0,\td\bb_0,e_0}(2\pi))=1$.
        Therefore, $\td\xi_{\td\aa_0,\td\bb_0,e_0}(2\pi)\approx N_1(-1,-1)\diamond D(2)$;

        \item if $\td\bb_{2n+1}(\td\aa_0,-1,e_0)\ne \td\bb_{2n+2}(\td\aa_0,-1,e_0)$ and $\td\bb_{2n+1}(\td\aa_0,-1,e_0)<\td\bb_0<\td\bb_{2n+2}(\aa_0,-1,e_0)$,
        then  $i_1(\td\xi_{\aa_0,\td\bb_0,e_0}(2\pi))=2n+1$, $\nu_1(\td\xi_{\td\aa_0,\td\bb_0,e_0}(2\pi))=0$,
        and $i_{-1}(\td\xi_{\td\aa_0,\td\bb_0,e_0}(2\pi))=2n+1$, $\nu_{-1}(\td\xi_{\td\aa_0,\td\bb_0,e_0}(2\pi))=0$.
        Therefore, $\td\xi_{\td\aa_0,\td\bb_0,e_0}(2\pi)\approx D(-2)\diamond D(2)$;

        \item if $\td\bb_{2n+1}(\td\aa_0,-1,e_0)\ne \td\bb_{2n+2}(\td\aa_0,-1,e_0)$ and $\td\bb_0=\td\bb_{2n+2}(\td\aa_0,-1,e_0)$,
        then  $i_1(\td\xi_{\td\aa_0,\td\bb_0,e_0}(2\pi))=2n+1$, $\nu_1(\td\xi_{\td\aa_0,\td\bb_0,e_0}(2\pi))=0$, and
        $i_{-1}(\td\xi_{\td\aa_0,\td\bb_0,e_0}(2\pi))=2n+1$,
        $\nu_{-1}(\td\xi_{\td\aa_0,\td\bb_0,e_0}(2\pi))=1$.
        Therefore, $\td\xi_{\td\aa_0,\td\bb_0,e_0}(2\pi)\approx N_1(-1,1)\diamond D(2)$;

        \item if $\td\bb_{2n+2}(\td\aa_0,-1,e_0)<\td\bb_0<\td\bb_{2n+1}(\td\aa_0,1,e_0)$,
        then  $i_1(\td\xi_{\td\aa_0,\td\bb_0,e}(2\pi))=2n+1$, $\nu_1(\td\xi_{\td\aa_0,\td\bb_0,e_0}(2\pi))=0$, and
        $i_{-1}(\td\xi_{\td\aa_0,\td\bb_0,e_0}(2\pi))=2n+2$, $\nu_{-1}(\td\xi_{\td\aa_0,\td\bb_0,e_0}(2\pi))=0$.
        Therefore, $\td\xi_{\td\aa_0,\td\bb_0,e_0}(2\pi)\approx R(\theta)\diamond D(2)$
        for some $\theta\in(\pi,2\pi)$;

        \item if $\td\bb_0=\td\bb_{2n+1}(\td\aa_0,1,e_0)(=\td\bb_{2n+2}(\td\aa_0,1,e_0))$,
        then  $i_1(\td\xi_{\td\aa_0, \td\bb_0,e_0}(2\pi))=2n+1$, $\nu_1(\td\xi_{\td\aa_0, \td\bb_0,e_0}(2\pi))=2$,
        and $i_{-1}(\td\xi_{\td\aa_0, \td\bb_0,e_0}(2\pi))=2n+2$,
        $\nu_{-1}(\td\xi_{\td\aa_0, \td\bb_0,e_0}(2\pi))=0$.
        Therefore, $\td\xi_{\td\aa_0, \td\bb_0,e_0}(2\pi)\approx I_2\diamond D(2)$.
    \end{enumerate}
\end{theorem}

This paper is organized as follows.
In Section 2, We reduce the linearized Hamiltonian systems at ERE
for the general $4$-body problem,
and for the case $m_4=0$,
we will prove Theorem \ref{main.theorem.decomposition}.
In Section 3, we analyze the linear stability of circular ERE namely $e = 0$.
In Section 4, we study some general properties of the $\om$-indices. Then in Section 5, 6 and 7  we will discuss the linear stability in
the hyperbolic region which satisfies $\aa \geq 3\bb >0$, non-hyperbolic region $\cR_{NH}$,
and the elliptic-hyperbolic region $\cR_{EH}$ respectively.
Especially, Theorem \ref{Th:stability.of.Pi_N} is proved in Section 7.
In Section 8, we apply the results to the cases of $m_1=m_2$ by the assistance of the numerical computation. In this paper, we use $\N$ to denote the positive integers and use $\N_0$ to denote the non-negative integers.

\setcounter{equation}{0}
\setcounter{figure}{0}

\section{The Symplectic Reduction of the Linearized Hamiltonian Systems}\label{sec:2}

\subsection{Preliminaries of $\om$-Maslov-type indices and $\om$-Morse indices}
For $T>0$, suppose $x$ is a critical point of the functional
$$ F(x)=\int_0^TL(t,x,\dot{x})dt,  \qquad \forall\,\, x\in W^{1,2}(\R/T\Z,\R^n), $$
where $L\in C^2((\R/T\Z)\times \R^{2n},\R)$ and satisfies the
Legendrian convexity condition $L_{p,p}(t,x,p)>0$. It is well known
that $x$ satisfies the corresponding Euler-Lagrangian
equation:
\bea
&& \frac{\d}{\d t}L_p(t,x,\dot{x})-L_x(t,x,\dot{x})=0,    \label{A2.7}\\
&& x(0)=x(T),  \qquad \dot{x}(0)=\dot{x}(T).    \label{A2.8}\eea
For such an extremal loop, define
\bea
P(t) = L_{p,p}(t,x(t),\dot{x}(t)),\quad
Q(t) = L_{x,p}(t,x(t),\dot{x}(t)),\quad
R(t) = L_{x,x}(t,x(t),\dot{x}(t)).  \nn\eea
Note that
\be F\,''(x)=-\frac{\d}{\d t}(P\frac{\d}{\d t}+Q)+Q^T\frac{\d}{\d t}+R. \nn\ee

For $\omega\in\U$, set
\be  D(\omega,T)=\{y\in W^{1,2}([0,T],\C^n)\,|\, y(T)=\omega y(0) \}.   \lb{A2.10}\ee
We define the $\omega$-Morse index $\phi_\omega(x)$ of $x$ to be the dimension of the
largest negative definite subspace of
$$ \langle F\,''(x)y_1,y_2 \rangle, \qquad \forall\;y_1,y_2\in D(\omega,T), $$
where $\langle\cdot,\cdot\rangle$ is the inner product in $L^2$. For $\omega\in\U$, we
also set
\be  \ol{D}(\omega,T)= \{y\in W^{2,2}([0,T],\C^n)\,|\, y(T)=\omega y(0), \dot{y}(T)=\om\dot{y}(0) \}.
                     \lb{A2.11}\ee
Then $F''(x)$ is a self-adjoint operator on $L^2([0,T],\R^n)$ with domain $\ol{D}(\omega,T)$.
We also define the $\om$-nullity $\nu_{\om}(x)$ of $x$ by
$$ \nu_{\om}(x)=\dim\ker(F''(x)).  $$
Note that we only use $n=2$ in (\ref{A2.11}) from in this paper.

In general, for a self-adjoint linear operator $A$ on the Hilbert space $\mathscr{H}$,
we set $\nu(A)=\dim\ker(A)$ and denote by $\phi(A)$ its Morse index which is the maximum dimension
of the negative definite subspace of the symmetric form $\langle A\cdot,\cdot\rangle$. Note
that the Morse index of $A$  is equal to the total multiplicity of the negative eigenvalues
of $A$.

On the other hand, $\td{x}(t)=(\partial L/\partial\dot{x}(t),x(t))^T$ is the solution of the
corresponding Hamiltonian system of (\ref{A2.7})-(\ref{A2.8}), and its fundamental solution
$\gamma(t)$ is given by
\bea \dot{\gamma}(t) &=& JB(t)\gamma(t), \nn\\
     \gamma(0) &=& I_{2n},  \nn\eea
with
\be B(t)=\bpm P^{-1}(t)& -P^{-1}(t)Q(t)\\
                       -Q(t)^TP^{-1}(t)& Q(t)^TP^{-1}(t)Q(t)-R(t) \epm. \nn\ee

\begin{lemma} \lb{lem:2.1}(\cite{Lon4}, p.172)\lb{L2.3}
For the $\omega$-Morse index $\phi_\omega(x)$ and nullity $\nu_\omega(x)$ of the solution $x=x(t)$
and the $\omega$-Maslov-type index $i_\omega(\gamma)$ and nullity $\nu_\omega(\gamma)$ of the
symplectic path $\ga$ corresponding to $x$, for any $\omega\in\U$ we have
\be \phi_\omega(x) = i_\omega(\gamma), \qquad \nu_\omega(x) = \nu_\omega(\gamma).  \lb{A2.15}\ee
\end{lemma}

A generalization of the above lemma to arbitrary  boundary conditions is given in \cite{HS1}.
For more information on these topics, we refer to \cite{Lon4}.
To measure the jumps between $i_\om(\ga)$
and $i_\lambda(\ga)$ with $\lambda\in\U$ near $\om$ from two sides of $\om$ in $\U$, the splitting numbers $S_M^{\pm}(\om)$ is defined by followings.
\begin{definition} (\cite{Lon2}, \cite{Lon4})\lb{D2.3}
For any $M\in\Sp(2n)$ and $\om\in\U$, choosing $\tau>0$ and $\ga\in\P_\tau(2n)$ with $\ga(\tau)=M$,
we define
\be
S_M^{\pm}(\om)=\lim_{\epsilon\rightarrow0^+}\;i_{\exp(\pm\epsilon\sqrt{-1})\om}(\ga)-i_\om(\ga).
\ee
They are called the splitting numbers of $M$ at $\om$.
\end{definition}

For any $\om_0=e^{\sqrt{-1}\th_0}\in\U$ with $0\le\th_0<2\pi$,
the eigenvalues of $M$ on $\U$ are denote by $\om_j$ with $1\le j\le p_0$ which are
distributed counterclockwise from $1$ to $\om_0$ and located strictly between $1$ and $\om_0$.
Then we have
\be i_{\om_0}(\ga)=i_1(\ga)+S_M^+(1)+\sum_{j=1}^{p_0}(-S_M^-(\om_j)+S_M^+(\om_j))-S_M^-(\om_0).  \lb{eqn:split}\ee
The splitting numbers have following properties.
\begin{lemma} (\cite{Lon4}, p.198)
The integer valued splitting number pair $(S_M^+(\om),S_M^-(\om))$ defined for all
$(\om,M)\in\U\times\cup_{n\ge1}\Sp(2n)$ are uniquely determined by the following axioms:

$1^{\circ}$ (Homotopy invariant) $S_M^{\pm}(\om)=S_N^{\pm}(\om)$ for all $N\in\Omega^0(M)$.

$2^{\circ}$ (Symplectic additivity) $S_{M_1\diamond M_2}^{\pm}(\om)=S_{M_1}^{\pm}(\om)+S_{M_2}^{\pm}(\om)$
for all $M_i\in\Sp(2n_i)$ with $i=1$ and $2$.

$3^{\circ}$ (Vanishing) $S_M^{\pm}(\om)=0$ if $\om\not\in\sigma(M)$.

$4^{\circ}$ (Normality) $(S_M^+(\om),S_M^-(\om))$ coincides with the ultimate type of
$\om$ for $M$ when $M$ is any basic normal form.
\end{lemma}

Moreover, by Lemma 9.1.6 on p.192 of \cite{Lon4} for $\om\in\C$ and $M\in\Sp(2n)$, we have
\be  S_M^+(\om)=S_M^-(\overline\om).  \ee

The ultimate type of $\om\in\U$ for a symplectic matrix $M$ mentioned in the above lemma
is given in Definition 1.8.12 on pp.41-42 of \cite{Lon4} algebraically with its more properties studied
there.

For the reader's conveniences, following the List 9.1.12 on pp.198-199 of \cite{Lon4},
the splitting numbers (i.e., the ultimate types) for all basic normal forms are given by:

$\langle$1$\rangle$ $(S_M^+(1),S_M^-(1))=(1,1)$ for $M=N_1(1,b)$ with $b=1$ or $0$.

$\langle$2$\rangle$ $(S_M^+(1),S_M^-(1))=(0,0)$ for $M=N_1(1,-1)$.

$\langle$3$\rangle$ $(S_M^+(-1),S_M^-(-1))=(1,1)$ for $M=N_1(-1,b)$ with $b=-1$ or $0$.

$\langle$4$\rangle$ $(S_M^+(-1),S_M^-(-1))=(0,0)$ for $M=N_1(-1,1)$.

$\langle$5$\rangle$ $(S_M^+(e^{\sqrt{-1}\th}),S_M^-(e^{\sqrt{-1}\th}))=(0,1)$
for $M=R(\theta)$ with $\th\in(0,\pi)\cup(\pi,2\pi)$.

$\langle$6$\rangle$ $(S_M^+(\om),S_M^-(\om))=(1,1)$ for $M=N_2(\om,b)$
being non-trivial (cf. Definition 1.8.11 on p.41 of \cite{Lon4}) with
$\om=e^{\sqrt{-1}\th}\in\U\backslash\R$.

$\langle$7$\rangle$ $(S_M^+(\om),S_M^-(\om))=(0,0)$ for $M=N_2(\om,b)$
being trivial (cf. Definition 1.8.11 on p.41 of \cite{Lon4}) with
$\om=e^{\sqrt{-1}\th}\in\U\backslash\R$.

$\langle$8$\rangle$ $(S_M^+(\om),S_M^-(\om))=(0,0)$ for $\om\in\U$ and $M\in\Sp(2n)$
satisfying $\sigma(M)\cap\U=\emptyset$.

For any symplectic path $\ga \in \mathcal{P}_{2\pi}(2n)$ and $m \in \N$, the  $m$-th iteration $\ga_m:[0,\tau] \longrightarrow \Sp(2n)$ is by
\bea
\ga^m: [0,m\tau] \longrightarrow \Sp(2n)
\eea
with $\ga^m(t) = \ga(t-j\tau)\ga(\tau)^j$ for $g\tau\le t \leq (j+1)\tau$ and $j\in\{0,1,\dots, m-1\}$
The next Bott-type iteration formula is will be used in this paper.
\begin{lemma}[(See [19, Theorem 9.2.1, p. 199].)]
For any $z\in \U$,
\bea
i_z(\ga^m) = \sum_{\om^m=z}i_{\om}(\ga)
\eea
\end{lemma}

\subsection{Two useful maps}
In this subsection, we introduce two useful maps for our later discussion.
We define $\varphi,\psi:\C\rightarrow GL(2,\R)$, $z=x+\sqrt{-1}y$ by
\bea
\varphi(z)
\bpm
x & -y\\
y & x
\epm, \quad
\psi(z)=
\bpm
x & y \\
y & -x
\epm,   \label{map}
\eea
where $z= x+\sqrt{-1}y$ with $x,y\in\R$.
Thus both $\varphi$ and $\psi$ are real linear maps.
Moreover, by direct computations, we have the following properties:
\begin{lemma}
	(i) If $z\in\R$, then
	\be
	\varphi(z)=zI_2,\quad\quad
	\psi(z)=z\bpm
	1 & 0 \\
	0 & -1\epm;
	\ee
	(ii) for any $z\in\C$, we have
	\bea
	\varphi(z)^T=\varphi(\bar{z}),
	\quad
	\psi(z)^T=\psi(z);
	\eea
	(iii) for any $z,w\in\C$, we have
	\bea
	\varphi(z)\varphi(w)&=&\varphi(zw),
	\\
	\psi(z)\psi(w)&=&\varphi(z\bar{w}),
	\\
	\varphi(z)\psi(w)&=&\psi(zw),
	\\
	\psi(z)\varphi(w)&=&\psi(z\bar{w}).
	\eea
	Specially, we have
	\bea
	\varphi(\bar{z})\varphi(z)=\varphi(z)\varphi(\bar{z})=\varphi(|z|^2)=|z|^2I_2,
	\\
	\psi(z)\psi(z)=\psi(\bar{z})\psi(\bar{z})=\varphi(|z|^2)=|z|^2I_2.
	\eea
\end{lemma}

\begin{remark} \label{map.matrix}
	For a complex matrix $N=(N_{ij})_{m\times n}$, the $2m\times 2n$ matrix  $\varphi(N)$ is given by
	\be
	\varphi(N)=
	\bpm
	\varphi(N_{11})& \varphi(N_{12})& \ldots& \varphi(N_{1n})\\
	\varphi(N_{21})& \varphi(N_{22})& \ldots& \varphi(N_{2n})\\
	\ldots&          \ldots&          \ldots& \ldots\\
	\varphi(N_{m1})& \varphi(N_{m2})& \ldots& \varphi(N_{mn})
	\epm.
	\ee
\end{remark}

\subsection{The essential part of the fundamental solution}

In \cite{MS} (cf. p.275), Meyer and Schmidt gave the essential part of the fundamental solution of the
elliptic Lagrangian orbit. Readers may also refer \cite{Lon5} for details.
Based on their method, we reduce the linearized system of the planar restricted 4-body problem to 3 sub-systems.

Suppose the four particles which form one central configuration are in $\R^2$ at $a_1=(a_{1x},a_{1y}),a_2=(a_{2x},a_{2y}),a_3=(a_{3x},a_{3y}),a_4=(a_{4x},a_{4y})$. By identifying $\R^2$ with $\C$, we write $a_i$s as
\be
z_{a_i}=a_{ix}+\sqrt{-1}a_{iy},\quad i=1,2,3,4.  \label{complex.rep}
\ee
Without lose of generality, we normalize the four masses by
\begin{equation}\label{nomorlize.the.masses}
\sum_{i=1}^n m_i=1,
\end{equation}
fix the center of mass at origion and normalize the positions $a_i$s by
\bea
\sum_{i=1}^4 m_ia_i=0, \quad \sum_{i=1}^4 m_i|a_i|^2=2I(a)=1.\label{inertia}
\eea
Using the notations in (\ref{complex.rep}), (\ref{inertia}) are equivalent to
\bea
\sum_{i=1}^4 m_iz_{a_i}=0, \quad
\sum_{i=1}^4 m_i|z_{a_i}|^2=2I(a)=1.\label{inertia'}
\eea
Moreover, let
\begin{equation}\label{mu}
\mu=U(a)=\sum_{1\le i<j\le 4}\frac{m_im_j}{|a_i-a_j|}=\sum_{1\le i<j\le 4}\frac{m_im_j}{|z_{a_i}-z_{a_j}|},
\quad
\sigma=(\mu p)^{1/4},
\end{equation}
where $p$ is given by \eqref{eqn.r.p} and
\be
\tilde{M}=\diag(m_1,m_2,m_3,m_4),\quad M=\diag(m_1,m_1,m_2,m_2,m_3,m_3,m_4,m_4).\lb{eqn:M}
\ee
Because $a_1,a_2,a_3,a_4$ form a central configuration, following equation holds.
\be\label{eq.of.cc}
\sum_{j=1,j\ne i}^4\frac{m_j(z_{a_{j}}-z_{a_{i}})}{|z_{a_{i}}-z_{a_{j}}|^3}=
\frac{U(a)}{2I(a)}z_{a_{i}}=\mu z_{a_{i}}.
\ee
Let $B$ be a $4\times 4$ symmetric matrix such that
\begin{equation}
B_{ij}=
\left\{\begin{array}{c}
\frac{m_im_j}{|z_{a_i}-z_{a_j}|^3}, \quad \mbox{if}\;i\ne j,1\le i,j\le 4,\\
-\sum_{j=1,j\ne i}^4\frac{m_im_j}{|z_{a_i}-z_{a_j}|^3}, \quad  \mbox{if}\;i=j,1\le i\le 4,
\end{array}\right. \nn
\end{equation}
and define $D$ and $\tilde{D}$ by
\bea
D&=&\mu I_4+\tilde{M}^{-1}B, \label{D}
\\
\tilde{D}&=&\mu I_4+\tilde{M}^{-1/2}B\tilde{M}^{-1/2}=\tilde{M}^{1/2}D\tilde{M}^{-1/2}, \label{tilde.D}
\eea
where $\mu$ is given by (\ref{mu}).

Firstly, $D$ has two direct eigenvalues: $\lambda_1=\mu$ with $v_1=(1,1,\ldots,1)^T$, and $\lambda_2=0$ with
$v_2=(z_{a_1},z_{a_2},z_{a_3},z_{a_4})^T$.
Namely, by direct computations, we have
\bea
(Dv_1)_i&=&\mu-\sum_{j=1,j\ne i}^4\frac{m_j}{|z_{a_i}-z_{a_j}|^3}+\sum_{j=1,j\ne i}^4\frac{m_j}{|z_{a_i}-z_{a_j}|^3}=\mu,
\\
(Dv_2)_i&=&(\mu-\sum_{j=1,j\ne i}^4\frac{m_j}{|z_{a_i}-z_{a_j}|^3})z_{a_i}+\sum_{j=1,j\ne i}^4\frac{m_jz_{a_j}}{|z_{a_i}-z_{a_j}|^3}
\nonumber
\\
&=&\mu z_{a_i}+\sum_{j=1,j\ne i}^4\frac{m_j(z_{a_j}-z_{a_i})}{|z_{a_i}-z_{a_j}|^3}
\nonumber
\\
&=&0, \lb{eqn:Dv2}
\eea
where the last equality of (\ref{eqn:Dv2}) holds by (\ref{eq.of.cc}).
Moreover by (\ref{nomorlize.the.masses})-(\ref{inertia}), we have
\bea
\overline{v}_1^T\tilde{M}v_1=\sum_{i=1}^n m_i=1,&\quad&
\overline{v}_2^T\tilde{M}v_2=\sum_{i=1}^n m_i|z_{a_i}|^2=1. \label{v1.M.v1}
\\
\overline{v}_1^T\tilde{M}v_2=\sum_{i=1}^4 m_iz_{a_i}=0, &\quad&
\overline{v}_2^T\tilde{M}v_1=\sum_{i=1}^4 m_i\overline{z}_{a_i}=0.\label{v1.M.v2}
\eea
Let $\overline{v}_2=(\overline{z}_{a_1},\overline{z}_{a_2},\overline{z}_{a_3},\overline{z}_{a_4})^T$.
Because, $a_1,a_2,a_3,a_4$ form a non-colinear central configuration,
$\overline{v}_2$ is independent with $v_2$. Moreover, $\overline{v}_2$ is also independent with $v_1$.
So $\overline{v}_2$ is another eigenvector of $D$ corresponding to eigenvalue $\lambda_3=0$.

For $v_3$, suppose
\be
v_3=k\overline{v}_2+lv_2,  \label{v_3}
\ee
where $k\in\R,l\in\C$ are defined by
\bea
k=\frac{1}{\sqrt{1-|\sum_{i=1}^4 m_i\overline{z}_{a_i}^2|^2}}, \quad
l=-\frac{\sum_{i=1}^4 m_i\overline{z}_{a_i}^2}{\sqrt{1-|\sum_{i=1}^4 m_i\overline{z}_{a_i}^2|^2}}. \label{param.l}
\eea
If ${v}_2^T\tilde{M}v_2=\sum_{i=1}^n m_iz_{a_i}^2=0$, we have that $k=1$ and $l=0$, i.e., $v_3=\overline{v}_2$.
Then we have
\bea
\overline{v}_1^T\tilde{M}v_3&=&\sum_{i=1}^4 m_i\overline{z}_{a_i}=0,\label{v1.M.v3}
\\
\overline{v}_2^T\tilde{M}v_3&=&\sum_{i=1}^4 m_i\overline{z}_{a_i}^2=0,\label{v2.M.v3}
\\
\overline{v}_3^T\tilde{M}v_3&=&\sum_{i=1}^n m_i|z_{a_i}|^2=1.\label{v3.M.v3}
\eea
In the other cases, (\ref{v1.M.v3})-(\ref{v3.M.v3}) also hold by
\bea
&&\overline{v}_2^T\tilde{M}v_3=\overline{v}_2^T\tilde{M}(k\overline{v}_2+lv_2)=k\sum_{i=1}^4 m_i\overline{z}_{a_i}^2+l = 0, \nn
\\
&&\overline{v}_3^T\tilde{M}v_3=(kv_2+\overline{l}\overline{v}_2)^T\tilde{M}(k\overline{v}_2+lv_2)
=k^2+|l|^2+kl\sum_{i=1}^4 m_iz_{a_i}^2+k\overline{l}\sum_{i=1}^4 m_i\overline{z}_{a_i}^2  =1. \nn
\eea

Based on $v_1,v_2$ and $v_3$, the unitary matrix $\tilde{A}$ is defined by
\be
\tilde{A}=
\bpm
1\quad z_{a_1}\quad b_1\quad c_1\\
1\quad z_{a_2}\quad b_2\quad c_2\\
1\quad z_{a_3}\quad b_3\quad c_3\\
1\quad z_{a_4}\quad b_4\quad c_4
\epm, \nn
\ee
where $(b_1,b_2,b_3,b_4)=v_3^T$, i.e., $b_i=k\overline{z}_{a_i}+lz_{a_i},1\le i\le4$.
Then $c_i=A_{i4}$, where $A_{i4}$ is the algebraic cofactor of $c_i$.

On the other hand, the signed area of the triangle formed by $a_i,a_j$ and $a_k$ is given by
\be
\Delta_{ijk}=\frac{\sqrt{-1}}{4}\det
\bpm
1\quad z_{a_1}\quad \overline{z}_{a_1}\\
1\quad z_{a_2}\quad \overline{z}_{a_2}\\
1\quad z_{a_3}\quad \overline{z}_{a_3}
\epm. \nn
\ee
Then $c_1=\overline{4k\sqrt{-1}\Delta_{234}}=-4k\sqrt{-1}\Delta_{234}$.
Note that, for any $\om\in\C,|\om|=1$, if $c_i$ are replaced by $\om c_i,i=1,2,3,4$, $\tilde{A}$ is also a unitary matrix.
Thus let
\be
(c_1,c_2,c_3,c_4)=\l({4k\rho\over m_1}\Delta_{234},-{4k\rho\over m_2}\Delta_{134},{4k\rho\over m_3}\Delta_{124},-{4k\rho\over m_4}\Delta_{123}\r), \label{c}
\ee
with $\rho=\sqrt{m_1m_2m_3m_4}$.
Abusing the notations, we also write $v_4$ as
\be
v_4=(c_1,c_2,c_3,c_4)^T\in\R^4.  \label{v_4}
\ee
Now $v_1,v_2,v_3,v_4$ form a unitary basis of $\C^4$.
Note that $v_1,v_2,v_3$ are eigenvectors of matrix $D$, then $v_4$ is also an eigenvector of $D$ with the corresponding eigenvalue
\be
\lambda_4=tr(D)-\lambda_1-\lambda_2-\lambda_3=tr(D)-\mu. \nn
\ee
Moreover, we define $\bb_1$ and $\bb_2$ by
\bea
\bb_1=-\frac{\lambda_3}{\mu}=0,  \quad
\bb_2=-\frac{\lambda_4}{\mu}=1-\frac{tr(D)}{\mu}.  \label{bb2}
\eea

In the following, without causing the confusion, we will use $a_i$ to represent $z_{a_i}$, $1\le i\le4$.
By the definition of (\ref{v_3}) and (\ref{v_4}),
$Dv_k=\lambda_kv_k$, $k=3,4$, read
\bea
\mu b_i-\sum_{j=1,j\ne i}^4\frac{m_j(b_j-b_i)}{|a_i-a_j|^3}=\lambda_3 b_i,\quad 1\le i\le 4, \nn
\\
\mu c_i-\sum_{j=1,j\ne i}^4\frac{m_j(c_j-c_i)}{|a_i-a_j|^3}=\lambda_4 c_i,\quad 1\le i\le 4. \nn
\eea
For $1\leq i \leq 4$, let
\be
F_i=\sum_{j=1,j\ne i}^4\frac{m_im_j(b_i-b_j)}{|a_i-a_j|^3},\;
G_i=\sum_{j=1,j\ne i}^4\frac{m_im_j(c_i-c_j)}{|a_i-a_j|^3},\quad \label{F_ki}
\ee
then we have
\be
F_i=(\mu-\lambda_3)m_ib_i=\mu(1+\bb_1)m_ib_i,\quad G_i=(\mu-\lambda_4)m_ic_i=\mu(1+\bb_2)m_ic_i. \label{Fi.bi}
\ee

Now as in p.263 of \cite{MS}, Section 11.2 of \cite{Lon5}, we define
\begin{equation}\label{PQYX}
P=\bpm p_1\\ p_2\\ p_3\\ p_4 \epm,
\quad
Q=\bpm q_1\\ q_2\\ q_3\\ q_4 \epm,
\quad
Y=\bpm G\\ Z\\ W_1\\  W_2 \epm,
\quad
X=\bpm g\\ z\\ w_1\\ w_2 \epm,
\end{equation}
where $p_i$, $q_i$, $i=1,2,3,4$ and $G$, $Z$, $W_1$, $W_2$, $g$, $z$, $w_1$, $w_2$ are all column vectors in $\R^2$.
We make the symplectic coordinate change
\be\lb{transform1}  P=A^{-T}Y,\quad Q=AX,  \ee
where the matrix $A$ is constructed as in the proof of Proposition 2.1 in \cite{MS}.
Concretely, the matrix $A\in {\bf GL}(\R^{8})$ is given by
\begin{equation}
A=
\bpm
I\quad A_1\quad B_{1}\quad C_{1}\\
I\quad A_2\quad B_{2}\quad C_{2}\\
I\quad A_3\quad B_3\quad   C_3\\
I\quad A_4\quad B_{4}\quad C_4
\epm, \nn
\end{equation}
where each $A_i$ is a $2\times2$ matrix given by
\begin{eqnarray}
A_i = (a_i, Ja_i)=\varphi(a_i),\quad
B_i = (b_i, Jb_i)=\varphi(b_i), \quad
C_i = (c_i, Jc_i)=\varphi(c_i)=c_iI_2, \label{Cc}
\end{eqnarray}
and $\varphi$ is given by (\ref{map}).
Moreover, by the definition of $v_i,1\le i\le 4$, we obtain
\bea
\overline{\tilde{A}}^T\tilde{M}\tilde{A}&=&(\overline{v}_1,\overline{v}_2,\overline{v}_3,\overline{v}_4)^T\tilde{M}(v_1,v_2,v_3,v_4)
=I_4. \nn
\eea
By (\ref{map.matrix}), we have
$A^TMA=\varphi(\tilde{A})^T\varphi(\tilde{M})\varphi(\tilde{A})=\varphi(\overline{\tilde{A}}^T\tilde{M}\tilde{A})=\varphi(I_4)=I_{8}$
is fulfilled (cf. (13) in p.263 of \cite{MS}).

Under the coordinate change (\ref{transform1}), kinetic energy of the Hamiltonian function of the four-body problems is given by
\begin{equation}
K=\frac{1}{2}(|G|^2+|Z|^2+|W_1|^2+|W_2|^2),
\end{equation}
and the potential function is given by
\begin{eqnarray}
U(z,w_1,w_2)=\sum_{1\le i<j\le 4}U_{ij}(z,w_1,w_2),\label{U}
\end{eqnarray}
where $U_{ij}(z,w_1,w_2)=\frac{m_im_j}{|d_{ij}(z,w_1,w_2)|}$
with
\begin{eqnarray}
d_{ij}(z,w_1,w_2)&=&(A_i-A_j)z+(B_i-B_j)w_1+(C_i-C_j)w_2
\nonumber
\\
&=&\varphi(a_i-a_j)z+\varphi(b_i-b_j)w_1+\varphi(c_i-c_j)w_2,
\end{eqnarray}
by (\ref{Cc}).

Let $\theta$ be the true anomaly.
Then based on symplectic transformation in the proof of Theorem 11.10 (p. 100 of \cite{Lon5}),
the resulting Hamiltonian function of the $4$-body problem is given by
\begin{eqnarray}\label{new.H.function}
&&H(\theta,\bar{Z},\bar{W_1},\bar{W}_2,\bar{z},\bar{w_1},\bar{w}_2)=\frac{1}{2}(|\bar{Z}|^2+|\bar{W_1}|^2+|\bar{W_2}|^2)
+(\bar{z}\cdot J\bar{Z}+\bar{w_1}\cdot J\bar{W_1}+\bar{w_2}\cdot J\bar{W_2})
\nonumber
\\
&&\quad\quad\quad\quad\quad+\frac{p-r(\theta)}{2p}(|\bar{z}|^2+|\bar{w_1}|^2+|\bar{w_2}|^2)
-\frac{r(\theta)}{\sigma}U(\bar{z},\bar{w_1},\bar{w}_2),
\end{eqnarray}
where $\mu$ is given by (\ref{mu}) and $r(\theta)$ satisfies
\begin{equation}
r(\theta)=\frac{p}{1+e\cos\theta}. \lb{eqn.r.p}
\end{equation}

We now derived the linearized Hamiltonian system at the elliptic relative equilibria.
\begin{proposition}\label{linearized.Hamiltonian}
	Using notations in (\ref{PQYX}), elliptic relative equilibrium solution $(P(t),Q(t))^T$ of the system (\ref{1.2}) with
	\begin{equation}
	Q(t)=(r(t)R(\theta(t))a_1,r(t)R(\theta(t))a_2,r(t)R(\theta(t))a_3,r(t)R(\theta(t))a_4)^T,\quad P(t)=M\dot{Q}(t)
	\end{equation}
	in time $t$ with the matrix $M$ is given by (\ref{eqn:M}),
	is transformed to the new solution $(Y(\theta),X(\theta))^T$ in the variable true anomaly $\theta$
	with $G=g=0$ with respect to the original Hamiltonian function $H$ of (\ref{new.H.function}), which is given by
	\begin{equation}
	Y(\theta)=
	\bpm
	\bar{Z}(\theta)\\
	\bar{W}_1(\theta)\\
	\bar{W}_2(\theta)
	\epm
	=\bpm
	0\\
	\sigma\\
	0\\
	0\\
	0\\
	0
	\epm,
	\quad
	X(\theta)=\bpm
	\bar{z}(\theta)\\
	\bar{w_1}(\theta)\\
	\bar{w_2}(\theta)
	\epm
	=\bpm
	\sigma\\
	0\\
	0\\
	0\\
	0\\
	0
	\epm.
	\end{equation}
Moreover, the linearized Hamiltonian system at the elliptic relative equilibrium
${\xi}_0\equiv(Y(\theta),X(\theta))^T =$
\newline
$(0,\sigma,0,0,0,0,\sigma,0,0,0,0,0)^T\in\R^{12}$
depending on the true anomaly $\theta$ with respect to the Hamiltonian function
$H$ of (\ref{new.H.function}) is given by
\begin{equation}
\dot\zeta(\theta)=JB(\theta)\zeta(\theta),  \label{general.linearized.Hamiltonian.system}
\end{equation}
with
\begin{eqnarray}
B(\theta)&=&H''(\theta,\bar{Z},\bar{W_1},\bar{W}_2,\bar{z},\bar{w_1},\bar{w}_2)|_{\bar\xi=\xi_0}
\nonumber
\\
&=&
\bpm
I& O& O& -J&  O&  O\\
O& I& O&  O& -J&  O\\
O& O& I&  O&  O& -J\\
J& O& O& H_{\bar{z}\bar{z}}(\theta,\xi_0)& O& O\\
O& J& O& O& H_{\bar{w_1}\bar{w_1}}(\theta,\xi_0)& H_{\bar{w_1}\bar{w_2}}(\theta,\xi_0)\\
O& O& J& O& H_{\bar{w_2}\bar{w_1}}(\theta,\xi_0)& H_{\bar{w_2}\bar{w_2}}(\theta,\xi_0)
\epm,
\end{eqnarray}
and
\begin{eqnarray}
H_{\bar{z}\bar{z}}(\theta,\xi_0)&=&
\bpm
-\frac{2-e\cos\theta}{1+e\cos\theta} & 0\\
0 & 1
\epm, \nn
\\
H_{\bar{w_i}\bar{w_i}}(\theta,\xi_0)&=&I_2-\frac{r}{p}\left[\frac{3+\bb_i}{2}I_2+\psi(\bb_{ii})\right],\ \ i=1,2, \nn
\\
H_{\bar{w_1}\bar{w_2}}(\theta,\xi_0)&=&-\frac{r}{p}\psi(\bb_{12}),\nn
\end{eqnarray}
where $\bb_1=0$ and
$\bb_2$ are given by (\ref{bb2}),
and $\bb_{11},\bb_{12},\bb_{22}$ are given by
\bea
\bb_{11}&=&{3\over2\mu}\sum_{1\le i<j\le 4}\frac{m_im_j(a_i-a_j)^2(\overline{b}_i-\overline{b}_j)^2}{|a_i-a_j|^5},  \label{bb_11}
\\
\bb_{12}&=&{3\over2\mu}\sum_{1\le i<j\le 4}\frac{m_im_j(a_i-a_j)^2(\overline{b}_i-\overline{b}_j)(\overline{c}_i-\overline{c}_j)}{|a_i-a_j|^5},
\label{bb_12}
\\
\bb_{22}&=&{3\over2\mu}\sum_{1\le i<j\le 4}\frac{m_im_j(a_i-a_j)^2(\overline{c}_i-\overline{c}_j)^2}{|a_i-a_j|^5},\label{bb_22}
\eea
and $H''$ is the Hessian matrix of $H$ with respect to its variable $\bar{Z}$,
$\bar{W_1},\bar{W}_2$, $\bar{z}$, $\bar{w_1},\bar{w}_2$.
The corresponding quadratic Hamiltonian function is given by
\begin{eqnarray}
H_2(\theta,\bar{Z},\bar{W_1},\bar{W}_2,\bar{z},\bar{w_1},\bar{w}_2)
&=&\frac{1}{2}|\bar{Z}|^2+\bar{Z}\cdot J\bar{z}+\frac{1}{2}H_{\bar{z}\bar{z}}(\theta,\xi_0)|\bar{z}|^2+H_{\bar{w_1}\bar{w_2}}(\theta,\xi_0)\bar{w_1}\cdot\bar{w_2}
\nonumber\\
&&+\left(\frac{1}{2}|\bar{W_1}|^2+\bar{W_1}\cdot J\bar{w_1}+\frac{1}{2}H_{\bar{w_1}\bar{w_1}}(\theta,\xi_0)|\bar{w_1}|^2\right)
\nonumber\\
&&+\left(\frac{1}{2}|\bar{W_2}|^2+\bar{W_2}\cdot J\bar{w_2}+\frac{1}{2}H_{\bar{w_2}\bar{w_2}}(\theta,\xi_0)|\bar{w_2}|^2\right).
\end{eqnarray}
\end{proposition}

\begin{proof}
The proof is similar to those of Proposition 11.11 and Proposition 11.13 of \cite{Lon5}.
We only need to compute $H_{\bar{z}\bar{z}}(\theta,\xi_0)$, $H_{\bar{z}\bar{w_i}}(\theta,\xi_0)$
and $H_{\bar{w_i}\bar{w_j}}(\theta,\xi_0)$ for $i,j=1,2$.

For simplicity, we omit all the upper bars on the variables of $H$ in (\ref{new.H.function}) in this proof.
By (\ref{new.H.function}), the derivatives of $H$ with respect with the $z$ and $w_i$ is given by
\bea
H_z&=&JZ+\frac{p-r}{p}z-\frac{r}{\sigma}U_z(z,w_1,w_2),  \nn\\
H_{w_i}&=&JW_i+\frac{p-r}{p}w_i-\frac{r}{\sigma}U_{w_i}(z,w_1,w_2), \quad i=1,2, \nn
\eea
and second derivatives of $H$ with respect with the $z$ and $w_i$ by is given by
\be\lb{Hessian}\left\{
\begin{array}{l}
	H_{zz}=\frac{p-r}{p}I-\frac{r}{\sigma}U_{zz}(z,w_1,w_2),
	\\
	H_{zw_i}=H_{w_lz}=-\frac{r}{\sigma}U_{zw_i}(z,w_1,w_2),\quad i=1,2,
	\\
	H_{w_iw_i}=\frac{p-r}{p}I-\frac{r}{\sigma}U_{w_iw_i}(z,w_1,w_2),\quad i=1,2,
	\\
	H_{w_1w_2}=H_{w_2w_1}=-\frac{r}{\sigma}U_{w_1w_2}(z,w_1,w_2),
\end{array}\right. \ee
where all the items above are $2\times2$ matrices.

For $U_{ij}$ defined in (\ref{U}) with $1\le i<j\le 4$, $1\le l\le 2$,
we have
\bea
\frac{\partial U_{ij}}{\partial z}(z,w_1,w_2) &=& -\frac{m_im_j\varphi(a_i-a_j)^Td_{ij}(z,w_1,w_2)}{|d_{ij}(z,w_1,w_2)|^3},
\nn\\
\frac{\partial U_{ij}}{\partial w_1}(z,w_1,w_2) &=& -\frac{m_im_j\varphi(b_i-b_j)^Td_{ij}(z,w_1,w_2)}{|d_{ij}(z,w_1,w_2)|^3}
,
\nn\\
\frac{\partial U_{ij}}{\partial w_2}(z,w_1,w_2) &=& -\frac{m_im_j\varphi(c_i-c_j)^Td_{ij}(z,w_1,w_2)}{|d_{ij}(z,w_1,w_2)|^3},\nn
\eea
and then
\bea
\frac{\partial^2 U_{ij}}{\partial z^2}(z,w_1,w_2)
&=&-\frac{m_im_j|a_i-a_j|^2I_2}{|d_{ij}(z,w_1,w_2)|^3} +\frac{3m_im_j\varphi(a_i-a_j)^T\Phi(z,w_1,w_2)\varphi(a_i-a_j)}{|d_{ij}(z,w_1,w_2)|^5}, \nn\\
\frac{\partial^2 U_{ij}}{\partial {w_1}^2}(z,w_1,w_2)
&=&-\frac{m_im_j|b_i-b_j|^2I_2}{|d_{ij}(z,w_1,w_2)|^3} +\frac{3m_im_j\varphi(b_i-b_j)^T\Phi(z,w_1,w_2)\varphi(b_i-b_j)}{|d_{ij}(z,w_1,w_2)|^5},\nn\\
\frac{\partial^2 U_{ij}}{\partial z\partial w_1}(z,w_1,w_2)&=&
-\frac{m_im_j\varphi(a_i-a_j)^T\varphi(b_i-b_j)}{|d_{ij}(z,w_1,w_2)|^3}+\frac{3m_im_j\varphi(a_i-a_j)^T\Phi(z,w_1,w_2)\varphi(b_i-b_j)}{|d_{ij}(z,w_1,w_2)|^5}, \nn
\eea
where $\Phi(z,w_1,w_2) = d_{ij}(z,w_1,w_2)d_{ij}(z,w_1,w_2)^T$.
Let
$$
K=\bpm 2 & 0 \\
0 & -1 \epm, \quad
K_1=\bpm 1 & 0 \\
0 & 0 \epm,\quad
K_2=\bpm 1 & 0 \\
0 & -1 \epm=\psi(1),
$$
where $\psi$ is given by (\ref{map}).
Now evaluating these functions at $\bar\xi_0=(0,\sigma,0,0,0,0,\sigma,0,0,0,0,0)^T\in\R^{12}$
with $z=(\sigma,0)^T,w_i=(0,0)^T,1\le i\le 2$, and summing them up,
we obtain
\begin{eqnarray}
\left.\frac{\partial^2 U}{\partial z^2}\right|_{\xi_0}&=&
\sum_{1\le i<j\le 4} \left.\frac{\partial^2 U_{ij}}{\partial z^2}\right|_{\xi_0}
\nonumber\\
&=&\sum_{1\le i<j\le 4}\left(-\frac{m_im_j|a_i-a_j|^2}{|(a_i-a_j)\sigma|^3}I
+3\frac{m_im_j\sigma^2|a_i-a_j|^2K_1|a_i-a_j|^2}{|(a_i-a_j)\sigma|^5}\right)
\nonumber\\
&=&\frac{1}{\sigma^3}\left(\sum_{1\le i<j\le4}\frac{m_im_j}{|a_i-a_j|}\right)K
\nonumber\\
&=&\frac{\mu}{\sigma^3}K,  \label{U_zz}
\end{eqnarray}
where third equality holds by (\ref{F_ki}),
and
\begin{eqnarray}
\frac{\partial^2 U}{\partial w_1^2}\left|_{\xi_0}\right.
&=& \sum_{1\le i<j\le 4}\frac{\partial^2 U_{ij}}{\partial     w_l^2}\left|_{\xi_0}\right.\nonumber\\
&=&\sum_{1\le i<j\le 4}\left(-\frac{m_im_j|b_i-b_j|^2}{|(a_i-a_j)\sigma|^3}I
  +3\frac{m_im_j\sigma^2\varphi(b_i-b_j)^T\varphi(a_i-a_j)K_1\varphi(a_i-a_j)^T\varphi(b_i-b_j)}{|(a_i-a_j)\sigma|^5}\right)
\nonumber\\
&=&\sum_{1\le i<j\le 4}\left(-\frac{m_im_j|b_i-b_j|^2}{|(a_i-a_j)\sigma|^3}I
  +{3\over2}\frac{m_im_j\sigma^2\varphi(b_i-b_j)^T\varphi(a_i-a_j)\varphi(a_i-a_j)^T\varphi(b_i-b_j)}{|(a_i-a_j)\sigma|^5}\right)\nonumber\\
  &&+\sum_{1\le i<j\le 4}
  \left({3\over2}\frac{m_im_j\sigma^2\varphi(b_i-b_j)^T\varphi(a_i-a_j)\psi(1)\varphi(a_i-a_j)^T\varphi(b_i-b_j)}{|(a_i-a_j)\sigma|^5}\right)\nonumber\\
&=&\sum_{1\le i<j\le 4}\left(-\frac{m_im_j|b_i-b_j|^2}{|(a_i-a_j)\sigma|^3}I
  +{3\over2}\frac{m_im_j\sigma^2\varphi(|b_i-b_j|^2|a_i-a_j|^2)}{|(a_i-a_j)\sigma|^5}\right)\nonumber\\
&&+\sum_{1\le i<j\le 4}
\left({3\over2}\frac{m_im_j\sigma^2\psi((a_i-a_j)^2(\overline{b}_i-\overline{b}_j)^2)}{|(a_i-a_j)\sigma|^5}\right)
\nonumber\\
&=&{1\over2\sigma^3}\sum_{1\le i<j\le 4}\left(\frac{m_im_j|b_i-b_j|^2}{|a_i-  a_j|^3}\right)I_2
  +{1\over\sigma^3}\psi\left({3\over2}\sum_{1\le i<j\le 4}
  \frac{m_im_j(a_i-a_j)^2(\overline{b}_i-\overline{b}_j)^2}{|a_i-a_j|^5}\right)\nonumber\\
&=&\frac{1}{2\sigma^3}\left(\sum_{i=1}^4 \overline{b}_i\sum_{j=1,j\ne   i}^4\frac{m_im_j(b_i-b_j)}{|a_i-a_j|^3}\right)I_2
  +{1\over\sigma^3}\psi\left({3\over2}\sum_{1\le i<j\le 4}
\frac{m_im_j(a_i-a_j)^2(\overline{b}_i-\overline{b}_j)^2}{|a_i- a_j|^5}\right)\nonumber\\
&=&\frac{1}{2\sigma^3}\left(\sum_{i=1}^4 \overline{b}_iF_i\right)I_2
  +{1\over\sigma^3}\psi\left({3\over2}\sum_{1\le i<j\le 4}
  \frac{m_im_j(a_i-a_j)^2(\overline{b}_i-\overline{b}_j)^2}{|a_i-a_j|^5}\right)\nonumber\\
&=&\frac{\mu(1+\bb_1)}{2\sigma^3}I_2
+{\mu\over\sigma^3}\psi(\bb_{11}), \label{U_w1w1}
\end{eqnarray}
where the second formula holds because of (\ref{Fi.bi}) and (\ref{bb_11}).
Similarly, we have
\bea
\left.\frac{\partial^2 U}{\partial w_2^2}\right|_{\xi_0}&=&
\frac{\mu(1+\bb_2)}{2\sigma^3}I_2
+{1\over\sigma^3}\psi\left({3\over2}\sum_{1\le i<j\le 4}
\frac{m_im_j(a_i-a_j)^2(\overline{c}_i-\overline{c}_j)^2}{|a_i-a_j|^5}\right)
\nonumber
\\
&=&\frac{\mu(1+\bb_2)}{2\sigma^3}I_2
+{\mu\over\sigma^3}\psi(\bb_{22}), \label{U_w2w2}
\\
\left.\frac{\partial^2 U}{\partial w_1\partial w_2}\right|_{\xi_0}&=&
{1\over\sigma^3}\psi\left({3\over2}\sum_{1\le i<j\le 4}
\frac{m_im_j(a_i-a_j)^2(\overline{b}_i-\overline{b}_j)(\overline{c}_i-\overline{c}_j)}{|a_i-a_j|^5}\right)
\nonumber\\
&=&{\mu\over\sigma^3}\psi(\bb_{12}). \label{U_w1w2}
\eea
Moreover, we have
\begin{eqnarray}
\left.\frac{\partial^2 U}{\partial z\partial w_1}\right|_{\xi_0}&=&
\left.\sum_{1\le i<j\le 4}\frac{\partial^2 U_{ij}}{\partial z\partial w_1}\right|_{\xi_0}
\nonumber\\
&=&\sum_{1\le i<j\le4}\left(-\frac{m_im_j\varphi(a_i-a_j)^T\varphi(b_i-b_j)}{|(a_i-a_j)\sigma|^3}
+3\frac{m_im_j\sigma^2|a_i-a_j|^2K_1\varphi(a_i-a_j)^T\varphi(b_i-b_j)}{|(a_i-a_j)\sigma|^5}\right)
\nonumber\\
&=&\frac{K}{\sigma^3}\left(\sum_{1\le i<j\le 4}\frac{m_im_j\varphi((\overline{a}_i-\overline{a}_j)(b_i-b_j))}{|a_i-a_j|^3}\right)
\nonumber\\
&=&\frac{K}{\sigma^3}\varphi\left(\sum_{1\le i<j\le 4}\frac{m_im_j(\overline{a}_i-\overline{a}_j)(b_i-b_j)}{|a_i-a_j|^3}\right)
\nonumber\\
&=&\frac{K}{\sigma^3}\varphi\left(\sum_{1\le i<j\le 4}\frac{m_im_j\overline{a}_i(b_i-b_j)}{|a_i-a_j|^3}
-\sum_{1\le i<j\le 4}\frac{m_im_j\overline{a}_j(b_i-b_j)}{|a_i-a_j)|^3}\right)
\nonumber\\
&=&\frac{K}{\sigma^3}\varphi\left(\sum_{1\le i<j\le 4}\frac{m_im_j\overline{a}_i(b_i-b_j)}{|a_i-a_j|^3}
-\sum_{1\le j<i\le 4}\frac{m_jm_i\overline{a}_i(b_j-b_i)}{|a_j-a_i)|^3}\right)
\nonumber\\
&=&\frac{K}{\sigma^3}\varphi\left(\sum_{i=1}^4\overline{a}_i\sum_{j=i+1}^4\frac{m_im_j(b_i-b_j)}{|a_i-a_j|^3}
+\sum_{i=1}^4\overline{a}_i\sum_{j=1}^{i-1}\frac{m_im_j(b_i-b_j)}{|a_i-a_j|^3}\right)
\nonumber\\
&=&\frac{K}{\sigma^3}\varphi\left(\sum_{i=1}^4\overline{a}_i\sum_{j=1,j\ne i}^4\frac{m_im_j(b_i-b_j)}{|a_i-a_j|^3}\right)
\nonumber\\
&=&\frac{K}{\sigma^3}\varphi\left(\sum_{i=1}^4\overline{a}_iF_i\right)
\nonumber\\
&=&\frac{K}{\sigma^3}\varphi\left(\mu(1+\bb_1)\sum_{i=1}^4m_i\overline{a}_ib_i\right)
\nonumber\\
&=&O,\label{U_zw1}
\end{eqnarray}
where the second last equality holds because (\ref{eq.of.cc}), and the last equality holds because of (\ref{Fi.bi}).
Similarly, we have
\be
\left.\frac{\partial^2 U}{\partial z\partial w_2}\right|_{\xi_0}=0. \label{U_zw2}
\ee

By $r(\th) = \frac{p}{1+e\cos\th}$, $\sigma^4 = \mu p$ and (\ref{U_zz})-(\ref{U_zw2}), the second derivative of $H$ are given by
\begin{eqnarray}
H_{zz}|_{\xi_0}&=&\frac{p-r}{p}I-\frac{r\mu}{\sigma^4}K=I-\frac{r}{p}I-\frac{r\mu}{p\mu}K
=I-\frac{r}{p}(I+K)
=\bpm -\frac{2-e\cos\theta}{1+e\cos\theta} & 0\\
0 & 1 \epm,  \nn\\
H_{zw_i}|_{\xi_0}&=&-\frac{r}{\sigma}\frac{\partial^2U}{\partial z\partial w_i}|_{\xi_0}=O,\quad 1\le i\le 2,
\nn\\
H_{w_1w_1}|_{\xi_0}&=&\frac{p-r}{p}I-\frac{r}{\sigma}
\left[\frac{\mu(1+\bb_1)}{2\sigma^3}I_2+{\mu\over\sigma^3}\psi(\bb_{11})\right]
=I-\frac{r}{p}I-\frac{r}{p}\left[\frac{1+\bb_1}{2}I_2+\psi(\bb_{11})\right]
\nonumber
\\
&=&I-\frac{r}{p}\left[\frac{3+\bb_1}{2}I_2+\psi(\bb_{11})\right],
\nn\\
H_{w_2w_2}|_{\xi_0}&=&\frac{p-r}{p}I-\frac{r}{\sigma}
\left[\frac{\mu(1+\bb_2)}{2\sigma^3}I_2+{\mu\over\sigma^3}\psi(\bb_{22})\right]
=I-\frac{r}{p}I-\frac{r}{p}\left[\frac{1+\bb_2}{2}I_2+\psi(\bb_{22})\right]
\nonumber
\\
&=&I-\frac{r}{p}\left[\frac{3+\bb_2}{2}I_2+\psi(\bb_{22})\right],
\nn\\
H_{w_1w_2}|_{\xi_0}&=&H_{w_2w_1}|_{\xi_0}=-\frac{r}{\sigma}\frac{\partial^2U}{\partial w_1\partial w_2}|_{\xi_0}=-\frac{r}{p}\psi(\bb_{12}).
\end{eqnarray}
Thus the proof is complete.
\end{proof}

\subsection{The reduction at ERE of Lagrangian central configuration and one zero mass}
We consider the central configurations of restricted 4-body problem where  three primaries form an equilateral triangle.
We fix $m_1,m_2\in (0,1)$ such that $m_1+m_2<1$,
and let $m_3=1-m_1-m_2-\ep$ and $m_4=\ep$ with
$0<\ep<1-m_1-m_2$. Therefore, $\sum_{i = 0} ^4 m_i=1$.

Let $q_1=0$ and $q_2=1$. When $\epsilon\to0$, by our assumption,
$q_3$ must tend to one of the Lagrangian points of $m_1$ and $m_2$.
Without lose of generality, in the complex plane, we suppose such Lagrangian point
\be
z_L={1\over2}+\sqrt{-1}{\sqrt{3}\over2}.
\ee
Moreover, the limit positions of the zero mass $q_4$ are discussed in \cite{Leandro},
and we denote it by $z^*$, namely,
\bea
\lim_{\ep\to0}q_3 = z_L, \quad
\lim_{\ep\to0}q_4 = z^*. \lb{OneSmallSec:q4.limit}
\eea

The center of mass of the four particles is
\be
q_c=\sum_{i=1}^4m_iq_i=m_2+(1-m_1-m_2)q_3+\epsilon(q_4-q_3). \nn
\ee
For $i=1$, $2$, $3$ and $4$, re-scaling the distance between each body and the center of mass $q_i-q_c$ by
\be
a_i=(q_i-q_c)\alpha, \nn
\ee
where $\alpha>0$  such that
\be
\sum_{i=1}^4m_i|a_i|^2=1. \nn
\ee
Moreover, let
\be
\alpha_0=\lim_{\epsilon\to 0}\alpha=[m_1+m_2-(m_1^2+m_1m_2+m_2^2)]^{-{1\over2}},\lb{eq:aa0}
\ee
and
\be
q_{c,0}=\lim_{\epsilon\to 0}q_c=\left[{1\over2}(1-m_1+m_2)+\sqrt{-1}{\sqrt{3}\over2}(1-m_1-m_2)\right]\alpha_0, \nn
\ee
and hence
\bea
a_{1,0} &=& \lim_{\epsilon\to 0}{a_1}
= -\left[{1\over2}(1-m_1+m_2)+\sqrt{-1}{\sqrt{3}\over2}(1-m_1-m_2)\right]\alpha_0, \lb{OneSmallSec:a_10}\\ 
a_{2,0} &=& \lim_{\epsilon\to 0}{a_2}
= \left[{1\over2}(1+m_1-m_2)-\sqrt{-1}{\sqrt{3}\over2}(1-m_1-m_2)\right]\alpha_0,  \\
a_{3,0} &=& \lim_{\epsilon\to 0}{a_3}
= \left[{1\over2}(m_1-m_2)+\sqrt{-1}{\sqrt{3}\over2}(m_1+m_2)\right]\alpha_0,  \\ 
a_{4,0} &=& \lim_{\epsilon\to 0}{a_4}
= \left[z^*-{1\over2}(1-m_1+m_2)-\sqrt{-1}{\sqrt{3}\over2}(1-m_1-m_2)\right]\alpha_0.\lb{OneSmallSec:a_40}
\eea 
The potential $\mu = U(a)$ is given by
\be
\mu=\mu_{\epsilon,\tau}=\sum_{1\le i<j\le4}\frac{m_im_j}{|a_i-a_j|}, \nn
\ee
and by Lemma 3 of \cite{IM}, we have
\be
\mu_0=\lim_{\epsilon\to 0}\mu
=\frac{m_1m_2+m_2(1-m_1-m_2)+(1-m_1-m_2)m_1}{\alpha_0}=\alpha_0^{-3}.\label{mu0}
\ee
In the following, we will use the subscript $0$ to denote the limit value of the parameters
when $\epsilon\to 0$.

We now calculate $k$ and $l$ defined by (\ref{param.l}) for our case.
We first have
\bea
\lim_{\ep\to0} \sum_{i=1}^4 m_i\bar{a}_i^2
&=&\lim_{\ep\to0} \frac{1}{2}\sum_{i=1}^4 \sum_{j=1}^4 m_im_j(\bar{a}_i-\bar{a}_j)^2
\nn\\
&=& \frac{1}{2} \sum_{i=1}^3 \sum_{j=1}^3 m_im_j(\bar{a}_{i,0}-\bar{a}_{j,0})^2
\nn\\
&=&\alpha_0^2\left[m_1m_2-{1\over2}(m_1+m_2)(1-m_1-m_2)+\sqrt{-1}{\sqrt{3}\over2}(m_2-m_1)(1-m_1-m_2)\right].\nn\\
\eea
where the second equality holds by $m_4=\ep\to0$ and the third equality holds by (\ref{OneSmallSec:q4.limit}).
Hence by the definition of $k$ in (\ref{param.l}), $k_0$ is given by
\bea
k_0&=&\lim_{\ep\to0}\l(1-\l|\sum_{i=1}^4m_i\bar{a}_i^2\r|^2\r)^{-\frac{1}{2}}
\nn\\
&=&\l(1-\alpha_0^4\left[m_1m_2-{1\over2}(m_1+m_2)(1-m_1-m_2)\right]^2-{3\over4}\alpha_0^4[(m_2-m_1)(1-m_1-m_2)]^2\r)^{-\frac{1}{2}}
\nn\\
&=&\l(1-\alpha_0^4\left[\alpha_0^{-2}-{3\over2}(m_1+m_2)(1-m_1-m_2)\right]^2-{3\over4}\alpha_0^4[(m_2-m_1)(1-m_1-m_2)]^2\r)^{-\frac{1}{2}}
\nn\\
&=&\alpha_0^2\l(3m_1m_2(1-m_1-m_2)\r)^{-\frac{1}{2}},  \label{OneSmallSec:k0}
\eea
where one may verify the last equality holds by expending all the brackets and plug in $\aa_0$ which is defined by (\ref{eq:aa0}).
By the definition of $l_0$ in (\ref{param.l}) and direct computations, we obtain that
\bea
l_0
&=&-k_0\lim_{\ep\to0}\sum_{i=1}^4m_i\bar{a}_i^2
\nonumber\\
&=&-\frac{m_1m_2-{1\over2}(m_1+m_2)(1-m_1-m_2)+\sqrt{-1}{\sqrt{3}\over2}(m_2-m_1)(1-m_1-m_2)}
{\sqrt{3m_1m_2(1-m_1-m_2)}}.
\label{OneSmallSec:l0}
\eea
Moreover, by (\ref{c}), we obtain for $i$, $j$, $k\in \{1,2,3\}$,
\bea
\lim_{\ep\to0}c_i=\lim_{\ep\to0}4k\sqrt{m_j m_k m_4\over m_i}\Delta_{jk4}=0, \label{OneSmallSec.lim.c1}
\eea
and
\bea
\lim_{\ep\to0}\sqrt{m_4}c_4&=&\lim_{\ep\to0}-4k\sqrt{m_1m_2m_3}\Delta_{123}=-4k_0\sqrt{m_1m_2(1-m_1-m_2)}{\sqrt{3}\over4}\alpha_0^2
=-1, \label{OneSmallSec.lim.c4}
\eea
where the last equality holds by (\ref{OneSmallSec:k0}).

Plugging (\ref{OneSmallSec:a_10})-(\ref{OneSmallSec:a_40}) and (\ref{mu0}) into $\bb_{2}$ of (\ref{bb2}), $\bb_{2,0}$ can be obtained by
\bea
\bb_{2,0}&=&\lim_{\ep\to0}\bb_2
\nn\\
&=&1-\lim_{\ep\to0}{{\rm tr}(D)\over\mu}
\nn\\
&=&1-{1\over\mu_0}\left[4\mu_0-\sum_{1\le i<j\le3}{m_i+m_j\over|a_{i,0}-a_{j,0}|^3}
-\sum_{i=1}^3{m_i\over|a_{i,0}-a_{4,0}|^3}\right]
\nn\\
&=&1-{1\over\mu_0}\left[4\mu_0-2\mu_0-\sum_{i=1}^3{m_i\over|a_{i,0}-a_{4,0}|^3}\right]
\nn\\
&=&{1\over\mu_0}\sum_{i=1}^3{m_i\over|a_{i,0}-a_{4,0}|^3}-1.\nn
\eea
By (\ref{bb_12}) and (\ref{OneSmallSec.lim.c1})-(\ref{OneSmallSec.lim.c4}), we have
\be
\bb_{12,0}={3\over2\mu_0}\lim_{\ep\to0}
\sum_{1\le i<j\le4}\frac{m_im_j(a_i-a_j)^2(\bar{b}_i-\bar{b}_j)(\bar{c}_i-\bar{c}_j)}{|a_i-a_j|^5}
=0.
\label{bb_120}
\ee
Note that $\bb_{11,0}$ is given by
\begingroup\allowdisplaybreaks
\bea
\bb_{11,0}&=&{3\over2\mu_0}\lim_{\ep\to0}\sum_{1\le i<j\le4}\frac{m_im_j(a_i-a_j)^2(\bar{b}_i-\bar{b}_j)^2}{|a_i-a_j|^5}
\nonumber\\
&=&{3\over2\mu_0}\sum_{1\le i<j\le3}\frac{m_im_j(a_i-a_j)^2(\bar{b}_i-\bar{b}_j)^2}{|a_i-a_j|^5}
\nn\\
&=&{3\over2\mu_0}\alpha_0^{-5}
\sum_{1\le i<j\le3}m_{i,0}m_{j,0}(a_{i,0}-a_{j,0})^2[k(a_{i,0}-a_{j,0})+\bar{l}_0(\bar{a}_{i,0}-\bar{a}_{j,0})]^2
\nonumber\\
&=&{3\over2\mu_0}\alpha_0^{-5}\left[k_0^2\sum_{1\le i<j\le3}m_{i,0}m_{j,0}(a_{i,0}-a_{j,0})^4
+2\alpha_0^2k_0\bar{l}_0\sum_{1\le i<j\le3}m_{i,0}m_{j,0}(a_{i,0}-a_{j,0})^2+\alpha_0^2\bar{l}_0^2\right]
\nn\\
&=&{3\over2\mu_0}\alpha_0^{-5}\left[k_0^2\sum_{1\le i<j\le 3}m_{i,0}m_{j,0}(a_{i,0}-a_{j,0})^{-2}\alpha_0^6
+2\alpha_0^2k_0\bar{l}_0\left({-\bar{l}_0\over k_0}\right)+\alpha_0^2\bar{l}_0^2\right]
\nn\\
&=&{3\over2\mu_0}\alpha_0^{-5}
\left[k_0^2\sum_{1\le i<j\le3}m_{i,0}m_{j,0}\overline{(a_{i,0}-a_{j,0})}^{2}\alpha_0^2
-\alpha_0^2\bar{l}_0^2\right]
\nn\\
&=&{3\over2\mu_0}\alpha_0^{-5}\left[k_0^2\left({-l_0\over k_0}\right)\alpha_0^2-\alpha_0^2\bar{l}_0^2\right]
\nn\\
&=&{3\over2}(-k_0l_0-\bar{l}_0^2)
\nn\\
&=&{3\over2}\left\{\frac{[m_1m_2-{1\over2}(m_1+m_2)m_3^*+\sqrt{-1}{\sqrt{3}\over2}(m_2-m_1)m_3^*]\alpha_0^{-2}}{3m_1m_2m_3^*}\right.
\nn\\
&& \quad-\left.\frac{[m_1m_2-{1\over2}(m_1+m_2)m_3^*-\sqrt{-1}{\sqrt{3}\over2}(m_2-m_1)m_3^*]^2}{3m_1m_2m_3^*}\right\}
\nn\\
&=&{1\over2m_1m_2m_3^*}\left\{[m_1m_2-{1\over2}(m_1+m_2)m_3^*][{3\over2}(m_1+m_2)m_3^*] +{3\over4}(m_2-m_1)^2(m_3^*)^2
\right.\nn\\
&&\qquad+\left.\sqrt{-1}{\sqrt{9}\over2}(m_2-m_1)m_3^*m_1m_2\right\}
\nn\\
&=&{m_3^*\over 2m_1m_2m_3^*}\left\{{3\over2}m_1m_2(m_1+m_2)-{3\over4}(m_1+m_2)^2m_3^*
+{3\over4}(m_2-m_1)^2m_3^*\right\}\nn\\
&&\quad
+\sqrt{-1}{3\sqrt{3}\over4}(m_2-m_1)
\nn\\
&=&{3\over4}\left[3(m_1+m_2)-2+\sqrt{-1}\sqrt{3}(m_2-m_1)\right],\lb{eq:bb110}
\eea
\endgroup
where $m_3^* = 1-m_1-m_2$, the fifth equality holds by $(a_{i,0}-a_{j,0})^6=|a_{i,0}-a_{j,0}|^6=\alpha_0^6$ for $1\le i<j\le3$,
the sixth equality
holds by $(a_{i,0}-a_{j,0})^2\overline{(a_{i,0}-a_{j,0})}^2=|a_{i,0}-a_{j,0}|^4=\alpha_0^4$ for $1\le i<j\le3$
and tenth equality holds by (\ref{OneSmallSec:k0}) and (\ref{OneSmallSec:l0}). Also, $\bb_{22,0}$ is given by
\bea
\bb_{22,0}&=&{3\over2\mu_0}\lim_{\ep\to0}
\sum_{1\le i<j\le4}\frac{m_im_j(a_i-a_j)^2(\bar{c}_i-\bar{c}_j)^2}{|a_i-a_j|^5}
\nonumber\\
&=&{3\over2\mu_0}\lim_{\ep\to0}\sum_{1\le i<j\le3}\frac{m_im_j(a_i-a_j)^2(\bar{c}_i-\bar{c}_j)^2}{|a_i-a_j|^5}
+{3\over2\mu_0}\lim_{\ep\to0}\sum_{i=1}^3\frac{m_im_4(a_i-a_4)^2(\bar{c}_i-\bar{c}_4)^2}{|a_i-a_4|^5}
\nonumber\\
&=&{3\over2\mu_0}\lim_{\ep\to0}\sum_{i=1}^3\frac{m_i(a_i-a_4)^2(\sqrt{m_4}\bar{c}_i-\sqrt{m_4}\bar{c}_4)^2}{|a_i-a_4|^5}
\nonumber\\
&=&{3\over2\mu_0}\sum_{i=1}^3\frac{m_i(a_{i,0}-a_{4,0})^2}{|a_{i,0}-a_{4,0}|^5}.\nn
\eea

When $m_4=\epsilon\to0$,
since $\bb_{12,0}=0$ of (\ref{bb_120}),
the linearized Hamiltonian system (\ref{general.linearized.Hamiltonian.system})
can be decomposed to three independent Hamiltonian systems where
the first one is the linearized Hamiltonian system of the Kepler two-body problem at Kepler elliptic orbit,
and the other two systems can be written as
\be
\dot{\zeta}_i(\th) = JB_{i,0}(\th)\zeta_i(\th), \label{linearized.system.sep_i}\\
\ee
with
\be
B_{i,0}=\bpm I_2& -J_2\\ J_2& I_2-{r\over p}\left[{3+\bb_{i,0}\over2}I_2+\psi(\bb_{ii,0})\right] \epm, \nn
\ee
for $i=1,2$.
Thus the linear stability restricted four-body problem in our case
can be reduced to the linear stability problems of system (\ref{linearized.system.sep_i}) with $i=1,2$.

Let
\be
D_i={3+\bb_{i,0}\over2}I_2+\psi(\bb_{ii,0}), \label{D_i}
\ee
for $i=1,2$.
Then by $\bb_{1,0} =0$ in (\ref{bb2}) and (\ref{eq:bb110}), $D_1$ is given by
\be
D_1=\bpm  {9\over4}(m_1+m_2)& {3\sqrt{3}\over4}(m_2-m_1)\\
{3\sqrt{3}\over4}(m_2-m_1)&3-{9\over4}(m_1+m_2)
\epm, \nn
\ee
and hence the two characteristic roots of $D_1$ are:
\be
\lambda_{1,2}=\frac{3\pm\sqrt{9-\bb_L}}{2}, \nn
\ee
where $\bb_L$ is given by
$ \bb_L=27\alpha_0^2=27[m_1m_2+(m_1+m_2)(1-m_1-m_2)]$ in (\ref{L:bb}).

As in the proof of Theorem11.14 of \cite{Lon4},
the system (\ref{linearized.system.sep_i}) for $i=1$ becomes
\be\label{linearized.system.sep_1}
\dot{\zeta}_1(\th)= J B_{1,0}\zeta_1(\th)
=J \bpm 1&  0&  0& 1\\
0&  1& -1& 0\\
0& -1& 1-\frac{3+\sqrt{9-\bb}}{2(1+e\cos\th)}& 0\\
1&  0&   0& 1-\frac{3-\sqrt{9-\bb}}{2(1+e\cos\th)}
\epm
\zeta_1(\th),
\ee
thus this system coincides with the essential part of the linearized Hamiltonian system near the elliptic Lagrangian
relative equilibria of the three-body problem with masses $m_1,m_2$ and $m_3 = 1-m_1-m_2$ because the zero mass has no effect on the other three masses.
The system (\ref{linearized.system.sep_i}) for $i=1$ has been studied in detail in \cite{HLS}.
We will focus on the system (\ref{linearized.system.sep_i}) for $i=2$, which is called the essential part in the rest of this paper. It corresponds to the interactions of the zero mass body and three primaries.

The matrix $D_{2}$ is given by
\be
D_2={3+\bb_{2,0}\over2}I_2+\psi(\bb_{22,0}). \label{D_2}
\ee
The characteristic polynomial $\det(D_2-\lm I_2)$ of $D_2$ is
\be
\lambda^2-(3+\bb_{2,0})\lambda+\l({3+\bb_{2,0}\over2}\r)^2-|\bb_{22,0}|^2,\nn
\ee
and the characteristic roots are
\be
\lambda_{3,4}={3+\bb_{2,0}\over2}\pm|\bb_{22,0}|. \nn
\ee
By direct computations, we have that
\bea\label{lambda_3}
\lm_3 &=& {3+\bb_{2,0}\over2}+|\bb_{22,0}| \nn \\
&=& 1+ \frac{1}{2\mu_0}\sum_{i=1}^3{m_i\over|a_{i,0}-a_{4,0}|^3} +  \frac{3}{2\mu_0}\l|\sum_{i=1}^3{m_i(a_{i,0}-a_{4,0})^2\over|a_{i,0}-a_{4,0}|^5}\r|\nn \\
&=&1 + {1\over2}\sum_{i=1}^3\frac{m_i}{|q_{i,0}-z^*|^3}+ {3\over2}\l|\sum_{i=1}^3\frac{m_i(q_{i,0}-z^*)^2}{|q_{i,0}-z^*|^5}\r|\nn \\
&=& 1 + \aa+3\bb,
\eea
where $\aa\geq \bb>0$ are defined by (\ref{aa.bb})
and
\bea\label{lambda_4}
\lm_4 ={3+\bb_{2,0}\over2}-|\bb_{22,0}|
=1 + \aa - 3\bb.
\eea
Then the system (\ref{linearized.system.sep_i}) when $i =2$ can be written as
\be
\dot{\xi}(\th) = J B_{2,0}(\th) \xi(\th)
=J \bpm 1&  0&  0& 1\\
0&  1& -1& 0\\
0& -1& 1-\frac{1+\aa+3\bb}{1+e\cos\th}& 0\\
1&  0&   0& 1-\frac{1+\aa-3\bb}{1+e\cos\th}
\epm
\xi(\th). \lb{linearized.system.main}
\ee

By the transformation introduced by Section 2.4 of \cite{HLS}, the system can be related with operator $\cA(\aa,\bb, e)$, i.e.,
\bea
\cA(\aa,\bb,e) &=& -\frac{\d^2}{\d t^2}I_2  -I_2 + \frac{1}{1+e\cos t}R(t)K_{\bb,e}R(t)^T \nn \\
&=& -\frac{\d^2}{\d t^2}I_2  -I_2 + \frac{1}{1+e\cos t}(
(1+\aa)I_2 + 3\bb S(t)), \lb{cA}
\eea
where $R(t) = \l(\begin{smallmatrix}\cos t & -\sin t\\ \sin t & \cos t \end{smallmatrix}\r)$,
$K_{\bb,e}=
(\begin{smallmatrix}
1+\aa+3\bb & 0\\ 0 & 1+\aa-3\bb
\end{smallmatrix})$
and
$S(t) = (\begin{smallmatrix}
\cos 2t & \sin 2t\\ \sin 2t & \cos 2t
\end{smallmatrix})$.
By Lemma \ref{lem:2.1},
we have for any $(\aa,\bb,e)\in[\bb, +\infty) \times [0, +\infty) \times [0,1)$ and $\om\in \U$,
the Morse indices $\phi_{\omega}(\cA(\aa,\bb,e))$ and nullity
$\nu_{\om}(\cA(\aa,\bb,e))$ on the domain $D(\om,2\pi)$ satisfy
\be
\phi_{\om} (\cA(\aa,\bb,e)) = i_{\om}(\xi_{\aa,\bb,e}), \quad \nu_{\om} (\cA(\aa,\bb,e)) = \nu_{\om}(\xi_{\aa,\bb,e}),
\quad \forall \om \in \U. \lb{eqn:ind.equ}
\ee
where $i_{\om}(\xi_{\aa,\bb,e})$ is the $\om$-Maslov-type index and $\nu_{\om}(\xi_{\aa,\bb,e}) = \dim\ker(\xi_{\aa,\bb,e}(2\pi)-\om I_4)$ is the nullity of the sympletic path $\xi_{\aa,\bb,e}(t)$ for $t\in [0,2\pi]$ where $\xi_{\aa,\bb,e}(t)$ is the solution of \eqref{linearized.system.main}.

\setcounter{equation}{0}
\setcounter{figure}{0}
\section{Stability of the Circular Orbits}\lb{sec:e=0}
When $e = 0$ and $\aa\geq \bb>0$, the orbit of each body is circle and the linearized system (\ref{linearized.system.main}) is given by
\be
\dot{\xi}(\th)
=JB_{2,0}\xi(\th)
=J\bpm 1&  0&  0& 1\\
0&  1& -1& 0\\
0& -1& -\aa-3\bb & 0\\
1&  0&   0& -\aa+3\bb
\epm
\xi(\th).
\ee
The corresponding characteristic polynomial $\det(JB_{2,0}-\lm I)$ of $JB_{2,0}$ is given by
\be
p(\lm) = \lm^4+(2-2\aa)\lm^2+ (1+\aa)^2-9\bb^2. \lb{p2.e0}
\ee
The four roots of $p(\lm) =0$ are given by
\bea
\lm_{1,\pm} = \pm \sqrt{\aa -1 +\sqrt{9\bb^2 - 4\aa}},\lb{lm1} \quad
\lm_{2,\pm} = \pm \sqrt{\aa -1 -\sqrt{9\bb^2 - 4\aa}}. \lb{lm2}
\eea
The four characteristic multipliers of the matrix $\xi_{\aa,\bb,0}(2\pi)$ are given by
\be
\rho_{i,\pm} = e^{2\pi \lm_{i,\pm}}, \mbox{ for } i = 1,2.
\ee

According our assumption that $\aa \geq \bb \geq 0$, we have divide the region of $(\aa,\bb,e) \in [0, \infty) \times [0, \infty) \times \{0\}$ into four sub-regions $\cR_i$ for $1\leq i\leq 4$ as following. The corresponding figure is shown in Figure \ref{fig:1}

\begin{description}
    \item[I]   The first region $\cR_1 \equiv \{(\aa,\bb)|\aa\geq\bb>0,\aa>\frac{9}{4}\bb^2\}$.

    In $\cR_1$, both $\eta_{1}$ and $\eta_{2}$ are complex numbers because $ 9\bb^2-4\aa < 0$. It yields that $\lm_{i,\pm} \in \C $ with non-zero imaginary parts. Then the four characteristic multipliers of the matrix $\xi_{\aa,\bb,0}(2\pi)$ satisfy $\sg(\xi_{\aa,\bb,0}(2\pi)) \subset \C\bs\U$.

    \item[II]  The second region $\cR_2 \equiv \{(\aa,\bb)|\aa\geq\bb>0,\aa \leq \frac{9}{4}\bb^2, \aa \geq3\bb-1,\aa\leq 1\}$.

    In $\cR_2$,  $\eta_1$ and $\eta_2$ satisfy that
        \bea
        \eta_1 = \aa -1 +\sqrt{9\bb^2 - 4\aa}\leq  0, \quad
        \eta_2 = \aa -1 -\sqrt{9\bb^2 - 4\aa} \leq 0.\lb{R2.eta}
        \eea
    Then we have that $\lm_{i,\pm}\in \sqrt{-1}\R$ for $1\leq i \leq 4$ and the four characteristic multipliers of the matrix $\xi_{\aa,\bb,0}(2\pi)$ satisfy $\sg(\xi_{\aa,\bb,0}(2\pi)) \subset \U$.

    \item[III] The third region $\cR_3 \equiv \{(\aa,\bb)| \aa\geq\bb>0, \aa < 3\bb-1\}\cup\{(\aa,\bb)|\aa= 3\bb-1, \aa>1\}$.

    Note that $\aa < 3\bb-1$ and $\aa\geq \bb 0$ yield $9\bb^2>\aa^2+2\aa+1 > 4\aa$ and $\bb \geq \frac{1}{2}$.
    In $\cR_3$,  $\eta_1$ and $\eta_2$ satisfy that
    \bea
    \eta_1 = \aa -1 +\sqrt{9\bb^2 - 4\aa}> 0, \quad
    \eta_2 =\aa -1 -\sqrt{9\bb^2 - 4\aa} \leq 0. \lb{R3.eta}
    \eea

    Therefore, $\lm_{1,\pm}\in\R\bs\{0\}$ and $\lm_{2,\pm}\in  \sqrt{-1}\R$. Then the four characteristic multipliers of the matrix $\xi_{\aa,\bb,0}(2\pi)$ satisfy $\rho_{1,\pm}\in \R^+$ and $\rho_{2,\pm}\in \U$.

    \item[IV] The fourth region $\cR_4 \equiv\{(\aa,\bb)|\aa>\bb>0, \aa \leq \frac{9}{4}\bb^2, \aa>3\bb-1,\aa>1\}$.

    In $\cR_4$,
     $\eta_1$ and $\eta_2$ satisfy that
    \bea
    \eta_1 = \aa -1 +\sqrt{9\bb^2 - 4\aa}> 0,\quad
    \eta_2 =\aa -1 -\sqrt{9\bb^2 - 4\aa} >0.
    \eea
    Therefore, we have that
    $\lm_{i,\pm} \in \R$. Then the four characteristic multipliers of the matrix $\xi_{\aa,\bb,0}(2\pi)$ satisfy $\sg(\xi_{\aa,\bb,0}(2\pi)) \subset \R^+\bs\{1\} $.
\end{description}

Note that the $\cR_1$ and $\cR_4$ are hyperbolic regions.
We will discuss the linear stability of the essential part in $\cR_2$ in Section \ref{Stab.R2} and the linear stability in $\cR_3$ in Section \ref{Stab.R3}.

\subsection{Stability in region $\cR_2$}\lb{Stab.R2}
In $\cR_2$,  by  (\ref{R2.eta}), $\eta_1 \leq 0$ and $\eta_2 \leq 0$.
Then we have that  $\lm_{i,\pm}\in \sqrt{-1}\R$.
The four characteristic multipliers of the matrix $\xi_{\aa,\bb,0}(2\pi)$ can be written as
\be
\rho_{i,\pm} (\aa,\bb) = e^{2\pi \lm_{i,\pm}} = e^{\pm 2\pi\sqrt{-1} \th_i(\aa,\bb)}, \; \mbox{for} \; i = 1,2, \lb{rho}
\ee
where $\th_{i}(\aa,\bb)$ are given by
\bea
\th_1(\aa,\bb) = \sqrt{1-\aa -\sqrt{9\bb^2 - 4\aa}},\quad
\th_2(\aa,\bb) = \sqrt{1-\aa +\sqrt{9\bb^2 - 4\aa}}. \lb{eqn:th1.th2}
\eea
To determine the maximum and minimum of
$\th_1(\aa,\bb)$ and $\th_2(\aa,\bb)$ in $\cR_2$, by direct computations, we have
\bea
\frac{\pt \th_1}{\pt \aa} &=& \frac{1}{\sqrt{1-\aa -\sqrt{9\bb^2 - 4\aa}}}\l(-1+\frac{2}{\sqrt{9\bb^2-4\aa}}\r)  > 0, \lb{th1.pt.aa}\\
\frac{\pt \th_1}{\pt \bb}&=& \frac{1}{\sqrt{1-\aa -\sqrt{9\bb^2 - 4\aa}}}\l(\frac{-9\bb}{\sqrt{9\bb^2-4\aa}}\r)<0, \lb{min.th1}\\
\frac{\pt \th_2}{\pt \aa}&=& \frac{1}{\sqrt{1-\aa +\sqrt{9\bb^2 - 4\aa}}}\l(-1-\frac{2}{\sqrt{9\bb^2-4\aa}}\r) <0,\lb{th2.pt.aa}\\
\frac{\pt \th_2}{\pt \bb}&=&\frac{1}{\sqrt{1-\aa +\sqrt{9\bb^2 - 4\aa}}}\l(\frac{9\bb}{\sqrt{9\bb^2-4\aa}}\r) >0. \lb{min.th2}
\eea
Note (\ref{th1.pt.aa}) holds because $\sqrt{9\bb^2-4\aa} < 2$ in $\cR_2$ by direct computations.
Therefore the maximum and minimum of $\th_{1}$ and $\th_{2}$ in $\cR_2$ are attained at the boundary of $\cR_2$ because
\bea
\max_{(\aa,\bb)\in \cR_2}{\th_1} = \th_1|_{\aa=4/9,\bb=4/9}= \sqrt{5}/3, &\quad&
\min_{(\aa,\bb)\in \cR_2}{\th_1} =\th_1|_{\aa=3\bb-1}  = 0,\lb{eqn:th1.max.min}\\
\max_{(\aa,\bb)\in \cR_2}{\th_2} =\th_2|_{\aa=1/2,\bb=1/2} =  1, &\quad&
\min_{(\aa,\bb)\in \cR_2}{\th_2} =\th_2|_{\aa=1,\bb=2/3} =  0.\lb{eqn:th2.max.min}
\eea
Therefore, in $\cR_2$, we have that the $\th_1 \in [0,\frac{\sqrt{5}}{3}]$
and $\th_2 \in [0,1]$ and the eigenvalues of $\xi_{\aa,\bb,0}(2\pi)$
are given by $\sigma(\xi_{\aa,\bb,0}(2\pi)) = \{e^{2\pi\sqrt{-1}\th_1}, e^{-2\pi\sqrt{-1}\th_1},e^{2\pi\sqrt{-1}\th_2}, e^{-2\pi\sqrt{-1}\th_2}\}$. If $\aa \neq \frac{9}{4}\bb^2$ and $\aa \neq 3\bb-1$, $\th_1 \neq \th_2$. Therefore, we have that $\xi_{\aa,\bb,0}(2\pi) \approx R(\th_1) \diamond R(\th_2)$.

By (\ref{eqn:th1.max.min}) and (\ref{eqn:th2.max.min}), $\frac{1}{2}$ is in both the range of $\th_1$ and $\th_2$. By direct computations, we define the $-1$-degenerate line $\cR_{2,\frac{1}{2}}^*$ in $\cR_2$ by
\be
\cR_{2,\frac{1}{2}}^* = \{(\aa,\bb)\in\ol{\cR_2}| \aa= -\frac{5}{4}+\sqrt{9\bb^2+1}\}. \lb{eqn:seg.l}
\ee
We further define two sub-regions of $\cR_2$ by
\bea
\cR_{2,\frac{1}{2}}^- &=&  \{(\aa,\bb)\in\ol{\cR_2}| \aa > -\frac{5}{4}+\sqrt{9\bb^2+1}\}; \\
\cR_{2,\frac{1}{2}}^+ &=&  \{(\aa,\bb)\in\ol{\cR_2}| \aa < -\frac{5}{4}+\sqrt{9\bb^2+1}\}.
\eea
When $(\aa,\bb) \in \cR_{2,\frac{1}{2}}^*$, $-1\in \sigma(\xi_{\aa,\bb,0}(2\pi))$.
Furthermore, when $\bb \in [\frac{5+\sqrt{97}}{32}, \frac{\sqrt{3}}{3}]$,
$\th_1 \equiv \frac{1}{2}$ and $\th_2 \in [\frac{\sqrt{23-\sqrt{97}}}{4},1/2]$.
Therefore, $\sigma(\xi_{\aa,\bb,0}(2\pi)) = \{-1,-1,e^{2\pi\sqrt{-1}\th_2}, e^{-2\pi\sqrt{-1}\th_2}\}$.
When
$\bb \in [\frac{\sqrt{3}}{3},\frac{5}{8}]$,
$\th_2 \equiv  \frac{1}{2}$ and $\th_1 \in [0,1/2]$,
$\sigma(\xi_{\aa,\bb,0}(2\pi)) = \{e^{2\pi\sqrt{-1}\th_1}, e^{-2\pi\sqrt{-1}\th_1},-1, -1\}$.
Especially, when $(\aa,\bb)= (\frac{3}{4},\frac{\sqrt{3}}{3})$, $\th_1 = \th_2 = \frac{1}{2}$.
Then $\sigma(\xi_{\frac{3}{4},\frac{\sqrt{3}}{3},0}(2\pi)) = \{-1,-1,-1, -1\}$.

For the boundary of $\cR_2$,
along the segment $\{(\aa,\bb)\in \ol{\cR_2}|\aa= 3\bb-1\}$, we have that $\th_1 = 0$ and $\th_2 > 0$.
Then $\sigma(\xi_{\aa,\bb,0}(2\pi)) = \{1, 1,e^{2\pi\sqrt{-1}\th_2}, e^{-2\pi\sqrt{-1}\th_2}\}$ where $\th_2\in[0,1]$.
Especially,
when $(\aa,\bb)= (\frac{1}{2},\frac{1}{2})\in \pt \cR_2$, $\th_1 =0$ and  $\th_2 = 1$.
Then $\sg(\xi_{1/2,1/2,0}(2\pi)) = \{1,1,1,1\}$.
When $(\aa,\bb)= (\frac{7}{8},\frac{5}{8} )\in \pt \cR_2$, we have $\th_1 =0$ and  $\th_2 = \frac{1}{2} $.
Then $\sg(\xi_{1,2/3,0}(2\pi)) = \{1,1,-1,-1\}$.
When $(\aa,\bb) = (1,\frac{2}{3}) \in \pt \cR_2$, we have $\th_1 = \th_2 = 0$. Then $\sg(\xi_{1,2/3,0}(2\pi)) = \{1,1,1,1\}$.

\subsection{Stability in region $\cR_3$}\lb{Stab.R3}
When $(\aa,\bb) \in \cR_3$, by (\ref{rho}), the characteristic multipliers are given by
\be
\rho_{1,\pm} = e^{\pm 2\pi\lm_1}\in \R^+\bs\{1\}, \quad
\rho_{2,\pm} = e^{\pm 2\pi \sqrt{-1}\th} \in \U,
\ee
where $\th$ is given by
\be
\th = \sqrt{1-\aa+\sqrt{9\bb^2-4\aa}}.
\ee
We compute the derivatives of $\th$ with respect to $\aa$ and $\bb$ in $\cR_3$ and obtain that
\bea
\frac{\partial \th}{\pt \aa}  &=& \frac{1}{2\sqrt{1-\aa+\sqrt{9\bb^2-4\aa}}}\l(-1-\frac{2}{\sqrt{9\bb^2-4\aa}}\r) < 0,\\
\frac{\partial \th}{\pt \bb}  &=& \frac{1}{2\sqrt{1-\aa+\sqrt{9\bb^2-4\aa}}}\l(\frac{9\bb}{\sqrt{9\bb^2-4\aa}}\r)>0.
\eea
Therefore, the maximum and minimum of $\th$ must be attained at the boundary,
i.e., $\aa= \bb$ and $\aa = 3\bb-1$, or $(\aa,\bb)=(\infty,\infty)$.
By direct computations, along the line $\aa =\bb$, $\th$ tends to infinity when $\bb$ and $\aa$ tend to infinity.
Furthermore, when $\aa = 3\bb-1$ and $\bb \geq \frac{2}{3}$, we have $\th \equiv 0$.
Then range of $\th$ in $\cR_3$ is $[0,\infty)$.

For $(\aa,\bb) \in \cR_3 $ and $\th \geq 0$, we use $\aa_{\th} $ and $\bb_{\th}$ to denote the values of $(\aa,\bb)$ such that
\be
\th = \sqrt{1-\aa_{\th}+\sqrt{9\bb_{\th}^2-4\aa_{\th}}}, \lb{thaabb}
\ee
and hence
\be
\aa_{\th}(\bb) = -(\th^2+1) + \sqrt{9\bb_{\th}^2+4\th^2}. \lb{aath}
\ee
Then define the subsets $\cR_{3,n}^*$ and $\cR_{3,n+\frac{1}{2}}^*$ of $\cR_3$ by setting $\th = n$ and $\th = n+\frac{1}{2}$ in $\aa_{\th}(\bb)$.
\bea
\cR_{3,0}^* &\equiv& \{(\aa,\bb)\in \ol{\cR_{3}}| \aa =\aa_{0}(\bb) =  3\bb-1\},\lb{eqn:R30*}\\
\cR_{3,n+\frac{1}{2}}^* &\equiv& \l\{(\aa,\bb)\in \ol{\cR_{3}}|\aa =\aa_{n+\frac{1}{2}}(\bb) = -\l(n+\frac{1}{2}\r)^2-1+\sqrt{9\bb^2+4\l(n+\frac{1}{2}\r)^2}\r\},\lb{eqn:R3n+1/2*}\\
\cR_{3,n}^* &\equiv& \{(\aa,\bb)\in \ol{\cR_{3}}|\aa =\aa_{n}(\bb) =-n^2-1+ \sqrt{9\bb^2+4n^2} \}.
\eea
Note that $\cR_{3,0}^*\cap \cR_{3,\frac{1}{2}}^* = (\frac{7}{8},\frac{5}{8})$ and $(\frac{1}{2},\frac{1}{2}) \in\cR_{3,0}^* \cap \cR_{3,1}^*$.
The regions between $\cR_{3,n}^*$ and $\cR_{3,n+\frac{1}{2}}^*$ for $n\in\N_0$ are defined by
\bea
\cR_{3,0}^+ &\equiv& \{(\aa,\bb)\in \ol{\cR_{3}}| \aa_{\frac{1}{2}}(\bb)< \aa <\aa_{0}(\bb) \},\lb{eqn:R30+}\\
\cR_{3,1}^- &\equiv& \{(\aa,\bb)\in \ol{\cR_{3}}| \aa_{1}(\bb)< \aa < \min\{\aa_{0}(\bb),\aa_{\frac{1}{2}}(\bb)\}\},\\
\cR_{3,n}^- &\equiv& \{(\aa,\bb)\in \ol{\cR_{3}}| \max\{\aa_{n}(\bb), \bb\}  < \aa < \aa_{n-\frac{1}{2}}(\bb) \},\\
\cR_{3,n}^+ &\equiv& \{(\aa,\bb)\in \ol{\cR_{3}}|\max\{\aa_{n+\frac{1}{2}}(\bb), \bb\} < \aa < \aa_{n}(\bb) \}.\lb{eqn:R3n+}
\eea

By (\ref{thaabb}) and direct computations, for given $(\aa,\bb)\in\cR_3$, $\th(\aa,\bb)$ is the function of $(\aa,\bb)$. Then for the given positive $\th_1\neq \th_2$, $\aa_{\th_1}(\bb)$ cannot intersect with $\aa_{\th_2}(\bb)$. Then $\cR_{3,n}^+ $ and $\cR_{3,n}^- $ are pairwise disjoint for all $n\in \N_0$.

Note that $\lm_1 \in \R^+\bs\{1\}$. Then eigenvalues of $\xi_{\aa,\bb,0}(2\pi)$ for $(\aa,\bb) \in \cR_3$ can be given as followings.
\begin{enumerate}[label=(\roman*)]
    \item When $(\aa,\bb)\in \cR_{3,0}^*$, the case of $\bb\in(\frac{1}{2},\frac{2}{3})$ has been discussed in Section \ref{Stab.R2}.
    For $\bb > \frac{2}{3}$, we always have that $\th =  0$
    and $\sg(\xi_{\aa,\bb,0}(2\pi)) = \{1,1, e^{2\pi\lm_1}, e^{- 2\pi\lm_1}\}$.

    \item Let $i\in\N_0$. When $(\aa,\bb)\in \cR_{3,i}^*$, it yields $\th(\aa,\bb)=i$.
    Then $\rho_{2,\pm}=e^{\pm 2\pi\sqrt{-1} } = 1$ $\sg(\xi_{\aa,\bb,0}(2\pi)) = \{1, 1, e^{2\pi\lm_1}, e^{- 2\pi\lm_1}\}$.

    \item Let $i\in\N_0$. When $(\aa,\bb) \in \cR_{3,i}^+$, $\th(\aa,\bb)\in (i,  i+\frac{1}{2})$.
    Therefore $\rho_{2,+}(\aa,\bb)=e^{2\pi \sqrt{-1}\th(\aa,\bb)}$ on upper semi-unit circle in the complex plane $\C$.
    Correspondingly $\rho_{2,-}(\aa,\bb)=e^{-2\pi \sqrt{-1}\th(\aa,\bb)}$ on lower semi-unit circle in $\C$.
    Then $\sg(\xi_{\aa,\bb,0}(2\pi)) = \{e^{2\pi\lm_1}, e^{-2\pi\lm_1}, e^{2\pi \sqrt{-1}\th(\aa,\bb)}, e^{- 2\pi \sqrt{-1}\th(\aa,\bb)}\}$ with $\th\in(0, \pi)$.

    \item Let $i\in\N_0$.
    When $(\aa,\bb)\in \cR_{3,i+\frac{1}{2}}^*$, it yields $\th(\aa,\bb)=i+\frac{1}{2}$.
    Then $\rho_{2,\pm}=e^{\pm \sqrt{-1} \pi} = -1$ and $\sg(\xi_{\aa,\bb,0}(2\pi)) = \{-1, -1, e^{2\pi\lm_1}, e^{- 2\pi\lm_1}\}$.

    \item Let $i\in\N$.
    When $(\aa,\bb) \in \cR_{3,i}^-$, the angle $\th(\bb) \in (i-\frac{1}{2},i)$.
    Thus $\rho_{2,+}(\aa,\bb)=e^{2\pi \sqrt{-1}\th(\aa,\bb)}$ on the lower semi-unit circle in $\C$.
    Correspondingly, $\rho_{2,-}(\aa,\bb)=e^{-2\pi \sqrt{-1}\th(\aa,\bb)}$ on the upper semi-unit circle in $\C$.
    Then $\sg(\xi_{\aa,\bb,0}(2\pi)) = \{e^{2\pi \sqrt{-1}\th(\aa,\bb)}, e^{- 2\pi \sqrt{-1}\th(\aa,\bb)} e^{2\pi\lm_1}, e^{- 2\pi\lm_1}\}$ with $\th\in(\pi, 2\pi)$.
\end{enumerate}

\subsection{$\pm 1$-indices when $e = 0$}\lb{sec:e0.index}
First define an orthogonal basis $\{f_{0,1}, f_{0,2}, f_{n,1},f_{n,2},f_{n,3},f_{n,4}| n\in \N\}$ of $\ol{D}(1,2\pi)$ in (\ref{A2.10}) by
\bea
f_{0,1} = R(t)\bpm 1 \\ 0 \epm, \quad
f_{0,2} = R(t)\bpm 0 \\ 1 \epm,
\eea
and for $n \in \N$
\bea
f_{n,1} = R(t)\bpm \cos nt \\ 0 \epm, \quad
f_{n,2} = R(t)\bpm 0 \\ \cos nt \epm,\\
f_{n,3} = R(t)\bpm \sin nt \\ 0 \epm, \quad
f_{n,4} = R(t)\bpm 0 \\ \sin nt \epm.
\eea
 By (\ref{cA}) and $\frac{\d R(t)}{\d t} = JR(t)$, $\cA(\aa,\bb, 0)f_{n,1}$ is given by
\bea
\cA(\aa,\bb,0)f_{n,1} &=& \l(-\frac{\d^2}{\d t^2}I_2-I_2+R(t)K_{\aa,\bb,0}(t)R(t)^T\r)
R(t)
\bpm \cos nt \\ 0 \epm \nn \\
&=& R(t)
\l((n^2+1+\aa+3\bb)\cos nt, 2n\sin nt\r)^T \nn \\
&=& (n^2+1+\aa+3\bb)f_{n,1}+2nf_{n,4}.
\eea
Similarly, for $n \in \N_0$, it yields that
\bea
\bpm \cA(\aa,\bb,0) & 0 \\ 0 & \cA(\aa,\bb,0)\epm
\bpm f_{0,1} \\ f_{0,2} \epm &=&
B_0
\bpm f_{0,1} \\ f_{0,2} \epm ,\lb{A0}\\
\bpm \cA(\aa,\bb,0) & 0 \\ 0 & \cA(\aa,\bb,0) \epm
\bpm f_{n,1} \\ f_{n,4} \epm&=&
B_n
\bpm f_{n,1} \\ f_{n,4} \epm,\lb{A14}\\
\bpm \cA(\aa,\bb,0) & 0 \\ 0 & \cA(\aa,\bb,0) \epm
\bpm f_{n,2} \\ f_{n,3} \epm &=&
\bar{B}_n
\bpm f_{n,2} \\ f_{n,3} \epm,\lb{A23}
\eea
where $B_n$ and $\bar{B}_n$ are given by
\bea
B_n &=&  \bpm n^2+1+\aa+3\bb & 2n \\ 2n & n^2+1+\aa-3\bb \epm, \mbox{ for } n\in \N_0; \\
\bar{B}_n  &=&  \bpm n^2+1+\aa+3\bb & -2n \\ -2n & n^2+1+\aa-3\bb \epm, \mbox{ for } n\in \N.
\eea
The characteristic polynomials of $B_n$ and $\bar{B}_n$ receptively are denoted by $p_n(\lm)$ and $\bar{p}_n(\lm)$ which are
\be
p_n(\lm) = \bar{p}_n(\lm) =
\lm^2 - (2n^2 + 2\aa+ 2)\lm + (n^2+1+\aa)^2-9\bb^2-4n^2.
\ee
Let $i\in \N$. Fixing $\aa(\bb) = -i^2-1 + \sqrt{9\bb^2+4i^2}$,
$p_n(0) = \bar{p}_n(0) = 0$ if and only if $n = i$.

Note that when $\aa > 3\bb-1 $, i.e., $(\aa,\bb)\in\cR_1\cup \cR_2 \cup \cR_4$, $(n^2+1+\aa)^2-9\bb^2-4n^2 > 0 $ for all $n\geq 0$.
Therefore, when $(\aa,\bb)\in\cR_1\cup \cR_2 \cup \cR_4$, $(n^2+1+\aa)^2-9\bb^2-4n^2 > 0 $ and then
\be i_1(\xi_{\aa,\bb,0}) = 0, \quad \nu_1(\xi_{\aa,\bb,0}) = 0. \ee

We define that $G(n) = n^2-\sqrt{9\bb^2+4n^2}$ for given $\bb > 1/2$ and $n\geq 1$. Note that $\frac{\partial G}{\pt n} = 2n -\frac{4n}{\sqrt{9\bb^2+4n^2}} > 0$ because $\sqrt{9\bb^2+4n^2} > \frac{5}{2}$. Therefore, if $n_1<n_2$, we have that $G(n_1) < G(n_2)$, i.e.,
\be
n_1^2-\sqrt{9\bb^2+4n_1^2} < n_2^2-\sqrt{9\bb^2+4n_2^2}.
\ee
This yields that $p_n(0) = \bar{p}_n(0) < 0$ if $n < i$ , and $p_n(0) = \bar{p}_n(0) > 0$ if $n > i$ . Thus both $B_i$
and $\bar{B}_i$ have one zero and one positive eigenvalues when $n = i$; both $B_i$
and $\bar{B}_i$ with $n < i$ have one
negative and one positive eigenvalues; both $B_i$
and $\bar{B}_i$ with $n > i$ have two positive
eigenvalues.
Notice that $B_0$ has two positive eigenvalues when $\aa>3\bb+1$;  $B_0$ has one positive and one zero eigenvalue when $\aa=3\bb+1$;   $B_0$ has one positive and one negative eigenvalue when $\aa<3\bb+1$.
Therefore we have $i_1(\xi_{\aa,\bb,0}) = 2i+1$ and $\nu_1(\xi_{\aa,\bb,0}) = 2$ when $(\aa,\bb)\in \cR^*_{3,i}$.

For $(\aa, \bb) = (\frac{1}{2},\frac{1}{2})$, $B_0$, $B_1$ and $\bar{B}_1$ all possess one dimensional degenerate space and one dimension positive definite eigenspace.

When $(\aa, \bb) \in \cR_{3,i}^+ \cup \cR^*_{3,i+\frac{1}{2}}\cup \cR_{3,i+1}^-$ and $i\in\N_0$, then $p_n(0) = \bar{p}_n(0) \neq 0$. Similarly to the above argument, we have $p_n(0) = \bar{p}_n(0) < 0$ if $n \leq i$, and $p_n(0) = \bar{p}_n(0) > 0$ if $n > i$ .
Thus both $B_n$ and $\bar{B}_n$ with $n \leq i$ have a negative and a positive eigenvalues; both $B_n$ and $\bar{B}_n$ with $n > i$ have two positive eigenvalues. Notice that $B_0$ has a negative and a positive eigenvalues, we have $i_1(\xi_{\aa,\bb,0}) = 2i + 1$ and $\nu_1(\xi_{\aa,\bb,0}) = 0$ for $(\aa, \bb) \in \cR_{3,i}^+ \cup \cR^*_{3,i+\frac{1}{2}}\cup \cR_{3,i+1}^-$.
Therefore, we have (\ref{index.e=0}) and (\ref{eqn:null.e=0}) hold for $\aa \geq \bb > 0$ and $n\in \N_0$.

Following the same procedure, to compute $i_{-1}(\xi_{\aa,\bb,0})$, we define $\bar{f}_{n,i}$ by
\bea
\bar{f}_{n,1} = R(t) \bpm  \cos (n+1/2)t \\ 0 \epm, \quad
\bar{f}_{n,2} = R(t) \bpm 0 \\ \cos (n+1/2)t\epm,\\
\bar{f}_{n,3} = R(t) \bpm  \sin (n+1/2)t \\ 0 \epm, \quad
\bar{f}_{n,4} = R(t) \bpm 0 \\ \sin (n+1/2)t \epm,
\eea
for $n \in \N_0$.
Then we have that
\bea
\bpm \cA(\aa,\bb,0) & 0 \\ 0 & \cA(\aa,\bb,0) \epm
\bpm \bar{f}_{n,1} \\ \bar{f}_{n,4} \epm &=&
\bpm (n+\frac{1}{2})^2+1+\aa+3\bb & 2(n+\frac{1}{2}) \nn\\ 2(n+\frac{1}{2}) & (n+\frac{1}{2})^2+1+\aa-3\bb \epm
\bpm \bar{f}_{n,1} \\ \bar{f}_{n,4} \epm,\\
\bpm \cA(\aa,\bb,0) & 0 \\ 0 & \cA(\aa,\bb,0) \epm
\bpm \bar{f}_{n,2} \\ \bar{f}_{n,3} \epm &=&
\bpm (n+\frac{1}{2})^2+1+\aa+3\bb & -2(n+\frac{1}{2}) \nn\\ -2(n+\frac{1}{2}) & (n+\frac{1}{2})^2+1+\aa-3\bb \epm
\bpm \bar{f}_{n,2} \\ \bar{f}_{n,3} \epm .
\eea
By arguments similar as the $1$-index, we compute the
eigenvalues of $\cA(\aa,\bb, 0)$ in the domain $\ol{D}(-1, 2\pi)$, then the $-1$-indices of $\xi_{\aa,\bb, 0}$ and the nullity of $\xi_{\aa,\bb, 0}(2\pi)$.

Especially when $(\aa,\bb) \in \cR_{2,\frac{1}{2}} \bigcup(\cup_{n=0}^{\infty}\cR^*_{3,n+\frac{1}{2}})$, $\cA(\aa,\bb,0)$ has eigenvalue $-1$ with geometric multiplicity $2$. Thus
\be
\nu_{-1}(\xi_{\aa,\bb,0}(2\pi)) = 2.\lb{nu.-1.xi}
\ee

When $(\aa,\bb)\in \cR_1 \cup\cR_2\cup \cR_4$, as the discussion about the $1$-index we have that
\bea
i_{-1}(\xi_{\aa,\bb,0}) &=& \begin{cases}
    0 & \mbox{if} \; (\aa,\bb)\in \cR_1 \cup \cR_{2,\frac{1}{2}}^-\cup \cR_{2,\frac{1}{2}}^* \cup \cR_4; \\
    2 & \mbox{if} \; (\aa,\bb)\in \cR_{2,\frac{1}{2}}^+; \\
\end{cases}	\\
\nu_{-1}(\xi_{\aa,\bb,0}) &=&
\begin{cases}
	2, & \mbox{if} \;  (\aa,\bb)\in \cR_{2,\frac{1}{2}}^*;\\
	0  & \mbox{if} \;  (\aa,\bb)\in (\cR_1 \cup\cR_2\cup \cR_4) \bs \cR_{2,\frac{1}{2}}^*.
\end{cases}
\eea

When $(\aa,\bb)\in \cR_3$, $\xi_{\aa,\bb,0}(2\pi)$ possesses one pair of positive hyperbolic characteristic multipliers
$\rho_{1,\pm}(\bb)$ and one pair of elliptic characteristic multipliers $\rho_{2,\pm}(\bb)$ on the unit circle. Then
\be
\xi_{\aa,\bb,0}(2\pi) \approx D(e^{2\pi\sqrt{\aa_1(\bb)}}) \diamond M(2\pi), \lb{xi.nf.bb.0}
\ee
for some matrix $M(2\pi) \in \Sp(2)$. Due to the nullity of $M(2\pi)$ is even,  $M(2\pi)$ must be the normal form $N_1(\pm 1,b)$ with $b = 0$. Then we have for $n \in \N$
\be
M(2\pi) = \begin{cases}
	I_2, \quad \mbox{if} \quad
    (\aa,\bb) \in \cR_{3,n}^*; \\
	-I_2, \quad  \mbox{if} \quad (\aa,\bb) \in \cR_{3,n+\frac{1}{2}}^*;\\
	R(2\pi\th(\bb)) \; \mbox{or} \; -R(2\pi\th(\bb)),\quad \mbox{if} \quad (\aa,\bb)\in \cR_{3,n}^{\pm}.
\end{cases}
\ee

Note that there exists a path $M(t) \in {\cal P}_{2\pi}(2)$ which connecting $M(0) = I_2$ to $M(2\pi)$ such that the path $\xi_{\aa, \bb,0}(2\pi)$ is homotopic to the path $D(e^{t\sqrt{\aa_1(\bb)}}) \diamond M(t)$ for $t \in [0, 2\pi]$.

By the properties of splitting numbers \eqref{eqn:split}, for $(\aa,\bb) \in \cR_{3,n}^+$ and $\om = -1$, we have that
\bea
i_{-1}(\xi_{\aa,\bb,0})&=& i_{1}(\xi_{\aa,\bb,0})
+S^+_{\xi_{\aa,\bb,0}(2\pi)}(1)
-S^-_{\xi_{\aa,\bb,0}(2\pi)}(e^{\sqrt{-1}2\pi(\th-n)})\nn \\
&&+S^+_{\xi_{\aa,\bb,0}(2\pi)}(e^{\sqrt{-1}2\pi(\th-n)})
-S^-_{\xi_{\aa,\bb,0}(2\pi)}(-1)\nn\\
&=&i_{1}(\xi_{\aa,\bb,0})
-S^-_{M(2\pi)}(e^{\sqrt{-1}2\pi(\th-n)})
+S^+_{M(2\pi)}(e^{\sqrt{-1}2\pi(\th-n)})\nn\\
&=&
\begin{cases}
	i_{1}(\xi_{\aa,\bb,0}) - 1 = 2n, \quad \mbox{if} \quad M(2\pi) = R(2\pi\th(\bb)); \\
	i_{1}(\xi_{\aa,\bb,0}) + 1 = 2n + 2, \quad \mbox{if} \quad M(2\pi) = R(-2\pi\th(\bb)).
\end{cases}
\eea

When $(\aa,\bb) \in \cR_{3,n+1}^-$, we have that
\bea
i_{-1}(\xi_{\aa,\bb,0})&=& i_{1}(\xi_{\aa,\bb,0})
+S^+_{\xi_{\aa,\bb,0}(2\pi)}(1)
-S^-_{\xi_{\aa,\bb,0}(2\pi)}(e^{\sqrt{-1}2\pi(n+1-\th)}) \nn\\
&&+S^+_{\xi_{\aa,\bb,0}(2\pi)}(e^{\sqrt{-1}2\pi(n+1-\th)})
-S^-_{\xi_{\aa,\bb,0}(2\pi)}(-1) \nn\\
&=&i_{1}(\xi_{\aa,\bb,0})
-S^-_{M(2\pi)}(e^{\sqrt{-1}2\pi(n+1-\th)})
+S^+_{M(2\pi)}(e^{\sqrt{-1}2\pi(n+1-\th)}) \nn\\
&=&
\begin{cases}
	i_{1}(\xi_{\aa,\bb,0}) + 1 = 2n + 2, \quad \mbox{if} \quad M(2\pi) = R(2\pi\th(\bb)); \\
	i_{1}(\xi_{\aa,\bb,0}) - 1 = 2n , \quad \mbox{if} \quad M(2\pi) = R(-2\pi\th(\bb)).
\end{cases}\lb{itra.-1.xi}
\eea

If $M(2\pi) = R(-2\pi\th(\bb))$ for $(\aa,\bb) \in \cR_{3,n}^+$,
we have that $i_{-1}(\xi_{\aa,\bb,0}) = 2n+2$. By (\ref{nu.-1.xi}) and the non-decreasing of $i_{-1}(\xi_{\aa,\bb,0})$ with respect to $\bb$, we must have $i_{-1}(\xi_{\aa,\hat{\bb}++n +1/2+\ep,0}) = i_{-1}(\xi_{\aa,\hat{\bb}+n+1/2,0}) + \nu_{-1}(\xi_{\aa,\hat{\bb}+n+1/2,0}) \geq 2n+4$
which contradicts (\ref{itra.-1.xi}). Similarly, we cannot have $M(2\pi) = R(-2\pi\th(\bb))$ when $ (\aa,\bb) \in \cR_{3,n+1}^-$. Thus, we must have that $M(2\pi) = R(2\pi\th(\bb))$ when $(\aa,\bb) \in \cR_{3,0}^+ \bigcup\l(\cup_{n = 1}^{\infty}\cR_{3,n}^-\cup\cR_{3,n}^+ \r)$.

Therefore, we have that for $\aa\geq \bb > 0$ and $e = 0$, the $i_{-1}(\xi_{\aa,\bb,0})$ and $\nu_{-1}(\xi_{\aa,\bb,0})$ are given by  (\ref{eqn:intro.-1.index}) and (\ref{null.-1.index}).

\setcounter{equation}{0}
\section{The $\om$-Index Properties}
In this section,  we first discuss the special case $\bb = 0$ of the index and then the general properties of index.

At first, we consider the case of $\aa = \bb = 0$. By (\ref{cA}), $\cA(0,0,e)$ is given by
\be
\cA(0,0,e) =-\frac{\d^2}{\d t^2}I_2  -I_2 + \frac{1}{1+e\cos t}I_2.
\ee
The operator $\cA(0,0,e)$ has been discussed in Lemma 4.1 of \cite{ZL15ARMA}. We paraphrase their results in our notations. For details, readers may refer to \cite{ZL15ARMA} for details.

\begin{lemma}\lb{lem:00e.pos}
	For $\aa= \bb = 0$ and $0 \leq e < 1$, there holds

	(i) $\cA(0,0,e)$ are non-negative definite for the $\om = 1$ boundary condition and
	\be
	\ker \cA(0,0,e) = \l\{(c_1(1+e\cos t),c_2(1+e\cos t))^T
	|c_1,c_2 \in \C  \r\}.\lb{eqn:ker.00e}
	\ee

	(ii) $\cA(0,0,e)$ are positive definite for the $\om \neq 1$ boundary condition.
\end{lemma}

Since that $1+\aa > 1$ for $\aa> \bb = 0$, by Proposition 2 of \cite{HO}, we have Proposition \ref{prop:posiv.def}.
\begin{proposition}\lb{prop:posiv.def}
	For $\aa> \bb = 0$, the operator $\cA(\aa,0,e)$ is positive definite for any $\om$ boundary conditions.
\end{proposition}

When $\aa,\bb >0$, the operator $\cA(\aa, \bb,e)$ can be written as
\bea
\cA(\aa, \bb,e)
=  \bb\l(\frac{\cA(\aa,0, e)}{\bb} +\frac{ 3 S(t)}{1+e\cos t}\r)= \bb \bar{\cA}(\aa,\bb, e),
\eea
where $\bar{\cA}(\aa,\bb, e)  = \frac{\cA(\aa,0, e)}{\bb} +\frac{ 3 S(t)}{1+e\cos t}$.
Note that for $\aa,\bb >0$ and $e\in (0,1)$, we have that
\be
\phi_{\om}(\cA(\aa , \bb,e)) = \phi_{\om} (\bar{\cA}(\aa,\bb, e)), \quad
\nu_{\om}(\cA(\aa , \bb,e)) = \nu_{\om} (\bar{\cA}(\aa,\bb, e)).
\ee

We here follow the proof of Proposition 3.5 of \cite{HLO} and obtain following monotonic of the index and nullity of $\cA(\aa,\bb, e)$.
\begin{lemma}\lb{lem:A.mono}
	(i) The $\om$-Maslov index  $\phi_{\om}(\cA(\aa,\bb, e))$ and the corresponding index $i_{\om}(\xi_{\aa,\bb,e})$ are non-increasing in $\aa \in (0,\infty)$ and they are non-decreasing in $\bb \in (0,\infty)$.

	(ii) The sum of $\om$-Maslov index and nullity  $\phi_{\om}(\cA(\aa,\bb, e)) + \nu_{\om}(\cA(\aa,\bb, e))$ and the corresponding $i_{\om}(\xi_{\aa,\bb,e})+ \nu_{\om}(\xi_{\aa,\bb,e})$ are non-increasing in $\aa \in (0,\infty)$ and are non-decreasing in $\bb \in (0,\infty)$.
\end{lemma}

\begin{proof}
	Note that for given $\bb_0>0$ and $e_0\in [0,1)$, we have that in $D(2\pi, \om)$ for $\aa_2>\aa_1\geq 0$, $e_0 \in [0,1)$ and any $\om\in \U$,
	\be
	\cA(\aa_2, \bb_0,e_0) - \cA(\aa_1, \bb_0,e_0) = \frac{\aa_2-\aa_1}{(1+e_0\cos t)} I_2 > 0. \lb{eqn:4.5}
	\ee
	Therefore, we have that the operator is increasing respect to $\aa \in (0,\infty)$.

	When $\bb_2> \bb_1 \geq 0$, since $\cA(\aa,0,e) \geq 0$, then
	\be
	\phi_{\om} (\bar{\cA}(\aa,\bb_2, e)) \leq \phi_{\om} (\bar{\cA}(\aa,\bb_1, e)),\lb{eqn:4.6}
	\ee
	where the equality holds only if $\aa = 0$ and $\om =1$, i.e., on $D(1,2\pi)$. Then we have (i) holds.

	By the monotonic of the operator $\cA(\aa, \bb,e)$ and  $\bar{\cA}(\aa,\bb, e)$ in (\ref{eqn:4.5}) and (\ref{eqn:4.6}), we have that (ii) holds. Readers may also refer the proofs of  Lemma 4.4 and Corollary 4.5 in \cite{HLS} to obtain (ii) holds
\end{proof}

When $\aa> \bb =0$ and $e \in [0,1)$, we have that the operator $\cA(\aa,0,e)$ is given by
\bea
\cA(\aa,0,e) = -\frac{\d^2}{\d t^2}I_2  -I_2 + \frac{1+\aa}{1+e\cos t}I_2.
\eea

\begin{corollary}
	For $\aa\geq\bb>0$ and $e\in [0,1)$, if $	3\bb-\aa>1+2C_n$, we have
	\begin{equation}
	i_{\om}(\xi_{\aa,\bb,e})\ge n,
	\end{equation}
	where $C_n$ is a constant depending only on $n$.
	Then as $3\bb-\aa$ increases to infinity, the index increases to infinity.
\end{corollary}

\begin{proof}
	We firstly define a space
	\begin{equation}\label{En}
	E_n=\span\left\{
	R(t)\begin{pmatrix}
	0 \\ \cos it
	\end{pmatrix}\;\Big|\;0\le t\le 2\pi,\;i=1,2,...,n\right\}.
	\end{equation}
	Then we have $\dim E_n=n$. Let $\eta$ be a nonzero $C^{\infty}$ function such that
	$\eta^{(m)}(0)=\eta^{(m)}(2\pi)=0$ for any integer $m\ge0$. Then we have
	$\eta E_n\subseteq \ol{D}(\om,2\pi)$ for any $\om\in\U$.

	For $e\in [0,1)$, $0\ne y(t)=R(t)(0,x(t))^T\in E_n$, we have
	\bea
	\<\cA(\aa,\bb,e)\eta{y},\eta{y}\>
	&=&\left<\left[-\frac{d^2}{dt^2}I_2-I_2+RK_{\bb,e}(t)R^T\right]R
	\begin{pmatrix}
		0 \\ \eta{x}
	\end{pmatrix},R
	\begin{pmatrix}
		0 \\ \eta{x}
	\end{pmatrix}\right>
	\nonumber\\
	&=&\l<\left[R
	\begin{pmatrix}
		0 \\ \eta{x}
	\end{pmatrix}
	-2RJ_2
	\begin{pmatrix}
		0 \\ (\eta{x})'
	\end{pmatrix}
	+R
	\begin{pmatrix}
		0 \\  -(\eta{x})''+({1+\aa-3\bb\over1+e\cos t}-1)(\eta{x})
	\end{pmatrix}\right],
	R \begin{pmatrix}
		0 \\ \eta{x}
	\end{pmatrix}\r>
	\nonumber\\
	&=&\l<R
	\begin{pmatrix}
		2(\eta{x})' \\ -(\eta{x})''+{1+\aa-3\bb\over1+e\cos t}\eta{x}
	\end{pmatrix}
	,R\begin{pmatrix}
		0 \\ \eta{x}
	\end{pmatrix}\r>
	\nonumber\\
	&=&\int_0^{2\pi}[(\eta(t)x(t))']^2\d t+(1+\aa-3\bb)\int_0^{2\pi}{(\eta(t)x(t))^2\over1+e\cos t}\d t
	\nonumber\\
	&\le&(C_n+{1+\aa-3\bb\over1+e})\int_0^{2\pi}(\eta(t)x(t))^2\d t,
	\eea
	where we have used the property $\eta(t)x(t)|_{t=0}=0$,
	and $C_n$ is a constant which depend on space $E_n$ because of the finite dimension of $E_n$.
	When $\bb>\frac{1}{3}(1+\aa+(1+e)C_n)$, we obtain that $\cA(\aa,\bb,e)$ is negative definite on the subspace $\eta E_n$ of $\ol{D}(\om,2\pi)$. Hence
	\begin{equation}
	i_{\om}(\xi_{\aa,\bb,e})\ge n,\quad{\rm if}\;
	3\bb-\aa>1+2C_n \quad{\rm and} \;
	e \in [0,1). \lb{eqn:cond}
	\end{equation}
	Therefore, for any given $n_0\in \N$, if $\aa\geq\bb>0$ such that $3\bb-\aa>1+2C_n$, then $i_{\om}(\xi_{\aa,\bb,e})\ge n_0$. Then this corollary holds.
\end{proof}

\setcounter{equation}{0}
\section{Stability in the Hyperbolic Region}\lb{sec:hyper}
In this section, we will prove that the operator $ \cA(\aa,\bb,e)$  is positive definite with zero nullity when $\aa \geq 3\bb>0$ and $e\in [0,1)$.

 \begin{theorem}\lb{prop:pos}
     For any $\om$ boundary condition, $ \cA(\aa,\bb,e)$ is positive definite  for any $e\in [0,1)$ when  $\aa \geq 3\bb >0$. Furthermore, for any $\aa \geq 3\bb >0$, $e\in [0,1)$ and $\om \in \U$,
     \be
     i_{\om}(\xi_{\aa,\bb,e}) = 0, \quad \nu_{\om}(\xi_{\aa,\bb,e}) = 0.
     \ee
     Then $\xi_{\aa,\bb,e}(2\pi)$ possesses two pairs of hyperbolic eigenvalues and it is linearly unstable.
 \end{theorem}

\begin{proof}
The operator $\cA(\aa,\bb,e) $ can be written as
\bea
\cA(\aa,\bb,e) &=& -\frac{\d^2}{\d t^2}I_2  -I_2 + \frac{1}{1+e\cos t}(
(1+\aa)I_2 + 3\bb S(t)) \nn \\
&=&-\frac{\d^2}{\d t^2}I_2  -I_2 + \frac{1}{1+e\cos t}\l((1+\aa)I_2+3\bb\l(-I_2+2R(t)\begin{pmatrix}
    1& 0 \\
    0 & 0
\end{pmatrix}R(t)^T\r)\r)\nn\\
&=&-\frac{\d^2}{\d t^2}I_2  -I_2 + \frac{1}{1+e\cos t}\l((1+\aa-3\bb)I_2+6\bb R(t)\begin{pmatrix}
    1& 0 \\
    0 & 0
\end{pmatrix}R(t)^T\r) .
\eea
where $S(t) $ is given by (\ref{cA}).
When $\aa = 3\bb$, the operator $\cA(\aa,\bb,e)$ can be given by
\bea
\cA(3\bb,\bb,e)
&=&-\frac{\d^2}{\d t^2}I_2  -I_2 + \frac{1}{1+e\cos t}\l(I_2+6\bb R(t)\begin{pmatrix}
    1& 0 \\
    0 & 0
\end{pmatrix}R(t)^T\r)\nn \\
&=& \cA(0,0,e)+ \frac{6\bb}{1+e\cos t} R(t)\begin{pmatrix}
    1& 0 \\
    0 & 0
\end{pmatrix}R(t)^T.
\eea
By the Lemma 5.3, the operator $\cA(0,0,e)$ is non-negative definite on $\ol{D(1,2\pi)}$ with the kernel
\be
\ker \cA(0,0,e) = \{( c_1(1+e\cos t),c_2(1+e\cos t) )^T|c_1,c_2 \in \C  \}
\ee
and it is positive definite on $\ol{D(\om,2\pi)}$ where $\om \in \U\bs\{1\}$.
The operator $\frac{6\bb}{1+e\cos t} R(t)\begin{pmatrix}
    1& 0 \\
    0 & 0
\end{pmatrix}R(t)^T$ is also non-negative definite on $\ol{D(\om,2\pi)}$ for $\om \in \U$.

Therefore, we only need to verify that $\cA(3\bb,\bb,e)$ is positive definite on $\ker \cA(0,0,e)$ for any $e\in [0,1)$.
Let $c_1$ and $c_2$ are complex numbers and
\be
x_0 = ( c_1(1+e\cos t),c_2(1+e\cos t) )^T.
\ee
Therefore, we have that $\<\cA(3\bb,\bb,e)x_0, x_0\> > 0$ for any $\om$ boundary condition by
\bea
\<\cA(3\bb,\bb,e)x_0, x_0\>&=&
\left<\frac{6\bb}{1+e\cos t}R(t)\begin{pmatrix}
    1& 0 \\
    0 & 0
\end{pmatrix}R(t)^T x_0,x_0\right> \nn\\
&>& \left<R(t)\begin{pmatrix}
    1& 0 \\
    0 & 0
\end{pmatrix}R(t)^T \begin{pmatrix} c_1 \\ c_2\end{pmatrix},
\begin{pmatrix} c_1 \\ c_2 \end{pmatrix}\right>\nn\\
&=& \pi (|c_1|^2+|c_2|^2), \lb{3bb.bb.positi}
\eea
where the second inequality holds because $6\bb(1+e\cos t) > 0$ for all $\bb>0$ and $e\in[0,1)$.

Additionally, by (i) of Lemma \ref{lem:A.mono}, when $\aa > 3\bb>0$ and $e\in [0,1)$ , the operator  $\cA(\aa,\bb,e)$ is positive definite for any $\om$ boundary condition. By (\ref{eqn:ind.equ}), this proposition holds.
\end{proof}

\setcounter{equation}{0}
\setcounter{figure}{0}
\section{Study in the Non-Hyperbolic Region}
Note that $T$ given by (\ref{eqn:trans.T}) is invertible whose invert is denoted by $T^{-1}$. The corresponding fundamental solution is defined by $\xi_{\aa,\bb,e}(t) =\td\xi_{\td\aa,\td\bb,e}(t)$ and the corresponding operator $\tilde\cA(\tilde{\aa},\tilde{\bb},e)=\cA(\aa,\bb,e)$ is given
by
\be
\tilde\cA(\tilde{\aa},\tilde{\bb},e) = -\frac{\d^2}{\d t^2}I_2  -I_2 + \frac{1}{2(1+e\cos t)}(3(\td\aa+1)(I_2+S(t)) + \td\bb(I_2+3S(t)).
\ee
We divide $(\td\aa,\td\bb,e) \in [0,1)\times[-1,\infty)\times [0,1)$ into two regions $\cR_{NH}$ and  $\cR_{EH}$ by
\bea
\cR_{NH}  &=&\{(\td\aa,\td\bb,e)| -1< \tilde{\bb} < 0,\tilde{\aa}>0,  e\in[0,1)\};\\
\cR_{EH} &=&\{(\td\aa,\td\bb,e)| \tilde{\bb} > 0, \tilde{\aa}>0, e\in[0,1)\}.
\eea

By the affine transformation $T$, we have the direct results from Lemma \ref{lem:A.mono} of the index and nullity of the operator $\td\cA(\td\aa,\td\bb, e)$.
\begin{lemma}\lb{lemma.index.tdphi}
	(i) The $\om$-Maslov index  $\phi_{\om}(\td\cA(\td\aa,\td\bb, e))$ and the corresponding index $i_{\om}(\td\xi_{\td\aa,\td\bb,e})$ are non-increasing in $\td\aa \in (0,\infty)$ and they are non-decreasing in $\td\bb \in (-1,\infty)$.

	(ii) The sum of $\om$-Maslov index and nullity  $\td\phi_{\om}(\td\cA(\td\aa,\td\bb, e)) + \nu_{\om}(\td\cA(\td\aa,\td\bb, e))$ and the corresponding $i_{\om}(\td\xi_{\td\aa,\td\bb,e})+ \nu_{\om}(\td\xi_{\td\aa,\td\bb,e})$ are non-increasing in $\td\aa \in (0,\infty)$ and they are non-decreasing in $\td\bb \in (-1,\infty)$.
\end{lemma}

\subsection{The index in $\ol{\cR_{NH}}$}
In this part, the $1$-index and nullity of $\td{\cA}(\aa,\bb,e)$ will be discussed when $(\td\aa,\td\bb,e) \in \ol{\cR_{NH}}$.
Note that the stability of $\td\bb\leq -1$ has been discussed in Section \ref{sec:hyper}.

\begin{proposition}\lb{prop:bnd.ind}
    For any $\td\aa \geq 0$, $-1<\td\bb \leq  0$, and $e\in [0,1)$,
    the $1$-index and nullity satisfy
    \be
    \phi_{1}(\td\cA(\td\aa,\td\bb,e) ) = 0, \quad \nu_1(\td\cA(\td\aa,\td\bb,e) ) =
    \begin{cases}
        3, \; \mbox{if}\; \td\aa = 0, \td\bb =0;\\
        1, \; \mbox{if}\; \td\aa >0, \td\bb =0;\\
        0, \; \mbox{others}.
    \end{cases}
    \ee
\end{proposition}

\begin{proof}
    We prove this proposition in 4 steps.

    {\bf Step 1.} {\it The $1$-index and nullity in $\{(\td\aa,\td\bb, e)| \td\aa =0 , -1<\td\bb \leq 0, e\in[0,1)\}$.}

    If $\td\aa =0$ and  $-1<\td\bb \leq 0$, the operator $\td\cA(0,\td\bb,e) $ can be written as
    \bea
    \td\cA(0,\td\bb,e)
    =-\frac{\d^2 }{\d t^2}I_2 -I_2+\frac{I_2}{1+e\cos t}+\frac{\td\bb+1}{2(1+e\cos t)}(I_2 +3 S(t)).
    \eea
    When $\td\bb = -1$, for any $\om \in \U$, $\td\cA(0,-1,e) =\cA(0,0,e)$, the $\om$-index and nullity is given by
    \be
    \phi_{\om}(\td\cA(0,-1,e)) = 0,
    \quad
    \nu_{\om}(\td\cA(0,-1,e) ) = \begin{cases}
        2,\; \mbox{if}\; \om = 1;\\
        0,\; \mbox{if}\; \om \in \U\bs\{1\}.
    \end{cases}
    \ee
    When $\bb = 0$, by (3.6) and (3.8) of \cite{HLS},
    \be
    \phi_{\om}(\td\cA(0,0,e) ) = \begin{cases}
        0, \; \mbox{if} \; \om = 1;\\
        2, \; \mbox{if} \; \om \in \U\bs\{1\};
    \end{cases}
    \quad
    \nu_{\om}(\td\cA(0,0,e)) = \begin{cases}
        3, \; \mbox{if} \; \om = 1;\\
        0, \; \mbox{if} \; \om \in \U\bs\{1\}.
    \end{cases}\lb{eqn.ind.00}
    \ee
    By Lemma \ref{lemma.index.tdphi}, for $-1<\td\bb <0$, the $\om$-index of $\td\cA(0,\td\bb,e)$ is non-decreasing with respect with $\td\bb$. Then for $\td\bb\in(-1,0)$, the $1$-index and nullity of $\td\cA(0,\td\bb,e)$ are given by
    \bea
    \phi_{1}(\td\cA(0,\td\bb,e) ) = 0, \quad \nu_{1}(\td\cA(0,\td\bb,e) \leq 3.
    \eea

    \medskip

    {\bf Step 2.} {\it The $1$-index and nullity in $\{(\td\aa,\td\bb, 0)|\td\aa >0, \td\bb =0\}$.}

    When $\bb = 0$, $e=0$ and $0\leq \td\aa \leq \frac{1}{3}$, the operator $\tilde{\cA}(\td\aa,0,0)$ is given by
    \be
    \tilde{\cA}(\td\aa,0,0) = -\frac{\d^2 }{\d t^2}I_2 -I_2+\frac{3(\td\aa+1)}{2}(I_2 +S(t)).
    \ee
    By the discussion in the Section \ref{Stab.R2}, if $\td\bb =0$ and  $0\leq \td\aa < \frac{1}{3}$ i.e., $\aa=3\bb-1$ and $\frac{1}{2}\leq\bb< \frac{2}{3} $,
    the $1$-indices and nullity satisfy
    \be
    \phi_{1}(\tilde{\cA}(\td\aa,0,0) ) = 0, \quad
    \nu_{1}(\tilde{\cA}(\td\aa,0,0) ) = \begin{cases}
        3,\; \mbox{if}\; \td\aa = 0;\\
        1,\;\mbox{if}\; \td\aa\in(0,\frac{1}{3}],
    \end{cases} \lb{tilde.A.bb.0.index}
    \ee

    When $\td\aa > \frac{1}{3}$, $\td\bb=0$ and $e = 0$, following the discussion in Section \ref{Stab.R3}, we have
    \bea
    \phi_{1}(\tilde{\cA}(\td\aa,0,0) ) = 0, \quad
     \nu_{1}(\tilde{\cA}(\td\aa,0,0)) = 1.
    \eea

    \medskip

    {\bf Step 3.} {\it The $1$-index and nullity in $\{(\td\aa,\td\bb, e)|\td\bb=0, \td\aa> 0, e\in[0,1)\}$.}

    Note that for any $e_0\in [0,1)$ and $\td\aa \geq 1$,
    \bea
    \tilde{\cA}(\td\aa,0,e) &=& -\frac{\d^2 }{\d t^2}I_2 -I_2+\frac{3(1+\td\aa)}{2(1+e\cos t)}(I_2 +S(t))\nn\\
    &\geq& -\frac{\d^2 }{\d t^2}I_2 -I_2+\frac{3(1+\td\aa)}{2(1+e_0)}(I_2 +S(t)) \nn\\
    &=& \tilde{\cA}\l(\frac{\td\aa-e_0}{1+e_0},0,0\r). \lb{e.to.0}
    \eea
    where the second inequality and the forth hold because $I_2 +S(t) \geq 0$ for $\ol{D}(\om,2\pi)$ and $\frac{3(1+\td\aa)}{2(1+e\cos t)} \geq \frac{3(1+\td\aa)}{2(1+e_0)}$ when $e_0 \in [0,1)$ and $\td\aa >1$.

    When $\td\aa >1$, it yields that $\frac{\td\aa-e_0}{1+e_0} > 0$. By (\ref{eqn.ind.00}), we have
    \be
    \phi_{1}(\tilde{\cA}(\td\aa,0,e) ) \leq \phi_{1}\l(\tilde{\cA}\l(\frac{\td\aa-e_0}{1+e_0},0,0\r)\r) =0.
    \ee
    By (\ref{e.to.0}), we have that
    \be
    \phi_{1}(\tilde{\cA}(\td\aa,0,e) ) + \nu_{1}(\tilde{\cA}(\td\aa,0,e))\leq \phi_{1}\l(\tilde{\cA}\l(\frac{\td\aa-e_0}{1+e_0},0,0\r)\r)
    +\nu_{1}\l(\tilde{\cA}\l(\frac{\td\aa-e_0}{1+e_0},0,0\r)\r),
    \ee
    and
    \be
    \nu_{1}(\tilde{\cA}(\td\aa,0,e)) \le \nu_{1}\l(\tilde{\cA}\l(\frac{\td\aa-e_0}{1+e_0},0,0\r)\r) =1.\lb{tilde.A.e}
    \ee

    When $\td\aa \in (0,1]$,
    the operator $\tilde{\cA}_0(\td\aa,0,e)$ is defined by
    \bea
    \tilde{\cA}(\td\aa,0,e) &=& -\frac{\d^2 }{\d t^2}I_2 -I_2+\frac{3(\td\aa+1)}{2(1+e\cos t)}(I_2 +S(t))\nn\\
    &=&-\frac{\d^2 }{\d t^2}I_2 -I_2+\frac{3}{2(1+e\cos t)}(I_2 +S(t)) +\frac{3\td\aa}{2(1+e\cos t)}(I_2 +S(t))\nn\\
    &=&\td\aa\l( \frac{\tilde{\cA}(0,0,e)}{\td\aa}+\frac{3}{2(1+e\cos t)}(I_2 +S(t))\r)\nn\\
    &=&\td\aa\cA_0(\td\aa, 0,e), \lb{eqn:cA0}
    \eea
    where $\tilde{\cA}(0,0,e) \geq 0$ on $\ol{D}(1,2\pi)$.
    Since $\aa> 0$, we have that
    \be
    \phi_{\om}(\cA(\td\aa, 0,e)) = \phi_{\om}(\cA_0(\td\aa, 0,e)),\; \nu_{\om}(\cA(\td\aa, 0,e)) = \nu_{\om}(\cA_0(\td\aa, 0,e)).\lb{eqn:ind.cA0}
    \ee
    By (i) of Lemma \ref{lemma.index.tdphi}, for $\td\aa\in(0,1)$, we have
    \be
    \phi_{1}(\td\cA(\td\aa, 0,e)) = \phi_{1}(\cA_0(\td\aa, 0,e) ) \leq  \phi_{1}(\cA_0(1,0,e) ) = \phi_{1}(\tilde{\cA}(1,0,e)) =0, \quad
    \ee
    and
    \be
    \nu_{1}(\td\cA(\td\aa, 0,e))\leq  \phi_{1}(\td\cA(\td\aa, 0,e) ) +\nu_{1}(\td\cA(\td\aa, 0,e)) \le \phi_{1}(\tilde{\cA}(1,0,e)) +\nu_{1}(\tilde{\cA}(1,0,e)) = 1.\lb{eqn:null.tildeA}
    \ee

    Suppose that $x_0 = R(t)(0,c_0)^T\in D(1, 2\pi)$ where $R(t)$ is given by (\ref{cA}) and $c\in \C$ is a constant.
    By direct computations, we have $\tilde{\cA}(\td\aa,0,e) x_0 = 0$ for $\td\aa \geq 0$.
    This yields $\nu_{1}(\tilde{\cA}(\td\aa,0,e))  \geq 1$.
    Together with (\ref{tilde.A.e}), for $\td\aa \geq 0$,
    \be \nu_{1}(\tilde{\cA}(\td\aa,0,e))  =1.\ee

    Above all, for all $e\in [0,1)$, the $1$-index and nullity satisfy
    \be
    \phi_{1}(\tilde{\cA}(\td\aa,0,e) ) = 0, \quad \nu_1(\tilde{\cA}(\td\aa,0,e) ) =
    \begin{cases}
        3, \; \mbox{if}\; \td\aa =0;\\
        1, \; \mbox{if}\; \td\aa >0.
    \end{cases}\lb{eqn:6.21}
    \ee

    \medskip

    {\bf Step 4.}{\it The $1$-index and nullity in $\{(\td\aa,\td\bb, e)|\td\aa>0, -1< \td\bb<0,e\in[0,1) \}$.}

     By (ii) of Lemma \ref{lemma.index.tdphi} and (\ref{eqn:6.21}), the $1$-index and nullity satisfy
    \be
    \phi_{1}(\td\cA(\td\aa,\td\bb,e)) = 0,\quad \nu_{1}(\td\cA(\td\aa,\td\bb,e)) \leq 1.
    \ee
    By Theorem \ref{Th:multiplicity}, the nullity must be even if $\td\bb \neq 0$. Therefore, when  $\td\aa>0$, $-1< \td\bb<0$ and $e\in[0,1)$,
    \be
     \nu_{1}(\td\cA(\td\aa,\td\bb,e)) = 0.
    \ee
    Then we have this proposition holds.
\end{proof}

\begin{proposition}\lb{prop:-1index}
    For $(\td\aa,\td\bb,e)$ where $\td\aa\geq0$, $-1<\td\bb \leq 0$ and $e\in[0,1)$, $\om$-index and nullity satisfy
    \be
    \phi_{-1}(\td\cA(\td\aa,\td\bb,e)) \leq 2, \quad
    \nu_{-1}(\td\cA(\td\aa,\td\bb,e)) \leq 2. \lb{eqn:general.-1.ind}
    \ee
    Especially, for given $e_0\in[0,1)$, when $\td\aa >\frac{1}{4}+ \frac{5}{4}e_0$ and $-1<\td\bb\leq 0$,
    \be
    \phi_{-1}(\tilde{\cA}(\td\aa,\td\bb,e) )  = 0, \quad
    \nu_{-1}(\tilde{\cA}(\td\aa,\td\bb,e))= 0. \lb{eqn: -1.index}
    \ee
\end{proposition}
\begin{proof}
	By the discussion in the Section \ref{Stab.R2}, for $\td\bb =0$ and  $0\leq \td\aa < \frac{1}{3}$, i.e., $\aa=3\bb-1$ and $\frac{1}{2}\leq\bb< \frac{2}{3} $,
	and in the Section \ref{Stab.R3} for $\td\aa>\frac{1}{3}$,
	\be
	\phi_{-1}(\tilde{\cA}(\td\aa,0,0)) = \begin{cases}
		2, \;\mbox{if}\; \td\aa \in [0,\frac{1}{4});\\
		0, \;\mbox{if}\; \td\aa\in[\frac{1}{4},\infty),
	\end{cases} \quad
	\nu_{-1}(\tilde{\cA}(\td\aa,0,0)) = \begin{cases}
		2, \;\mbox{if}\; \td\aa = \frac{1}{4};\\
		0, \;\mbox{if}\; \td\aa\neq \frac{1}{4};
	\end{cases}\lb{eqn:-1.index.tildeA}
	\ee
	By (\ref{e.to.0}) and (\ref{eqn:-1.index.tildeA}), when $\frac{\td\aa-e_0}{1+e_0} > \frac{1}{4}$, i.e., $\td\aa >\frac{1}{4}+ \frac{5}{4}e_0$, following  similar arguments, we have
	\be
	\phi_{-1}(\tilde{\cA}(\td\aa,0,e) ) \leq\phi_{-1}\l(\tilde{\cA}\l(\frac{\td\aa-e_0}{1+e_0},0,0\r)\r) = 0, \quad
	\nu_{-1}(\tilde{\cA}(\td\aa,0,e)) \le \nu_{-1}\l(\tilde{\cA}\l(\frac{\td\aa-e_0}{1+e_0},0,0\r)\r)= 0.
	\ee
	By Lemma \ref{lemma.index.tdphi}, we have that (\ref{eqn: -1.index}) holds.

	By (\ref{eqn.ind.00}) and (i) of Lemma \ref{lemma.index.tdphi}, we also have that for $\td\aa\in(0,2)$,

	\be
	\phi_{-1}(\td\cA(\td\aa, 0,e)) \leq   \phi_{-1}(\cA_0(0,0,e) ) = 2, \quad
	\ee
	and
	\be
	\nu_{-1}(\td\cA(\td\aa, 0,e))\leq  \phi_{-1}(\td\cA(0, 0,e) ) +\nu_{-1}(\td\cA(0, 0,e))  = 2.\lb{null.tildeA}
	\ee
	Again by Lemma \ref{lemma.index.tdphi}, (\ref{eqn:general.-1.ind}) holds.
	Then we have this proposition holds
\end{proof}

By the discussion in Theorem \ref{prop:pos} and Proposition \ref{prop:-1index}, especially (\ref{eqn:-1.index.tildeA}), we have that there exist two $-1$-degenerate surfaces. Then for $\tilde\aa > 0$, $e \in[0,1)$, we let $\tilde\bb_{1}(\tilde\aa,e)$, $\tilde\bb_{2}(\tilde\aa,e)$ be the two $-1$-degenerate surfaces  where the $-1$-index changes and further define $\tilde\bb_s(\tilde\aa,e)$ and $\tilde\bb_m(\tilde\aa,e)$ by
\bea
\tilde\bb_{s}(\tilde\aa,e)=\min \left\{\tilde\bb_{1}(\tilde\aa,e), \tilde\bb_{2}(\tilde\aa,e)\right\} \quad \text {and} \quad
\tilde\bb_{m}(\aa,e)=\max \left\{\tilde\bb_{1}(\tilde\aa,e), \tilde\bb_{2}(\tilde\aa,e)\right\}. \lb{eqn:bbs.bbm}
\eea
Note that, when $e = 0$, $\tilde\bb_{1}(\tilde\aa,0) = \tilde\bb_{2}(\tilde\aa,0)$ and $(\td\aa,\tilde\bb_{1}(\tilde\aa,0)) = T \cR_{2,\frac{1}{2}}^*$ where $\cR_{2,\frac{1}{2}}^*$ is given by (\ref{eqn:seg.l}).
When $e>0$, $\tilde\bb_{s}(\tilde\aa,e)$ and $\tilde\bb_{m}(\tilde\aa,e)$ bifurcate from $\cR_{2,\frac{1}{2}}^*$.

We define the boundary of the elliptic region in $\cR_{NH}$ by
\be
\tilde\bb_{k}(\td{\aa},e)=\inf \left\{\td\bb^{\prime} \in[-1,0] | \sigma\left(\td\xi_{\td\aa,\td\bb, e}(2 \pi)\right) \cap \mathbf{U} \neq \emptyset, \quad \forall \tilde\bb\in[-1, \tilde\bb^{\prime})\right\}. \lb{eqn:bbk}
\ee
Note that for any given $\tilde\aa_0 > 0$, if $\tilde\bb  = -1$, i.e., $\aa= 3\bb$, $\sigma\left(\td\xi_{\td\aa,\td\bb, e}(2 \pi)\right) \cap \mathbf{U} = \emptyset$. Therefore for any $(\tilde{\aa} ,e)$, $\tilde\bb_{k}(\tilde{\aa},e)$ is well defined.

By the definition of (\ref{eqn:bbs.bbm}) and (\ref{eqn:bbk}), we have that for the given $\tilde\aa_0$ and $e_0\in[0,1)$,  if $\tilde\bb_s(\aa_0,e_0)$ exists, following holds
\be
\tilde\bb_k(\tilde\aa_0,e) \leq \tilde\bb_s(\tilde\aa_0,e_0).
\ee

Then we define the $-1$-degenerate surfaces and the elliptic boundary by
\bea
\B_{m}^* &=& \{ (\td\aa,\td\bb,e)\in\cR_{NH}| \tilde\bb=\tilde\bb_{m}(\tilde\aa,e) \}, \lb{eqn:RE.Bm}\\
\B_{s}^* &=& \{ (\td\aa,\td\bb,e)\in\cR_{NH}| \tilde\bb=\tilde\bb_s(\tilde\aa,e)\}, \\
\B_{k}^* &=& \{ (\td\aa,\td\bb,e)\in\cR_{NH}| \tilde\bb =\tilde\bb_{k}(\tilde\aa,e)\}.  \lb{eqn:RE.-1.ind}
\eea
The regions between them are defined as $\B_s$, $\B_m$, $\B_k$ ,and $\B_{h}$ of $\cR_{NH}$ by
\bea
\B_{m} &=& \{ (\td\aa,\td\bb,e)\in\cR_{NH}| \tilde\bb_{m}(\tilde\aa,e) <\tilde\bb< 0,\tilde{\aa} > 0 \}, \\
\B_{s} &=& \{ (\td\aa,\td\bb,e)\in\cR_{NH}| \tilde\bb_s(\td\aa,e)<\tilde\bb<0, \tilde{\aa} > 0\}\bs \ol{\B_{m}},\\
\B_{k} &=& \{ (\td\aa,\td\bb,e)\in\cR_{NH}| \tilde\bb_{k}(\tilde\aa,e) < \tilde\bb < 0,\tilde\aa> 0\} \bs (\ol{\B_{m}\cup \B_{s}} ),\\
\B_{h} &=& \cR_{NH}\bs ( \ol{\B_{m} \cup \B_{s} \cup \B_{k}}).\lb{eqn:RE.-1.nul}
\eea

\begin{proposition}\lb{prop:-1.index}
	For $(\td\aa,\td\bb,e) \in \cR_{NH}$, the $-1$-index and nullity of $\td\xi_{\td\aa,\td\bb,0}$ satisfy
	\bea
	i_{-1}\left(\td\xi_{\td\aa,\td\bb, e}\right) &=&
	\left\{\begin{array}{l}
		{2, \text { if } (\td\aa,\td\bb,e) \in \B_m};\\
		{1, \text { if } (\td\aa,\td\bb,e) \in \B_s\cup\B_m^*}; \\
		{0, \text { if } (\td\aa,\td\bb,e) \in \cR_{NH}\bs(\B_s\cup \B_m\cup\B_m^*)};
	\end{array}\right. \\
	\nu_{-1}\left(\xi_{\td\aa,\td\bb, e}\right) &=&
	\left\{\begin{array}{l}
		{0, \text { if } (\td\aa,\td\bb,e) \in \cR_{NH}\bs (\B_{s}^*(\td\aa,e)\cup\B_{m}^*(\td\aa,e))}; \\
		{1, \text { if } (\td\aa,\td\bb,e) \in \B_{s}^*(\td\aa,e) \cup \B_{m}^*(\td\aa,e)\bs(\B_{s}^*(\td\aa,e) \cap \B_{m}^*(\td\aa,e)}); \\
		{2, \text { if } (\td\aa,\td\bb,e) \in \B_{s}^*(\td\aa,e) \cap \B_{m}^*(\td\aa,e)}.
	\end{array}\right.
	\eea
\end{proposition}
\begin{proof}
	When $e = 0$, as the discussion in Section \ref{Stab.R2}, the two degenerate surfaces satisfy $(\td\aa,\td\bb_{s}) = (\td\aa,\td\bb_{m}) \in T\cR_{2,\frac{1}{2}}^*$
	where $\cR_{2,\frac{1}{2}}^*$ is given by (\ref{eqn:seg.l}).

	For $0< e <1 $, by the definition of $\bb_s(\aa,e)$ and $\bb_m(\aa,e)$ satisfying  for $\aa > 0$ and $e \in[0,1)$ in (\ref{eqn:bbs.bbm}),
	we have that $-1$-index stays the same and only changes when $(\aa,\bb,e) \in \B_{m}^*\cup \B_{s}^*$.
	Then we have this proposition holds.
\end{proof}

\begin{lemma}\lb{lem:mono.ind.nul}
(i) For the given $\td\aa_0$ and $e_0$, if
$(\td\aa_0, \td\bb_1, e_0)$ and $(\td\aa_0, \td\bb_2, e_0)$ are both in $\cR_{NH}$ with $-1<\td\bb_1 \leq \td\bb_2 < 0$
and $\td\xi_{\td\aa_0,\td\bb_2,e_0}(2\pi)$ is hyperbolic, then $\td\xi_{\td\aa_0,\td\bb_1,e_0}(2\pi)$ is hyperbolic.
Consequently, the hyperbolic region of $\td\xi_{\td\aa,\td\bb,e}$ in $\cR_{NH}$ is connected.

(ii) For $(\td\aa,\td\bb,e) \in \B_h$, every matrix $\td\xi_{\td\aa,\td\bb,e}(2\pi)$ is hyperbolic.
Thus $\B_k^*$ is the boundary set of this hyperbolic region.

(iii) For the any $e\in [0,1)$, and $\td\aa_*\in( 0, \infty)$,
the total multiplicity of $\om$ degeneracy  of
\be \ga(t) = \td\xi_{\td\aa(t),\td\bb(t), e}(2 \pi), \lb{eqn:path.ga}\ee
for $t\in[0,1]$  with $\td\aa(t) = t\td\aa_* $ and $\td\bb(t) = -t$ satisfies that
\be
\sum_{ 0\leq t\leq 1} \nu_{\om}(\ga(t))=2, \quad \forall \omega \in \mathbf{U} \backslash\{1\}.
\ee
\end{lemma}

\begin{proof}
By the Lemma \ref{lemma.index.tdphi}, for any fixed $\td\aa$, $\om$ and $e$,
\bea
\phi_{\om}(\td\cA(\aa,\bb_1,e))&<& \phi_{\om}(\td\cA(\aa,\bb_2,e)) ; \\
\phi_{\om}(\td\cA(\aa,\bb_1,e))+\nu_{\om}(\td\cA(\aa,\bb_1,e))&<& \phi_{\om}(\td\cA(\aa,\bb_2,e)) +\nu_{\om}(\td\cA(\aa,\bb_2,e)).
\eea
Suppose that $\td\xi_{\td\aa,\td\bb_2,e}(2\pi)$ is hyperbolic. This implies $\phi_{\om}(\td\cA(\aa,\bb_2,e)) =0 $ and $\nu_{\om}(\td\cA(\aa,\bb_2,e)) =0$. Then $\phi_{\om}(\td\cA(\aa,\bb_1,e)) =0 $ and $\nu_{\om}(\td\cA(\aa,\bb_1,e))$ for any  $\om \in \U$. Therefore $\td\xi_{\td\aa,\td\bb,e}(2\pi)$ must be hyperbolic for all $\bb \in [0,\bb_2)$.

Note that when $\td\bb = -1$ and $\td\aa> 0$, the matrix $\td\xi_{\td\aa,\td\bb,e}$ is hyperbolic by Theorem \ref{prop:pos}.
Therefore, the hyperbolic region of $\td\xi_{\td\aa,\td\bb,e}$ is connected in $\cR_{NH}$.

(ii) By the definition of $\td\bb_k(\td\aa,e)$, there exists a sequence$\{\td\bb_i\}_{i\in \N}$ satisfying $\td\bb_i<\td\bb_k(\aa,e)$, $\td\bb_i\to \td\bb_k(\td\aa,e)$, and $\td\xi_{\td\aa,\td\bb_i,e}(2\pi)$ is hyperbolic.
Therefore $\td\xi_{\td\aa,\td\bb,e}(2\pi)$ is hyperbolic
for every $\td\bb\in [-1,\td\bb_k(e))$ by (i).
Then (\ref{eqn:bbk}) holds and $\td\bb_k(\td\aa,e)$ is the envelope surface of this hyperbolic region.

(iii) Note that both $\td\xi_{0,0, e}(2 \pi)$ and $\td\xi_{\aa_*,-1, e}(2 \pi)$ are both non-degenerate when $\om\in \U\bs\{1\}$.
The corresponding operator path is defined by $\ga^*(t) = \cA(\td\aa(t),\td\bb(t),e)$.
For $\td\aa_* \in (0,\infty)$ and $t_0\in(0,1)$ such that $\td\cA(\td\aa(t_0),\td\bb(t_0),e)$ is
degenerate, the $\om$-index must decrease strictly.
By Theorem \ref{prop:pos} and Lemma \ref{lemma.index.tdphi}, there exist at most two $t_1$ and $t_2$ such that
at each of which the $\om$-index decreases
by 1 if $t_1 \neq t_2$, or the $\om$-index decreases by 2 if $t_1 =t_2$. Suppose that the two values are given by $t_1 = t_1(\td\aa_*,e)$
and $t_2 = t_2(\td\aa_*, e)$ such that for $\ep > 0$ small enough, we have that
$\phi_{\om}(\ga^*(0)) = \phi_{\om}(\ga^*(t_1-\ep))$, $\phi_{\om}(\ga^*(t_1)) = \phi_{\om}(\ga^*(t_2-\ep))$ and $\phi_{\om}(\ga^*(t_2)) =\phi_{\om}(\ga^*(1))$.
Then we have that
\bea
2 &=&\phi_{\om}(\ga^*(0))-\phi_{\om}(\ga^*(1)) \nn\\ &=&\phi_{\om}(\ga^*(t_1-\ep))-\phi_{\om}(\ga^*(t_1))) +\phi_{\om}(\ga^*(t_2-\ep))-\phi_{\om}(\ga^*(t_2)) \nn\\
&=&\dim \ker(\ga^*(t_1))+\dim \ker(\ga^*(t_2)) \nn \\
&=&\nu_{\om}(\ga(t_1))+\nu_{\om}(\ga(t_2)) \nn\\
&=&\sum_{t \in[0,1]} \nu_{\om}(\ga(t)).
\eea
Then we have that (iii) of this lemma holds.
\end{proof}

\begin{corollary}\lb{coro:null.sep}
	For $e\in [0,1)$ and $\td\aa_* \in (0,\infty)$, suppose the continuous path $\ga(t)$ is defined by (\ref{eqn:path.ga}). There exists $t_*\in (0,1)$ such that $\ga(t_*) = \td\xi_{\td\aa(t_*),\td\bb(t_*),e}(2 \pi)\in \B_s^*$, and
	\be
	\sum_{t \in[0,t^* ]} \nu_{-1}(\ga(t))=2, \quad
	\sum_{t \in(t^*,1 ]} \nu_{-1}(\ga(t))=0. \lb{eqn:part.null}
	\ee
\end{corollary}

\begin{proof}
	Fix $\td\aa_*$ and $e\in [0,1)$. There exists a $t_*$ such that $\ga(t_*) \in \B_s^*$. Then if $(t_*\td\aa_*, -t_*, e) \in \B_s^*$ and $(t_*\td\aa_*, -t_*, e) \notin \B_m^*$,
	\be
	\sum_{t\in[0, t_*]} \nu_{-1}(\xi_{\td\aa,\td\bb, e}(2 \pi)) \geq \nu_{-1}(\xi_{\td\aa,\td\bb_{s}(\td\aa,e), e}(2 \pi))+\nu_{-1}(\xi_{\td\aa,\td\bb_{m}(\td\aa,e), e}(2 \pi)) \geq 2;
	\ee
	if $(t_*\td\aa_*, -t_*, e) \in \B_s^*\cap \B_m^*$,
		\be
	\sum_{t\in[0, t_*]} \nu_{-1}(\xi_{\td\aa,\td\bb, e}(2 \pi)) \geq \nu_{-1}(\xi_{\td\aa,\td\bb_{m}(\td\aa,e), e}(2 \pi)) \geq 2.
	\ee
	Together with (iii) of Lemma \ref{lem:mono.ind.nul}, we have that (\ref{eqn:part.null}) holds.
\end{proof}

\begin{proposition}
    The function $\td\bb_k(\td\aa,e)$ is continuous in $\td\aa$ and $e$.
\end{proposition}

\begin{proof}
    We prove this proposition by contradiction. Suppose that $\td\bb_k(\td\aa,e)$ is not continuous in $\td\aa$ or $e$. There must exist some $(\hat{\aa},\hat{e})$ and a sequence $\{(\td\aa_i,e_i)\}_{i=1}^{\infty} \subset [0,\infty]\times[0,1)\bs\{(\hat{\aa},\hat{e})\}$ and $\td\bb_0 \in [-1,0]$ such that
    \be
    \td\bb_k(\td\aa_i,e_i) \to \td\bb_0 \neq \td\bb_k(\hat{\aa},\hat{e}),\; (\td\aa_i,e_i) \to (\hat{\aa},\hat{e}), \; {\rm if } \; i \to \infty.
    \ee
    We discuss the two cases of the continuity according to the sign of $\td\bb_0 -\td\bb_k(\hat{\aa},\hat{e})$. By the continuity of the eigenvalues of the matrix $\td\xi_{\td\aa,\td\bb,e}(2\pi)$ and by (\ref{eqn:bbk}) following holds.
    \be
    \sg(\xi_{\hat{\aa},\bb_0,\hat{e}}(2\pi))\cap \U \neq \emptyset.
    \ee
    By the definition of $\td\bb_k(\hat{\aa},\hat{e})$ and (i) of Lemma \ref{lem:mono.ind.nul}, we must have $\td\bb_k(\hat{\aa},\hat{e})<\bb_0$.

    Now we suppose $\td\bb_k(\hat{\aa},\hat{e})<\td\bb_0$. By the continuity of $\td\bb_s(\td\aa,e)$ and the definition of $\td\bb_0$,
    \be
    \td\bb_k(\hat{\aa},\hat{e}) < \td\bb_0
    \leq \td\bb_s(\hat{\aa},\hat{e}).
    \ee
    By the definition of $\td\bb_k(\hat{\aa},\hat{e})$, let $\om_0\in \sg(\td\xi_{\hat{\aa},\td\bb_k(\hat{\aa},\hat{e}),\hat{e}}(2\pi))\cap \U$.
    Let $L=\{(\hat{\aa},\td\bb, \hat{e}) | \td\bb \in     [-1,\td\bb_{k}(\hat{\aa},\hat{e}))\}$, $V=\{(\td\aa,-1, e)| \td\aa\in(0,\infty), e \in[0,1)\}$, and $L_i =\{(\td\aa_i,\td\bb,e_i)|\td\bb\in [-1,\td\bb_k(\td\aa_i,e_i)]\}$.
    \be
    i_{\omega_{0}}\left(\td\xi_{\td\aa,\td\bb, e}\right)=
    \nu_{\omega_{0}}\left(\td\xi_{\td\aa,\td\bb, e}\right)=0,
    \quad \forall(\td\aa,\td\bb,e) \in L \cup V \cup \bigcup_{i \geq 1} L_{i}.
    \ee
    In particular, we have
    \be
    i_{\omega_{0}}\left(\td\xi_{\hat{\aa},\td\bb_{k}(\hat{\aa},\hat{e}), \hat{e}}\right)=0 \quad \text { and } \quad \nu_{\omega_{0}}\left(\td\xi_{\hat{\aa},\td\bb_{k}(\hat{\aa},\hat{e}) \hat{e}}\right) \geq 1.
    \ee
    Therefore, by the definition of $\om_0$, there exists $\hat{\bb} \in (\td\bb_k(\hat{\aa},\hat{e}),\td\bb_0)$ sufficiently close to $\td\bb_k (\hat{\aa},\hat{e})$ such that
    \be
    i_{\omega_{0}}\left(\td\xi_{\hat{\aa},\hat{\bb}, \hat{e}}\right)
    =i_{\omega_{0}}\left(\td\xi_{\hat{\aa},\td\bb_{k}(\hat{\aa},\hat{e}), \hat{e}}\right)+\nu_{\omega_{0}}\left(\td\xi_{\hat{\aa},\td\bb_{k}(\hat{\aa},\hat{e}), \hat{e}}(2 \pi)\right) \lb{eqn:accol}
    \geq 1.
    \ee
    Note that (\ref{eqn:accol}) holds for all $\td\bb \in (\hat{\bb}_k(e), \hat{\bb}_0]$. Also $(\hat{\aa},\hat{\bb},\hat{e})$ is an accumulation point of $\cup_{i\geq 1} L_i$. This yields there exists $(\td\aa_i,\td\bb_i,e_i) \in L_i$ such that $\xi_{\td\aa_i,\td\bb_i,e_i}$ is $\om_0$-degenerate.
    Moreover $\td\bb_i\to\hat{\bb}_0$ and $\xi_{\td\aa_i,\td\bb_i,e_i}(2\pi) \to \xi_{\hat\aa,\td\bb_0,\hat e}(2\pi)$ as $i\to \infty$.
    Then we have following contradiction for $i \geq 1$ large enough
    \be
    1 \leq  i_{\omega_{0}}(\xi_{\hat{\aa},\hat{\bb}, \hat{e}}) \leq
     i_{\omega_{0}}(\xi_{\td\aa_i,\td\bb_i(\td\aa_i,e_i), e_i}) = 0.
    \ee
    Then we have the continuity of $\td\bb_k(\td\aa,e)$ in $\td\aa$ and $e$.
\end{proof}

\begin{proof}[Proof of Theorem \ref{thm:RE.norm.form}]
    (i) By the Bott-type formula (Theorem 9.2.1 in p. 199 of \cite{Lon1}), the index and nullity of $2$nd-iteration of the symplectic path $\td\xi_{\td\aa,\td\bb, e}(t)$ satisfies
    \bea
    i_{1}\l(\td\xi_{\td\aa,\td\bb, e}^{2}\r) &=& i_{-1}\l(\td\xi_{\td\aa,\td\bb, e}\r)+i_{1}\l(\td\xi_{\td\aa,\td\bb,e}\r),\nn\\
    \nu_{1}\left(\td\xi_{\td\aa,\td\bb, e}^{2}\right)&=&\nu_{1}\left(\td\xi_{\td\aa,\td\bb, e}\right)+ \nu_{-1}\left(\td\xi_{\td\aa,\td\bb, e}\right)=\nu_{-1}\left(\td\xi_{\td\aa,\td\bb, e}\right).\nn
    \eea
    where $\nu_{1}(\td\xi_{\td\aa,\td\bb, e}) = 0$ when $(\td\aa,\td\bb,e)\in\cR_{NH}$. Therefore, by (\ref{eqn:RE.-1.ind}) and (\ref{eqn:RE.-1.nul}), the 2-iteration of the index is given by
    \be
    i_1(\td\xi^2_{\td\aa,\td\bb, e})=\left\{
    \begin{array}{ll}
        {2,} & {\text { if } (\td\aa,\td\bb, e) \in \B_m}; \\
        {1,} & {\text { if } (\td\aa,\td\bb, e) \in \B_s}; \\
        {0,} & {\text { if } (\td\aa,\td\bb, e) \in \B_k }.
    \end{array}\right.
    \ee
    Follow the discussion in the proof of Theorem 1.2 of \cite{HS2}, if $(\td\aa,\td\bb, e)$ satisfies $\td\bb \neq \bb_{m}(\td\aa,e)$ and $\td\bb \neq \bb_{s}(\td\aa,e)$, the matrix  $\td\xi_{\td\aa,\td\bb, e}(4\pi) =
    \td\xi_{\td\aa,\td\bb, e}^{2}(2\pi)$ is non-degenerate with respect with eigenvalue $1$.

    For (i), suppose  $\omega_{i}=e^{\sqrt{-1 \theta_{i}}} \in \sigma\left(\td\xi_{\td\aa,\td\bb, e}(2\pi) \cap \mathbf{U}\right)$, $\th_i \in(0,\pi)$. Note that $i_{-1}(\td\xi_{\td\aa,\td\bb,e}) = 2$, and by (\ref{eqn:split})
     \bea
    i_{-1}(\td\xi_{\td\aa,\td\bb,e}) &=& i_{1}(\td\xi_{\td\aa,\td\bb,e})
    +\sum_{\omega_{i}}(S_{\td\xi_{\td\aa,\td\bb,e}(2\pi)}^{+} (\omega_{i} )-S_{\td\xi_{\td\aa,\td\bb,e}(2\pi)}^{-} (\omega_{i}))+ S_{\td\xi_{\td\aa,\td\bb,e}(2\pi)}^{+}(-1) \nn\\
    &=&\sum_{\omega_{i}} (S_{\td\xi_{\td\aa,\td\bb,e}(2\pi)}^{+} (\omega_{i} )-S_{\td\xi_{\td\aa,\td\bb,e}(2\pi)}^{-} (\omega_{i})). \lb{eqn:-1.split}
     \eea
    It yields that
     \be
     2=\sum_{\omega_{i}}\l(S_{\td\xi_{\td\aa,\td\bb,e}(2\pi)}^{+}(\omega_{i})-S_{\td\xi_{\td\aa,\td\bb,e}(2\pi)}^{-}(\omega_{i})\r) \leq \sum_{\omega_{i}}\left(S_{\td\xi_{\td\aa,\td\bb,e}(2\pi)}^{+}(\omega_{i})+S_{\td\xi_{\td\aa,\td\bb,e}(2\pi)}^{-}(\omega_{i})\right) \leq 2.
     \ee
     Then we have that $S_{\td\xi_{\td\aa,\td\bb,e}}^{-}(\omega_{i}) =0$, by the list of splitting number in Section 2.
     Therefore, there exist the two $\om_1$ and $\om_2$ such that $S_{\td\xi_{\td\aa,\td\bb,e}}^{+}(\omega_{i}) =1$.
     Then we have (i) of Theorem \ref{thm:RE.norm.form} holds.

     (ii) Note that $i_1(\td\xi^2_{\td\aa,\td\bb, e}) = 1$ implies that $i_{-1}(\td\xi_{\td\aa,\td\bb,e}) = 1$.
     Therefore, still by (\ref{eqn:-1.split}), there exists exact one eigenvalue $\omega=e^{\sqrt{-1 \theta_{i}}} \in \sigma\left(\td\xi_{\td\aa,\td\bb,e}(2\pi) \cap \mathbf{U}\right)$ for $\th \in (0,\pi)$ with the splitting number $(1,0)$.
     By the splitting number the list of splitting number in Section 2, we must have $\td\xi_{\td\aa,\td\bb,e}(2\pi) \approx D(\lm) \diamond R(\th)$.
     Also note that $i_1(\td\xi_{\td\aa,\td\bb,e}) = 0$ implies $\lm < 0$.
     Then we have (ii) of Theorem \ref{thm:RE.norm.form} holds.

     (iii)
 	  Suppose that  $(\td\aa_{0},\td\bb_{0},e_{0})\in \B_k$.
 	  By (i) and (ii) of Lemma \ref{lem:mono.ind.nul}, the matrix $M\equiv\td\xi_{\td\aa_{0},\td\bb_{0},e_{0}}(2\pi)$ is not hyperbolic with at least one pair of on the unit circle $\U$.
 	  Furthermore, by the definition of $\td\bb_k(\td\aa_{0},e_{0})$ and $\td\bb_s(\td\aa_0,e_0)$, $\pm 1 \notin \sg(\td\xi_{\td\aa_{0},\td\bb_{0},e_{0}}(2\pi))$.
 	  Suppose that
 	\be
 	\sigma(\td\xi_{\td\aa_{0},\td\bb_{0},e_{0}}(2\pi))= \{\lambda_{1} , \lambda_{1}^{-1}, \lambda_{2}, \lambda_{2}^{-1}\},
 	\ee
 	where $\lambda_{1}\in \mathbf{U} \backslash \mathbf{R}$
 	and $\lambda_{2}\in(\mathbf{U} \cup \mathbf{R}) \backslash\{ \pm 1,0\}$.

 	If $\lambda_{2}\in \R\bs\{\pm 1, 0\}$, the normal form is given by  $\td\xi_{\td\aa_0,\td\bb_0,e_0}(2\pi) \approx D(\lm) \diamond R(\th)$ for some $\th \in (0,\pi) \cup (\pi,2\pi)$. Then by (\ref{eqn:split}), we obtain following contradiction.
 	\bea
 	0 &=&i_{-1}\left(\td\xi_{\td\aa_0,\td\bb_{0}, e_{0}}\right) \nn\\
  &=&i_{1}\left(\td\xi_{\td\aa_0,\td\bb_{0}, e_{0}}\right)+S_{M}^{+}(1)-S_{M}^{-}\left(e^{ \pm \sqrt{-1} \theta}\right)+S_{M}^{+}\left(e^{ \pm \sqrt{-1} \theta}\right)-S_{M}^{-}(-1) \nn\\
 	&=&0+0-S_{R(\theta)}^{-}\left(e^{ \pm \sqrt{-1} \theta}\right)+S_{R(\theta)}^{+}\left(e^{ \pm \sqrt{-1} \theta}\right)-0 \nn\\
 	&=&\pm 1. \lb{eqn:lm2.split}
 	\eea
 	Then we have $\lambda_{2}\in \mathbf{U} \bs \mathbf{R}$ and
 	$\xi_{\aa_0,\bb_{0}, e_{0}}(2\pi) \approx R(\th_1) \diamond R(\th_2)$ for $\th_1$, $\th_2 \in (0,\pi) \cup (\pi ,2\pi)$.
 	Again as (\ref{eqn:lm2.split}), we have that
 	\bea
 	0&=& i_{-1}\left(\td\xi_{\td\aa_0,\td\bb_{0}, e_{0}}\right) \nn\\
 	&=& i_{1}\left(\td\xi_{\td\aa_0,\td\bb_{0}, e_{0}} \right)+S_{M}^{+}(1)-S_{R\left(\theta_{1}\right)}^{-}\left(e^{ \pm \sqrt{-1} \theta_{1}}\right)+S_{R\left(\theta_{1}\right)}^{+}\left(e^{ \pm \sqrt{-1} \theta_{1}}\right) \nn\\ &&-S_{R\left(\theta_{2}\right)}^{-}\left(e^{ \pm \sqrt{-1} \theta_{2}}\right)+S_{R\left(\theta_{1}\right)}^{+}\left(e^{ \pm \sqrt{-1} \theta_{2}}\right)-S_{M}^{-}(-1) \nn\\
 	&=&-S_{R\left(\theta_{1}\right)}^{-}\left(e^{ \pm \sqrt{-1} \theta_{1}}\right)+S_{R\left(\theta_{1}\right)}^{+}\left(e^{ \pm \sqrt{-1} \theta_{1}}\right)-S_{R\left(\theta_{2}\right)}^{-}\left(e^{ \pm \sqrt{-1} \theta_{2}}\right)+S_{R\left(\theta_{2}\right)}^{+}\left(e^{ \pm \sqrt{-1} \theta_{2} }\right).
 	\eea
 	Note that if $\th_1$ and $\th_2$ locate at the same interval $(0,\pi)$ or $(\pi, 2\pi)$, the right hand side will be $\pm 2$. Thus, we must have that  $\th_1 \in (0,\pi)$ and $\th_2\in (\pi, 2\pi)$.

 	Suppose that If $\th_1 = 2\pi -\th_2$, following equation holds.
 	\bea
 	\sum_{\td\bb_k < \td\bb \leq  0} \nu_{\omega}(\td\xi_{\td\aa_0,\td\bb_{0}, e_{0}})
 	\geq \sum_{ \td\bb_{s}\leq \bb \leq 0} 	\nu_{\omega}(\td\xi_{\td\aa_0,\td\bb_{0}, e_{0}})
 	+\nu_{\omega}(\td\xi_{\td\aa_0,\td\bb_{0}, e_{0}}) \geq 1+2,
 	\eea
 	where $\om = \exp(\sqrt{-1}\th_1)$.
 	This is a contradiction with (iii) of Lemma \ref{lem:mono.ind.nul}.

 	If $2\pi -\th_2 > \th_1$, for $\om =\exp(\sqrt{-1}\th_1)$, we have that
 	\be
 	0 \leq i_{\omega}(\td\xi_{\td\aa_0,\td\bb_{0}, e_{0}})
 	=i_{1}(\td\xi_{\td\aa_0,\td\bb_{0}, e_{0}})+S_{M}^{+}(1)-S_{R(\theta_{1})}^{-}(e^{\sqrt{-1} \theta_{1}})
 	=-S_{R(\theta_{1})}^{-}(e^{\sqrt{-1} \theta_{1}})=-1.
 	\ee
 	where $M = \td\xi_{\td\aa_0,\td\bb_{0}, e_{0}}(2\pi)$.
 	This contradiction yields that $2\pi -\th_2 <\th_1$.
 	Then we have (ii) of Theorem \ref{thm:RE.norm.form} holds.
\end{proof}

\begin{lemma}\lb{lem:boundary}
    For some $(\td\aa,\td\bb_0,e) \in \cR_{NH}$, if $\td\xi_{\td\aa,\td\bb_0,e}(2\pi) \approx M_2$ where $M_2$ is given by \eqref{eqn:m2} with $c_1$, $c_2\in \R$
    holds,
    or it possesses the
    basic normal form $N_1(-1,a) \diamond N_1(-1,b)$ and $a$, $b \in \R$, then $\td\xi_{\td\aa,\td\bb,e}(2\pi)$ is hyperbolic for all $\td\bb \in [-1,\td\bb_0)$.
\end{lemma}
\begin{proof}
    Note that the basic normal form of the matrix $M_2(-1,c) $ is either $N_1(-1,\hat{a}) \diamond N_1(-1, \hat{b})$ or $N_1(-1,\hat{a})\diamond D(\lm)$ for some $\hat{a}$, $\hat{b} \in \R$ and $0 > \lm \neq -1$. Thus for any $\om \in \U\bs\{1\}$, by the study in Section 9.1 of
    \cite{Lon4}, we obtain
    \be
    0 \leq i_{\omega}(\td\xi_{\td\aa,\td\bb_{0}, e})=i_{1}(\td\xi_{\td\aa,\td\bb_{0}, e})+S_{M}^{+}(1)-S_{M}^{-}(\omega)=-S_{M}^{-}(\omega) \leq 0.
    \ee
    where $M = \td\xi_{\td\aa,\td\bb_{0}, e}(2\pi)$. Then we have that $i_{\omega}(\td\xi_{\td\aa,\td\bb_{0}, e}) = 0$ for all $\om \in \U$. Note that $\phi_{\om}(\td\cA(\td\aa,\td\bb_0,e)) = i_{\omega}(\td\xi_{\td\aa,\td\bb_{0}, e})$ and $\nu_{\om}(\td\cA(\td\aa,\td\bb_0,e)) = \nu_{\omega}(\td\xi_{\td\aa,\td\bb_{0}, e})$.
    Now from $\phi_{\om}(\td\cA(\td\aa,\td\bb_0,e)) = 0$ and the monotonic of eigenvalues of $\td\cA(\td\aa,\td\bb_0,e)$ in Lemma \ref{lem:mono.ind.nul}, we have that $\td\cA(\td\aa,\td\bb,e) > 0$ for all $\td\bb \in [-1,\td\bb_0)$ on $D(\om, 2\pi)$. Therefore, $\nu_{\omega}(\td\xi_{\td\aa,\td\bb, e}) = \nu_{\om}(\td\cA(\td\aa,\td\bb,e)) = 0$ holds for all $\td\bb \in [-1,\td\bb_0)$. Then this lemma holds.
\end{proof}

\begin{proof}[Proof of Theorem \ref{thm:RE.lim}]
(i) Let $e\in [0,1)$ and $\td\aa \in [0,\infty)$ satisfy $\td\bb_s(\td\aa,e) < \td\bb_m(\aa,e)$.
Then Corollary \ref{coro:null.sep} implies $\nu_{-1}(\td\xi_{\td\aa,\td\bb_{m}(\td\aa,e),e}(2\pi)) =1$.
As the limit case of (i) and (ii) of \ref{thm:RE.norm.form},
we have the eigenvalues of matrix $\td\xi_{\td\aa,\td\bb_{m}(\td\aa,e),e}(2\pi)$ are all on $\U$ and the normal form is either $M\approx M_2(-1,c)$ for some $c_2 \neq 0$ or $M \approx N_1(-1,1)\diamond R(\th)$ for some $\th \in (\pi,2\pi)$.

By Lemma \ref{lem:boundary} and $\td\bb_s(\td\aa,e) < \td\bb_m(\td\aa,e)$, we have that  $M\approx M_2(-1,c)$ for some $c_2 \neq 0$ cannot holds.
The normal form is given by $M \approx N_1(-1,1)\diamond R(\th)$.
So M is spectrally stable and linearly unstable.

(ii) Let $e\in [0,1)$ and $\td\aa \in [0,\infty)$ satisfy $\td\bb_k(\td\aa,e)<\td\bb_s(\td\aa,e) = \td\bb_m(\td\aa,e)$.
As the limit case of (i) and (ii) of the Theorem \ref{thm:RE.lim} and Lemma \ref{lem:boundary},
the normal form of the matrix $M \equiv \td\xi_{\td\aa,\td\bb_s(\td\aa,e),e}(2\pi)$ is either $M \approx N_1(-1,a) \diamond N_1(-1,b)$ for some $a$, $b \in \{-1,1\}$,
or $M \approx -I_2\diamond R(\th)$ for some $\th \in (\pi,2\pi)$.
However, the case of $M \approx N_1(-1,a) \diamond N_1(-1,b)$ is impossible by Lemma \ref{lem:boundary}.
Then we have that  $M \approx -I_2 \diamond R(\th)$ for some $\th \in (\pi,2\pi)$ and it is linear stable but not strongly linear stable.

(iii) Let $e\in [0,1)$ and $\td\aa \in [0,\infty)$ satisfy $\td\bb_k(\td\aa,e)<\td\bb_s(\td\aa,e) < \td\bb_m(\td\aa,e)$.
As the limiting case of Cases (ii) and (iii) of Theorem  \ref{thm:RE.lim}, the normal form of the matrix $M \equiv \td\xi_{\td\aa,\td\bb_s(\td\aa,e),e}(2\pi)$
must satisfy either $M \approx N_1(-1,-1)\diamond R(\th)$ for some $\th\in (\pi,2\pi)$ or $M \approx M_2(-1,c)$ for some $c_2\neq 0$.
Here the second case is also impossible by Lemma \ref{lem:boundary}, and the conclusion of (iii) follows.

(iv)
Let $e\in [0,1)$ and $\td\aa \in [0,\infty)$ satisfy $\td\bb_k(\td\aa,e)<\td\bb_s(\td\aa,e) < \td\bb_m(\td\aa,e)$.
As the limiting case of the cases (iii) of Theorem \ref{thm:RE.lim},
the matrix $\td\xi_{\td\aa,\td\bb_k(\td\aa,e),e}(2\pi)$ must have Krein collision eigenvalues $\sg(M) = \{ \lm_1,\bar{\lm}_1, \lm_2,\bar{\lm}_2\}$ with $\lm_1 = \bar{\lm}_2 = e^{\sqrt{-1}\th}$ for some $\th \in (0,\pi) \cup (\pi,2\pi)$.
By the Proposition \ref{prop:bnd.ind}, Proposition \ref{prop:-1index} and the definition of $\td\bb_k(\td\aa,e)$,
$\pm 1$ cannot be the eigenvalue of $\td\xi_{\td\aa,\td\bb_k(\td\aa,e),e}(2\pi)$. Therefore, we must have that $M \approx N_2(\om,b)$ for $\om = e^{\sqrt{-1}\th}$ and some matrix
$b=
(\begin{smallmatrix}
b_1 & b_2\\ b_3 & b_4
\end{smallmatrix})$
which has the form of (25-27) by Theorem 1.6.11 in p. 34 of \cite{Lon1}.
Because $(I_2 \diamond -I_2)^{-1} N_2( e^{\sqrt{-1}\th},b)(I_2 \diamond -I_2) = N_2( e^{\sqrt{-1}(2\pi-\th)},\hat{b}) $ where
$\hat{b}=(
\begin{smallmatrix}
b_1 & -b_2\\ -b_3 & b_4
\end{smallmatrix})$.
We can always suppose $\th \in (0,\pi)$ without changing the fact  $M \approx N_2(\om,b)$.

Note that by (\ref{eqn:bbk}) and Lemma \ref{lem:mono.ind.nul}, we have $i_{\om} (\td\xi_{\td\aa,\td\bb_{k},e}) = 0$.
Suppose $b_2-b_3 =0$, by Lemma 1.9.2 in p. 43 of \cite{Lon1},
$\nu_{\om} (N_2(\om,b)) =2$ and then $N_2(\om,b)$ has basic normal form $R(\th)\diamond R(2\pi-\th)$ by the study in case 4
in p. 40 of \cite{Lon1}.
Thus we have following contradiction
\be
0=i_{\omega}\left(\td\xi_{\td\aa,\td\bb_{k}(\td\aa,e), e}\right)=i_{1}\left(\td\xi_{\td\aa,\td\bb_{k}(\td\aa,e), e}\right)+S_{M}^{+}(1)-S_{R(\theta)}^{-}(\omega)-S_{R(\theta)}^{-}(\bar{\omega}) \leq-1.
\ee
Therefore $b_2 -b_3 \neq 0$ must hold. Then we obtain
\be
0=i_{\omega}(\td\xi_{\td\aa,\td\bb_{k}(\td\aa,e), e})=i_{1}(\td\xi_{\td\aa,\td\bb_{k}(\td\aa,e), e})+S_{M}^{+}(1)-S_{N_{2}(\omega, b)}^{-}(\omega)=-S_{N_{2}(\omega, b)}^{-}(\omega).
\ee

By $\<6\>$ and $\<7\>$ in the list of splitting number in Section 2, we obtain that $N_2(\om, b)$ must be
trivial.
Then by Theorem 1 of \cite{ZhuLong}, the matrix M is spectrally stable and is linearly unstable as claimed.

(v) Let $e\in [0,1)$ and $\td\aa \in [0,\infty)$ satisfy $\td\bb_k(\td\aa,e) =\td\bb_s(\td\aa,e) < \td\bb_m(\td\aa,e)$. Note that $-1$ must be an eigenvalue of $M \approx \td\xi_{\td\aa,\td\bb_k(\td\aa,e),e}(2\pi)$ with the geometric multiplicity 1 by Proposition \ref{prop:-1.index}.
As the limit case of (ii) and (iii) of Theorem \ref{thm:RE.norm.form}, the matrix $M$ must satisfy either $M \approx M_2(-1,b)$ with $b_1$,
$b_2 \in \R$ and $b_2 \neq 0$, and thus is spectrally stable and linearly unstable;
or $M \approx N_{-1}(-1,a)\diamond D(\lm)$ for some $a\in \{-1,1\}$ and $-1\neq \lm < 0$. Then in the later case we obtain
\be
0=i_{-1}(\td\xi_{\td\aa,\td\bb_{k}(\td\aa,e), e})=i_{1}(\td\xi_{\td\aa,\td\bb_{k}(\td\aa,e), e})+S_{M}^{+}(1)-S_{N_{1}(-1, a)}^{-}(-1)=-S_{N_{1}(-1, a)}^{-}(-1).
\ee
Then by $\<3\>$ and $\<4\>$ inthe list of splitting number in Section 2, we must have $a = 1$.
Thus $M = \td\xi_{\td\aa,\td\bb_{k}(\td\aa,e), e}(2\pi)$ is hyperbolic (elliptic-hyperbolic) and linearly unstable.
Note that by the above argument, the matrix $M_2(-1, b)$ also has the basic
normal form $N_1(-1, 1)\diamond	D(\lm)$ for some $-1 \neq \lm < 0$.

(vi) Let $e\in [0,1)$ and $\td\aa \in [0,\infty)$ satisfy $\td\bb_k(\td\aa,e) =\td\bb_s(\td\aa,e) = \td\bb_m(\td\aa,e)$.
As the limiting case of
cases (i), (ii) and (iii) of Theorem \ref{thm:RE.norm.form}, $-1$ must be the only eigenvalue of
$M= \td\xi_{\td\aa,\td\bb_k(\td\aa,e),e}$ with $\nu_{-1}(M) =2$ by  Proposition \ref{prop:-1.index}. Thus the matrix $M$ must satisfy $M \approx M_1(-1,c)$ with $c =0$ and $v_{-1}(M_2(-1,c)) = 2$; or $M \approx N_1(-1,\hat{a}) \diamond N_1(-1,\hat{b})$ for some $\hat{a}$ and $\hat{b} \in \{-1,1\}$. In both case $M$ possesses the basic normal form $ N_1(-1,\hat{a}) \diamond N_1(-1,\hat{b})$ for some $\hat{a}$ and $\hat{b} \in \{-1,1\}$. Thus we obtain
\bea
 0 &=& i_{-1}\left(\td\xi_{\td\aa,\td\bb_{k}(\td\aa,e), e}\right) \nn\\ &=&i_{1}\left(\td\xi_{\td\aa,\td\bb_{k}(\td\aa,e), e}\right)+S_{M}^{+}(1)-S_{N_{1}(-1, a)}^{-}(-1)-S_{N_{1}(-1, b)}^{-}(-1) \nn\\
 &=&-S_{N_{1}(-1, a)}^{-}(-1)-S_{N_{1}(-1, b)}^{-}(-1).
\eea
Then by $\<3\>$ and $\<4\>$ inthe list of splitting number in Section 2, we must have $a = b = 1$ similar
to our above study for (v). Therefore it is spectrally stable and linearly unstable as claimed.
\end{proof}

Next, we discuss one boundary of the $\cR_{NH}$ namely $\td\bb = 0$. When $\td\bb = 0$ and $\td\aa=0$, the operator $\td\cA(0,0,e)$ is the same as the operator $A(\bb_L,e)$ given by (2.29) of \cite{HLS} when $\bb_L = 0$.
The we have that for $e\in[0,1)$, the index and nullity satisfy
(\ref{eqn.ind.00}).
The norm form of $\td\xi_{0,0,e}(2\pi)$ is given by
\be
\td\xi_{0,0,e}(2\pi)\approx I_2 \diamond N_1(1,1).
\ee
By the continuity of the eigenvalues of $\td\xi_{0,0,e}(2\pi)$ and $\nu_{-1}(\td\xi_{\td\aa,0,e}(2\pi)) = 0$ when $\td\aa\geq \frac{3}{2}$ by Proposition \ref{prop:-1index}, we must that
 $\tilde\aa_s(e)$ is the intersection curve of  $\tilde\bb_{s}(\td\aa,e) = 0$,
 $\tilde\aa_m(e)$ is the intersection curve of  $\tilde\bb_{m}(\tilde\aa,e) = 0$, and $\tilde\aa_k(e)$ is the intersection curve of  $\tilde\bb_{k}(\tilde\aa,e) = 0$.
 Furthermore, $\tilde\aa_s(0) = \tilde\aa_m(0)= \frac{1}{4}$ and $\tilde\aa_k(0) = \frac{1}{3}$.

Similiar to Proposition \ref{prop:-1.index}, we have following results and omit the proof.
\begin{proposition}\lb{prop:bb.0.-1.ind}
    When $\td\bb = 0$, the $-1$-index and nullity of $\td\xi_{\td\aa,\td\bb, e}$ satisfy
	\bea
	i_{-1}\left(\td\xi_{\td\aa,0, e}\right) &=&
	\left\{\begin{array}{l}
		{2, \text { if } 0 \leq \td\aa <\td\aa_m(e);}\\
		{1, \text { if } \td\aa_m(e) \leq \td\aa <\td\aa_s(e);} \\
		{0, \text { if } \td\aa \geq \td\aa_s(e);}
	\end{array}\right. \\
	\nu_{-1}\left(\xi_{\td\aa,0, e}\right) &=&
	\left\{\begin{array}{l}
		{2, \text { if } \td\aa=\aa_m(e) =\aa_s(e);}\\
		{1, \text { if } \td\aa= \aa_m(e) \text{ or } \aa_s(e);}\\
		{0, \text { if } \td\aa\neq \aa_m(e) \text{ and } \aa_s(e).}
		\end{array}\right.
	\eea
\end{proposition}

Following the discussion in the proof of Theorem \ref{thm:RE.lim}, by Proposition \ref{prop:bb.0.-1.ind} and Proposition \ref{prop:bnd.ind},
we can have the norm form of $\td\xi_{\td\aa,0,e}(2\pi)$
when $(\td\aa_0,\td\bb_0,e)\in\ol{\cR_{EH}}\cap \ol{\cR_{NH}}$, namely, $\td\bb = 0$. The intersection $\td\bb_*(\td\aa,e)$ with $\td\bb = 0$ are denoted by $\td\aa_*(e)$ where $* = s, m, k$ respectively.
\begin{theorem}\lb{thm:bb=0.norm}
For $e\in[0,1)$, the normal form and linear stability of  $\td\xi_{\td\aa,0,e}(2\pi)$ satisfy followings.
\begin{enumerate}[label=(\roman*)]
      \item \lb{thm:bb=0.norm.aa=0} If $\td\aa = 0$, we have $\td\xi_{\td\aa,0,e}(2\pi) \approx I_2\diamond N_1(1,1)$.
      Thus it is spectrally stable and linear unstable;

      \item if $0 < \td\aa < \td\aa_m(e)$, we have $\td\xi_{\td\aa,0,e}(2\pi) \approx R(\th) \diamond N_1(1,1)$ with $\th \in ( \pi, 2\pi)$.
      Thus it is spectrally stable and linear unstable;

      \item if $\td\aa_m(e) < \td\aa_s(e)$, we have $\td\xi_{\td\aa_m(e),0,e}(2\pi) \approx N_1(-1,1) \diamond N_1(1,1) $ with $\th \in ( \pi, 2\pi)$.
      Thus it is spectrally stable and linear unstable;

      \item if $\td\aa_m(e) < \td\aa_s(e)$ and $\td\aa_m(e) < \td\aa < \td\aa_s(e)$,
      we have $\td\xi_{\td\aa,0,e}(2\pi) \approx D(\lm) \diamond N_1(1,1) $ with $0> \lm \neq -1$.
      Thus it is spectrally unstable;

      \item if $\td\aa_m(e) < \td\aa_s(e)$, we have $\td\xi_{\td\aa,0,e}(2\pi) \approx N_1(-1, -1) \diamond N_1(1,1)$.
      Thus it is spectrally stable and linear unstable;

      \item if $\td\aa_m(e) \leq \td\aa_s(e) < \td\aa_k(e)$ and $\td\aa_s(e) < \td\aa \leq \td\aa_k(e)$, we have $\td\xi_{\td\aa,0,e}(2\pi) \approx R(\th) \diamond N_1(1,1)$ with $\th \in (0, \pi)$.
      Thus it is spectrally stable and linear unstable;

      \item if  $\td\aa > \td\aa_k(e)$, we have $\td\xi_{\td\aa,0,e}(2\pi) \approx  D(\lm)  \diamond N_1(1,1)$ with with $0< \lm \neq 1$.
      Thus it is spectrally unstable;

      \item if $\td\aa_m(e) =\td\aa_s(e) < \td\aa_k(e)$, we have $\td\xi_{\td\aa_m(e),0,e}(2\pi) \approx  -I_2\diamond N_1(1,1)$.
      Thus it is spectrally stable and linear unstable;

      \item if $\td\aa_m(e) < \td\aa_s(e) = \td\aa_k(e)$, we have $\td\xi_{\td\aa_s(e),0,e}(2\pi) \approx N_1(-1, -1) \diamond N_1(1,1)$.
      Thus $\td\xi_{\td\aa,\td\bb_{k}(\td\aa,e),e}(2\pi)$ is spectrally stable and linearly unstable;

      \item if $\td\aa_m(e) =\td\aa_s(e) = \td\aa_k(e)$,  $\td\xi_{\td\aa_k(e),0,e}(2\pi) \approx -I_2\diamond N_1(1,1)$.
      Thus $\td\xi_{\td\aa,\td\bb_{k}(\td\aa,e),e}(2\pi)$ is spectrally stable and linearly unstable.
\end{enumerate}
\end{theorem}
The proof of this theorem is similar as the one of Theorem \ref{thm:RE.lim}. We omit it here.

\setcounter{equation}{0}

\section{Stability in the Elliptic-Hyperbolic Region}
In this section, we will discuss the linear stability in $\cR_{EH}$.
Firstly, we consider the degenerate surfaces in $\cR_{EH}$.
\subsection{Degenerate surfaces}
Since $\td\cA(\td\aa,\td\bb,e)$ is a self-adjoint operator on $\ol{D}(\om,2\pi)$,
and a bounded perturbation of the operator $-\frac{\d^2}{\d t^2}I_2$,  $\td\cA(\td\aa,\td\bb,e)$ has discrete spectrum on $\ol{D}(\om,2\pi)$.
Thus we can define the $n$-th degenerate point of $\td\bb$ for $\td\aa\geq 0$, $\om\in\U$ and $e\in[0,1)$:
\be\lb{bb_s}
\td\bb_n(\td\aa,\om,e)=\min\left\{\td\bb>0\;\bigg|\;
[i_\om(\td\xi_{\td\aa,\td\bb,e})+\nu_\om(\td\xi_{\td\aa,\td\bb,e})]-[i_\om(\td\xi_{0,0,e})+\nu_\om(\td\xi_{0,0,e})]\ge n
\right\}.
\ee
Furthermore, we define that the degenerate surfaces of $\td\xi(\td\aa,\td\bb,e)(2\pi)$ by
\be \Pi_n(\om,e) = \{(\td\aa,\td\bb,e)|\td\bb=\td\bb_n(\td\aa,\om,e), \td\aa> 0,\td\bb>0, e\in[0,1)\}. \lb{eqn:td.bb.n}
\ee

Additionally, when $\td\aa =0$ and $\td\bb \to \infty$, the $i_\om(\td\xi_{\td\aa,\td\bb,e})+\nu_\om(\td\xi_{\td\aa,\td\bb,e})$ tends to infinity and when $\td\aa=\td\bb=0$, $i_\om(\td\xi_{0,0,e})+\nu_\om(\td\xi_{0,0,e}) =2$ when $\om \neq 1$.
Indeed, $\td\xi_{\td\aa,\td\bb,e}(2\pi)$ is $\om$-degenerate at surface $(\td\aa,\td\bb_n(\td\aa,\om,e),e)$ respectively, i.e.,
\be\label{degenerate.of.bn}
\nu_\om(\td\xi_{\td\aa,\td\bb_n, e})\ge 1.
\ee
Otherwise, if there existed some small enough $\ep>0$ such that $\td\bb=\td\bb_n(\td\aa,\om,e)-\ep$ would
satisfy $[i_\om(\td\xi_{\td\aa,\td\bb,e})+\nu_\om(\td\xi_{\td\aa,\td\bb,e})] - [i_\om(\td\xi_{0,0,e})+\nu_\om(\td\xi_{0,0,e})]\ge n$ in (\ref{bb_s}),
it would yield a contradiction.

By Lemma \ref{lemma.index.tdphi}, $\td\bb_n(\td\aa,\om,e)$ is non-decreasing with respect to $n$ for fixed $\td\aa$, $\om$ and $e$.
For fixed $\om$,
$\Pi_n(\om,e)$ called the $n$-th $\om$-degenerate surface is continuous surface with respect to $\td\aa$ and $e$ .
\begin{lemma}
    For any fixed  $n\in\N$ and $\om\in\U$, the degenerate surface $\Pi_n(\om,e)$ is continuous with respect to $\td\aa$ and  $e$.
\end{lemma}

\begin{proof}
    We always assume that $n$ and $\om$ are fixed. In fact, if the function $\td\bb_n(\td\aa,\om,e)$ is not continuous in $(\td\aa,e)$,
    then there exist some $(\td{\aa}_*,e_*)$, a sequence $\{(\td\aa_i,e_i)|i\in\N\}$
    and $\td\aa_0\ge0$ such that
    \be
    \td\bb_n(\td\aa_i,\om,e_i)\to\td\bb_0\ne\td\bb_n(\td\aa_*,\om,e_*) \quad{\rm and}\quad \td\aa_i\to \td\aa_*,\quad e_i\to e_*  \quad{\rm as}\quad i\to+\infty. \lb{eqn:lim.bb}
    \ee
    By (\ref{degenerate.of.bn}), we have $\om\in\sigma(\td\xi_{\td\aa_i, \td\bb_n(\td\aa_i,\om,e_i),e_i}(2\pi))$.

    By the continuity of eigenvalues of $\td\xi_{\td\aa,\td\bb,e}(2\pi)$ with respect to $\td\aa$ and $e$, and
    (\ref{eqn:lim.bb}), we have $\om\in\sigma(\td\xi_{\td{\aa},\td\bb_n,e}(2\pi))$,
    and hence
    \be\label{bb0.om-degenerate}
    \nu_\om(\td\xi_{\td{\aa},\td\bb_0,e}) \ge 1.
    \ee
    We distinguish two cases according to the sign of the difference $\td\bb_0-\td\bb_n(\bar{\aa},\om,e)$.
    For convenience, let
    \be
    g(\td\aa,\td\bb,e)=[i_\om(\td\xi_{\td\aa,\td\bb,e})+\nu_\om(\td\xi_{\td\aa,\td\bb,e})] -[i_\om(\td\xi_{0,0,e})+\nu_\om(\td\xi_{0,0,e})].\lb{eqn:g.td}
    \ee

    If $\td\bb_0<\td\bb_n(\td\aa_*,\om,e_*)$, firstly we must have $g(\td\aa_*,\td\bb_0,e_*)<n$.
    Otherwise by the definition of $\td\bb_n(\td\aa_*,\om,e_*)$,  we must have $\td\bb_n(\td\aa_*,\om,e_*)\le\td\bb_0$.

    Let $\bar\bb \in (\bb_0,\td\bb_n(\td\aa_*,\om,e_*))$ such that $\nu_\om(\td\xi_{\td\aa_*,\bar\bb,e_*})=0$
    for any $\td\bb\in(\td\bb_0,\bar\bb]$. By the continuity of eigenvalues of $\td\xi_{\td\aa,\td\bb,e}(2\pi)$ with
    respect to $\td\aa$, $\td\bb$ and $e$, there exists a neighborhood $\mathcal{O}$ of $(\td\aa_*,\bar\bb, e_*)$
    such that $\nu(\td\xi_{\td\aa,\td\bb,e})=0$ for any $(\td\aa,\td\bb,e)\in\mathcal{O}$. Then $i_\om(\td\xi_{\td\aa,\td\bb,e})$, and
    hence $g(\td\aa,\td\bb,e)$ is constant in $\mathcal{O}$.
    By (\ref{eqn:lim.bb}), for $i$ large enough,
    we have $\td\bb_n(\td\aa_i,\om,e_i)<\bar\bb$ and $(\td\aa_i,\bar\bb,e_i)\in\mathcal{O}$, and hence
    $g(\td\aa_i,\bar\bb,e_i)\ge g(\td\aa_i,\td\bb_n(\td\aa_i,\om,e_i),e_i)\ge n$.
    Therefore, we have $g(\td{\aa}_*,\bar\bb,e_*)\ge n$.
    By the definition of (\ref{bb_s}), we have $\td\bb_n(\td{\aa}_*,\om,e_*)\le\bar\bb$
    which contradicts $\bar\bb\in(\td\bb_0,\td\bb_n(\td{\aa}_*,\om,e_*))$.

    If $\td\bb_n(\td\aa_*,\om,e_*)<\td\bb_0$,
    there exists $\bar\bb\in(\td\bb_n(\td\aa_*,\om,e_*),\td\bb_0)$ such that $\nu_\om(\td\xi_{\td\aa_*,\td\bb,e_*})=0$
    for any $\td\bb\in(\td\bb_n(\td\aa_*,\om,e_*),\bar\bb]$.
    By the continuity of eigenvalues of $\td\xi_{\td\aa,\td\bb,e}(2\pi)$ with respect to $\td\aa$, $\td\bb$ and $e$,
    there exists a neighborhood $\mathcal{U}$ of $(\td\aa_*,\bar\bb,e_*)$ such that $\nu(\td\xi_{\td\aa,\td\bb,e})=0$
    for any $(\td\aa,\td\bb,e)\in\mathcal{U}$.
    Then $i_\om(\td\xi_{\td\aa,\td\bb,e})$, and hence $g(\td\aa,\td\bb,e)$ is constant in $\mathcal{U}$.
    By (\ref{eqn:lim.bb}), for $i$ large enough,
    we have $\bar\bb<\td\bb_n(\td\aa_i,\om,e_i)$ and $(\td\aa_i,\bar\bb,e_i)\in\mathcal{U}$.
    $g(\td\aa_i,\bar\bb,e_i)=g(\td\aa,\bar\bb,e)\ge n$ implies $\td\bb_n(\td\aa_i,\om,e_i)\le \bar\bb$, a contradiction.

    Thus the continuity of $\td\bb_n(\td\aa,\om,e)$ in $\td\aa$, $e$ is proved.
\end{proof}

\subsection{The elliptic-hyperbolic region}
Note that in Lemma \ref{lem:EH.e0}, we will discuss that the degenerate curves when $e= 0$.
To be consistence with the discussion in Section \ref{Stab.R3}, we use the notations $(\aa,\bb)$ instead of $(\td\aa,\td\bb)$.

\begin{lemma}\label{lem:EH.e0}
    When $\om =1$ and $e = 0$, the degenerate surfaces satisfy
    \begin{equation}
    T^{-1}\Pi_n(1,0)=\{(\aa,\bb)|\aa=-n^2-1+\sqrt{9\bb^2+4n^2}\}, \quad {\rm if}\;\;n = 2m-1\;\;{\rm or}\;\;2m. \lb{bb_n0}
    \end{equation}
\end{lemma}

\begin{proof}
    As in Section \ref{Stab.R3}, we have $i_1(\xi_{1/2,1/2,0})+\nu_1(\xi_{1/2,1/2,0})=3$, and if $(\aa,\bb)\in \cR_{3,m+1}^* $, we have $\nu_1(\aa,\bb,0)=2$. As (\ref{eqn:g.td}), we define
    \be
    g(\aa,\bb,e) =[i_1(\xi_{\aa,\bb,0})+\nu_1(\xi_{\aa,\bb,0})]-[i_1(\xi_{1/2,1/2,0})+v_1(\xi_{1/2,1/2,0})],
    \ee
     and
    \begin{equation}
     g(\aa,\bb,e)
    \left\{
    \begin{array}{l}
    \le 2m-2, \quad {\rm if}\;(\aa,\bb)\in\cR_{3,0}^+\cup\l(\cup_{i=0}^{m}\cR_{3,i}^-\cup\cR_{3,i}^* \cup \cR_{3,i}^+\r);
    \\
    =2m , \quad {\rm if}\;\;(\aa,\bb)\in\cR_{3,m+1}^*\cup\cR_{3,m+1}^+\cup\cR_{3,m+2}^-;
    \\
    \ge 2m+2, \quad {\rm if}\;\;(\aa,\bb)\in (\cup_{i \geq m + 2}^{\infty}\cR_{3,i}^*\cup \cR_{3,i}^+) \bigcup (\cup_{i = m + 3}^{\infty}\cR_{3,m+2}^-).
    \end{array}
    \right.
    \end{equation}
    For $n=2m-1$ or $2m$, $[i_1(\xi_{\aa,\bb,0})+\nu_1(\xi_{\aa,\bb,0})]-[i_1(\xi_{1/2,1/2,0})+\nu_1(\xi_{1/2,1/2,0})]\ge n$ is equivalent to $(\aa,\bb) \in  (\cup_{i \geq m + 1}^{\infty}\cR_{3,i}^*\cup \cR_{3,i}^+) \bigcup (\cup_{i = m + 2}^{\infty}\cR_{3,m+2}^-)$.
    Then the minimal value of $\bb$ in $(\aa,\bb) \in  (\cup_{i \geq m + 1}^{\infty}\cR_{3,i}^*\cup \cR_{3,i}^+) \bigcup (\cup_{i = m + 2}^{\infty}\cR_{3,m+2}^-)$ such that $\cA(\aa,\bb,e)$ is degenerate on
    $\overline{D}(1,2\pi)$ is $(\aa,\bb) \in \cR_{3,m+1}^*$. Thus applying $T^{-1}
    $, we obtain (\ref{bb_n0}).
\end{proof}

Note that when $\td\bb =0$, we have already discussed in Proposition \ref{prop:bnd.ind} and the nullity is always odd.
When $(\td\aa,\td\bb,e)\in \cR_{NH}\cup\cR_{EH}$, we have following theorem.

\begin{theorem}\label{Th:multiplicity}
    For any given $(\td\aa,\td\bb,e)\in \cR_{NH}\cup\cR_{EH}$, every
    $1$-degenerate surface has even geometric multiplicity.
\end{theorem}

\begin{proof}
    The statement has already been proved for $e=0$ in (\ref{eqn:null.e=0}).
    We will give the proof when $e \neq 0$.
    We notice that the proof will be simple if we use parameters $\aa$,$\bb$.
    Therefore, we consider the region $T^{-1}(\cR_{EH}\cup \cR_{NH}) = \{(\aa,\bb,e)|0<\bb<\aa< 3\bb, \aa\neq 3\bb-1,e\in [0,1)\}$.
    Once we have the kernel of $\cA(\aa,\bb,e)z=0$ in $\ol{D}(1,2\pi)$ has even dimensions for $(\aa,\bb,e)\in T^{-1}(\cR_{NH}\cup\cR_{EH})$,
    we have this theorem holds.

    {\bf Claim.} {\it If $\cA(\aa,\bb,e)z=0$
        has a solution $z\in\ol{D}(1,2\pi)$ for a fixed value $e\in (0,1)$, there exists a second periodic
        solution which is independent of $z$. }

    If the claim holds, then the space of solutions of $\cA(\aa,\bb,e)z=0$ is the direct
    sum of two isomorphic subspaces, hence it has even dimension.
    This method is originally due to R. Mat\'{i}nez,
    A. Sam\`{a} and C. Sim\`{o} in \cite{MSS1}.
    Here we follows the arguments from Theorem 4.8 of \cite{ZL15ARMA}.

    Let $z(t)=R(t)(x(t),y(t))^T$ be a nontrivial solution of $\cA(\aa,\bb,e)z(t)=0$, then it yields
    \begin{equation}
    \left\{
    \begin{array}{l}
    (1+e\cos t)x''(t)=(1+\aa+3\bb)x(t)+2y'(t)(1+e\cos t),
    \\
    (1+e\cos t)y''(t)=(1+\aa-3\bb) y(t)-2x'(t)(1+e\cos t).
    \end{array}
    \right.
    \end{equation}
    By Fourier expansion, $x(t)$ and $y(t)$ can be written as
    \bea
    x(t) =a_0+\sum_{n\ge1}a_n\cos nt+\sum_{n\ge1}b_n\sin nt, \nn
    \\
    y(t) =c_0+\sum_{n\ge1}c_n\cos nt+\sum_{n\ge1}d_n\sin nt. \nn
    \eea
    Then the coefficients must satisfy the following uncoupled sets of recurrences:
    \begin{equation}
    \left\{
    \begin{array}{l}
    (1+\aa+3\bb)a_0=-e(d_1+\frac{a_1}{2}),
    \\
    eA_2\begin{pmatrix}
    a_2\\d_2
    \end{pmatrix}= B_1\begin{pmatrix}
    a_1\\ d_1
    \end{pmatrix},\\
    eA_{n+1}\begin{pmatrix}
    a_{n+1}\\ d_{n+1}
    \end{pmatrix}
    =B_n\begin{pmatrix}
    a_n\\ d_n
    \end{pmatrix}
    -eA_{n-1}
    \begin{pmatrix}
    a_{n-1}\\ d_{n-1}
    \end{pmatrix},
    \quad n\ge2,
    \end{array}
    \right.
    \label{ad.equations}
    \end{equation}
    and
    \begin{equation}
    \label{bc.equations}
    \left\{
    \begin{array}{l}
    (1+\aa-3\bb)c_0=e(b_1-\frac{c_1}{2}),
    \\
    eA_2
    \begin{pmatrix}
    b_2\\ -c_2
    \end{pmatrix}
    =B_1\begin{pmatrix}
    b_1\\ -c_1
    \end{pmatrix},
    \\
    eA_{n+1}
    \begin{pmatrix}
    b_{n+1}\\ -c_{n+1}
    \end{pmatrix}
    =B_n\begin{pmatrix}
    b_{n}\\ -c_{n}
    \end{pmatrix}-eA_{n-1}
    \begin{pmatrix}
    b_{n-1}\\ -c_{n-1}
    \end{pmatrix},
    \quad n\ge2,
    \end{array}
    \right.
    \end{equation}
    where
    \begin{equation}
    A_n=-\frac{n}{2}
    \begin{pmatrix}
    n&2\\ 2 & n
    \end{pmatrix},\quad
    B_n=
    \begin{pmatrix}
    n^2+1+\aa+3\bb & 2n\\ 2n & n^2+1+\aa-3\bb
    \end{pmatrix}.
    \end{equation}

    Thus $\det(B_1)=\aa^2+4\aa-9\bb^2\ne0$ if $\aa \neq -2 + \sqrt{9\bb^2+4}$ and $\det(A_n)\ne0$ when $n\ge3$.
    Thus given $(a_2,d_2)^T$, we can obtain $(a_1,d_1)^T$ uniquely from the second equality of (\ref{ad.equations}),
    and then obtain $(a_n,d_n)^T$ for $n\ge3$ by the last equality of (\ref{ad.equations}).

    By the non-triviality of $z=z(t)$, both (\ref{ad.equations}) and (\ref{bc.equations}) have solutions
    $\{(a_n,d_n)\}_{n=1}^\infty$ and $\{(b_n,c_n)\}_{n=1}^\infty$ respectively.
    We assume (\ref{ad.equations}) admits a nontrivial solutions.
    Then $\sum_{n\ge1}a_n\cos nt$ and $\sum_{n\ge1}d_n\sin nt$ are convergent.
    Thus, $\sum_{n\ge1}a_n\sin nt$ and $-\sum_{n\ge1}d_n\cos nt$ are convergent too.
    Moreover, by the similar structure between equations (\ref{ad.equations}) and (\ref{bc.equations}),
    we can construct a new solution of (\ref{bc.equations}) given below if $\aa+1-3\bb  \neq 0$,
    \bea
    \tilde{c}_0&=&-\frac{e}{1+\aa-3\bb}(a_1+\frac{d_1}{2}),\lb{tilde.c0}
    \\
    \begin{pmatrix}
        \tilde{b}_n \\ \tilde{c}_n
    \end{pmatrix}
    &=&
    \begin{pmatrix}
        \tilde{a}_n \\ -\tilde{d}_n
    \end{pmatrix},\quad n\ge1. \lb{bn.cn}
    \eea
    Therefore we can build two independent solutions of $A(\bb,e)w=0$ as
    \bea
    w_1(t)&=&R(t)
    \begin{pmatrix}
        a_0+\sum_{n\ge1}a_n\cos nt \\\sum_{n\ge1}d_n\sin nt
    \end{pmatrix},
    \\
    w_2(t)&=&R(t)
    \begin{pmatrix}
        \sum_{n\ge1}\tilde{b}_n\sin nt \\
        \tilde{b}_0+\sum_{n\ge1}\tilde{c}_n\cos nt
    \end{pmatrix}
    =R(t)
    \begin{pmatrix}
        \sum_{n\ge1}a_n\sin nt\\
        -\frac{e}{\bb}(a_1+\frac{d_1}{2})-\sum_{n\ge1}d_n\cos nt
    \end{pmatrix}.
    \eea
    If $\aa = -2 + \sqrt{9\bb^2+4}$, then $B_1$ is degenerate.
    Note that $\bb\neq 0$.
    Then the $\{a_n,d_n\}$ satisfies (\ref{ad.equations}), we must have that $a_2+d_2=0$.
    When $a_2+d_2=0$, the $(a_1,d_1)$ is given by $(a_1,d_1)^T = (\frac{\sqrt{9\bb^2+4}-3\bb-2}{\sqrt{9\bb^2+4}+3\bb-2},1)d_1$.
    Then $\{(a_n,d_n)\}$ is obtained by (\ref{ad.equations}) and  $\{c_n,d_n\}$ can be given by (\ref{bc.equations}).

    Therefore, the claim holds and then this theorem holds.
\end{proof}

\begin{proposition}\label{Th:odd.indices}
    For any $(\td\aa,\td\bb,e) \in \cR_{EH}$, $i_1(\td\xi_{\td\aa,\td\bb,e})$ is an odd number.
\end{proposition}

\begin{proof}
    When $e=0$, the conclusion of Proposition \ref{Th:odd.indices} follows from (\ref{index.e=0}).

    Now we suppose $0<e<1$.
    For the given $(\td\aa_*, \td\bb_*)$, we can choose the path that first increase $\td\aa$ from $0$ to $\td\aa_*$ and then increase $\td\bb$  from $0$ to $\td\bb_*$.
    Suppose there are $n$ intersections with the degenerate surfaces which are defined by $(\td\aa_*,\td\bb_{n_*})$.
    Then by (ii) of Proposition \ref{prop:bnd.ind}, we have
    \begin{equation}
    i_1(\td\xi_{\td\aa_*,\td\bb_*,e}) = i_1(\td\xi_{\aa_*,0,e})+ \nu_1(\td\xi_{\aa_*,0,e})+
    \sum \nu_1(\td\xi_{\aa_*,\bb_{k,\ast},e})
    = 1 + \sum_{k=1}^n\nu_1(\td\xi_{\aa_*,\bb_{k,\ast },e}).\label{odd.sum}
    \end{equation}
    By the proof of Theorem \ref{Th:multiplicity}, every $\nu_1(\td\xi_{\aa_*,\bb_{n_*},e})$ is even for $1\le k\le n$.
    Thus by $i_1(\td\xi_{\aa_*, \bb_{\ast},e})$ is odd by (\ref{odd.sum}).
\end{proof}

Note that $\nu_1(\td\cA(\td\aa,\td\bb,e)) =0$ when $(\td\aa,\td\bb,e) \in \cR_{NH}$ and for  $(\td\aa,\td\bb,e) \in \cR_{EH}$ the intersection of $1$-degenerates surface and the $-1$-degenerate surfaces satisfy following theorem.

\begin{proposition}\lb{thm:nonintersec.-1.1}
    If $(\td\aa,\td\bb,e) \in \cR_{EH}$, any $1$-degenerate surface and any $-1$-degenerate surface cannot intersect each other.
    That is, for any $0\le e<1$, for any  $n_1$ and $n_2\in\N$,
    $\Pi_{n_1}(1,e) \cap \Pi_{n_2}(-1,e)= \emptyset$.
\end{proposition}
\begin{proof}
    If not, suppose $0\leq e_*<1$ and $(\td\aa_*,\td\bb_*)$  is an
    intersection point of some $1$-degenerate surface $\Pi_{n_1}(1,e_*)$ and a $-1$-degenerate surface $\Pi_{n_2}(-1,e_*)$.
    Then $\nu_1(\td\xi_{\td\aa_*,\td\bb_*,e_*})\ge 1$ and $\nu_{-1}(\td\xi_{\td\aa_*,\td\bb_*,e_*})\ge 1$,
    and hence $\sigma(\td\xi_{\td\aa_*,\td\bb_*,e_*})=\{1,1,-1,-1\}$.
    Therefore
    there exist $b_1$ and $b_2$ such that $\td\xi_{\aa_*,\bb_*,e_*}(2\pi)\in \Sp(4)$ satisfies:
    \begin{equation}
    \td\xi_{\aa_*,\bb_*,e_{\ast}}(2\pi)\approx N_1(1,b_1)\diamond N_1(-1,b_2).
    \end{equation}

    Moreover, by Theorem \ref{Th:multiplicity}, the integer
    $\nu_1(\td\xi_{\aa_*,\bb_*,e_*}) \ge 1$ is even. Together with the fact
    $\nu_1(N_1(1,b_1))=1$ when $b_1\ne 0$, we must have $b_1=0$.

    There exist two paths $\td\xi_i\in\P_{2\pi}(2)$ such that we have
    $\td\xi_1(2\pi)=I_2$, $\td\xi_2(2\pi)=N_1(-1,b_2)$, $\td\xi_{\aa_*,\bb_*,e_{\ast}}\sim_1 \td\xi_1 \dm \td\xi_2$, and
    $i_1(\td\xi_{\aa_*,\bb_*,e_{\ast}})=i_1(\td\xi_1)+i_1(\td\xi_2)$.
    By $1^\circ$ and $2^\circ$ of Lemma 5.6 in Appendix of \cite{ZL15ARMA},
    both $i_1(\td\xi_1)$ and $i_1(\td\xi_2)$ are odd numbers. Therefore
    $i_1(\td\xi_{\td\aa_*,\td\bb_*,e_*})$ must be even. But Proposition \ref{Th:odd.indices} yields
    $i_1(\td\xi_{\aa_*,\bb_*,e_*})$ is an odd number. It is a contradiction. Then this theorem holds.
\end{proof}

By $4^\circ$ of Lemma 5.6 in Appendix of \cite{ZL15ARMA}, $i_1(R(\th))$ is also odd number for $\th \in (0,\pi)\cup (\pi,2\pi)$. Then we obtain Theorem \ref{thm:-1.om.nointersect} and omit the proof.

\begin{theorem}\lb{thm:-1.om.nointersect}
	For $\om\neq \pm 1$, any $\om$-degenerate curve and any $-1$-degenerate
	curve cannot intersect each other in $\cR_{EH}$. That is, for any $0 \leq  e < 1$, and for any $n_1$ and $n_2\in \N$ such that $\Pi_{n_1}(\om,e) \cap \Pi_{n_2}(-1,e)= \emptyset $.
\end{theorem}

\begin{remark}
    Note that when $\td\aa = 0$, i.e., $\aa = 3\bb-1$, we have the intersection of the $1$-degenerate surface and the $-1$-degenerate.
\end{remark}

By Proposition \ref{thm:nonintersec.-1.1}, the $-1$- and $1$-degenerate surfaces in $\cR_{EH}$ cannot intersect with each other. Therefore, we can order the $-1$ and $1$-degenerate surface, i.e., $\Sigma_{n}^{\pm}$ and $\Ga_{n}$ respectively, in the region $\cR_{EH}$ of $(\td\aa,\td\bb,e)$ from left to right as
\be
\Ga_0, \Sigma_1^-, \Sigma_1^+,\Ga_1,\Sigma_2^-, \Sigma_2^+, \Ga_2,\dots, \Sigma_n^-, \Sigma_n^+,  \Ga_n,\dots.
\ee
where $T^{-1}\Ga_n|_{e=0} = \cR_{3,n}^*$ and $T^{-1}\Sigma_n^-|_{e = 0} =T^{-1}\Sigma_n^+|_{e = 0}=\cR_{3,n+\frac{1}{2}}^*$ for $n\in\N_0$.
Note that $\td\bb = 0$ is equivalent to  $\Ga_0 =\{(\td\aa,\td\bb,e)|\td\aa>0, \td\bb = 0, e\in[0,1)\}$. According to Proposition \ref{prop:bnd.ind} and (\ref{eqn:-1.index.tildeA}), the $1$- and $-1$-degenerate do intersect. Then we have that $\Ga_0$, $\Sigma_1^-$ and  $\Sigma_1^+$ intersect.

For the given $e\in[0,1)$ and $\tilde\aa\in (0, \infty)$, we have the $1$-degenerate points satisfy $\tilde\bb_{2n}(\tilde\aa,1,e)= \tilde\bb_{2n+1}(\tilde\aa,1,e)$ by Theorem \ref{Th:multiplicity} and $\tilde\bb_{2n}(\tilde\aa,-1,e)$ are the $-1$-degenerate points. In $\cR_{EH}$,
\bea
&&0<\tilde\bb_{0}(\tilde\aa,1,e) \leq\tilde\bb_{1}(\tilde\aa,-1,e) \leq \tilde\bb_{2}(\tilde\aa, -1,e) < \tilde\bb_{1}(\tilde\aa,1,e) = \tilde\bb_{2}(\tilde\aa,1,e)\nn\\
&& \quad <\tilde\bb_{3}(\tilde\aa,-1,e) \leq \tilde\bb_{4}(\tilde\aa,-1,e)< \tilde\bb_{3}(\tilde\aa,1,e) = \tilde\bb_{4}(\tilde\aa,1,e) < \dots \nn\\
&& \quad <\tilde\bb_{2n}(\tilde\aa,-1,e) \leq \tilde\bb_{2n}(\tilde\aa,-1,e) <\tilde\bb_{2n}(\tilde\aa, 1,e) = \tilde\bb_{2n}(\tilde\aa,-1,e) \leq \tilde\bb_{2n}(\tilde\aa,-1,e) \dots.
\eea

\begin{remark}\label{Remark:multiplicity.2}
    If there are some points $(\td\aa_0,\td\bb_0,e_0)$ such that $\nu_{1}(\td\xi_{\aa_0,\bb_0,e_0})\ge4$.
    Then there must exist two different $1$-degenerate curves which intersect at $(\aa_0,\bb_0,e_0)$.
    This contradicts Proposition \ref{thm:nonintersec.-1.1}.
    Thus every $1$-degenerate curve has exact geometric multiplicity $2$.
\end{remark}

\begin{proof}[Proof of Theorem \ref{Th:stability.of.EH}]
	In the proof, we omit all the subscript of $\td\aa_0$ and $\td\bb_0$ for simplicity.

	(i) Suppose that $M = \td\xi_{\td\aa, \td\bb,e_0}(2\pi)$.
	Note that for $n \geq 0$, by the definition of $\td\bb_{2n}(\td\aa,1,e_0)$ and $\td\bb_{2n}(\td\aa,-1,e_0)$,
	we obtain that
	\bea
	&&i_1(\td\xi_{\td\aa,\td\bb,e_0})=2n+1, \quad  \nu_1(\td\xi_{\td\aa,\td\bb,e_0})=0,\\
	&&i_{-1}(\td\xi_{\td\aa, \td\bb,e_0}(2\pi))=2n, \quad
	\nu_{-1}(\td\xi_{\td\aa, \td\bb,e_0}(2\pi))=0.\lb{eqn:reh.-1}
	\eea
	By (\ref{eqn:split}),  we have that
	\be
	i_{-1}(\td\xi_{\td\aa,\td\bb,e})=i_1(\td\xi_{\td\aa,\td\bb,e})+S_M^+(1)+\sum_{\om}(-S_{M}^-(e^{\sqrt{-1}\theta})
	+S_{M}^+(e^{\sqrt{-1}\theta}) )-S_M^-(-1),
	\ee
	or
	\be
	i_{-1}(\td\xi_{\td\aa,\td\bb,e})=i_1(\td\xi_{\td\aa,\td\bb,e})+S_M^+(1)+\sum_{\om}(-S_{M}^-(e^{\sqrt{-1}(2\pi-\theta)})
	+-S_{M}^+(e^{\sqrt{-1}(2\pi-\theta)}))-S_M^-(-1).\lb{eqn:reh.-1.split}
	\ee
	Then we have that
	\be
	\sum_{\om}(-S_{M}^-(e^{\sqrt{-1}\theta})
	+S_{M}^+(e^{\sqrt{-1}\theta})) =1
	\mbox{ or } \sum_{\om}(-S_{M}^-(e^{\sqrt{-1}(2\pi-\theta)})
	+-S_{M}^+(e^{\sqrt{-1}(2\pi-\theta)})) =1.
	\ee
	Therefore, it is impossible that $\td\xi_{\td\aa,\td\bb,e}(2\pi)\approx N_2(\om,b)$.
	By the list of splitting number in Section 2, we must have that
	$\td\xi_{\td\aa,\td\bb,e}(2\pi)\approx M_1\diamond M_2$ where $M_1$ and $M_2$ are two basic normal forms in $\Sp(2)$.
	By Lemma 3.2 of \cite{ZL15ARMA} there exist two paths $\td\xi_1$ and
	$\td\xi_2$ in $\P_{2\pi}(2)$ such that $\td\xi_1(2\pi)=M_1$, $\td\xi_2(2\pi)=M_2$, $\td\xi_{\td\aa,\td\bb,e}\sim_1\td\xi_1\dm\td\xi_2$,
	and $i_1(\td\xi_{\td\aa,\td\bb,e})=i_1(\td\xi_1)+i_1(\td\xi_2)$ hold.
	Notice that $\nu_1(\td\xi_{\td\aa,\td\bb,e})=0$, then $\nu_1(\td\xi_2)=0$ and hence $1\not\in \sg(M_2)$.
	Therefore we must have $\sg(M_2)\cap\U=\emptyset$ and $\alpha(M_2)=0$.
	Thus, $M_2=D(2)$.

	Similarly, we have $\pm 1\not\in\sg(M_1)$.
	Then by Lemma 5.6 in Appendix 5.2 of \cite{ZL15ARMA}, we have $M_1=D(-2)$ or $M_1=R(\theta)$ for some $\theta\in(0,\pi)\cup(\pi,2\pi)$.
	If $M_1=D(-2)$, by the properties of splitting numbers in Chapter 9 of \cite{Lon4}, specially (9.3.3) on p.204,
	we 	obtain $i_{-1}(\td\xi_{\td\aa,\td\bb,e})=i_1(\td\xi_{\td\aa,\td\bb,e})$, which contradicts (\ref{eqn:reh.-1}) and (\ref{eqn:reh.-1.split}).
	Therefore, we must have $M_1=R(\theta)$.

	If $\theta\in(0,\pi)$, we have
	$i_{-1}(\td\xi_{\td\aa,\td\bb,e})=i_1(\td\xi_{\td\aa,\td\bb,e})-S_{R(\theta)}^-(e^{\sqrt{-1}\theta})+S_{R(\theta)}^+(e^{\sqrt{-1}\theta})=2n$.
	When $\theta\in(\pi,2\pi)$, we obtain
	$i_{-1}(\td\xi_{\td\aa,\td\bb,e})=i_1(\td\xi_{\td\aa,\td\bb,e})-S_{R(\theta)}^-(e^{\sqrt{-1}(2\pi-\theta)})+S_{R(\theta)}^+(e^{\sqrt{-1}(2\pi-\theta)})=2n+1$ contradicting (\ref{eqn:reh.-1}).
	Therefore, we have $\theta\in(0,\pi)$,
	and then $\td\xi_{\td\aa,\td\bb,e}(2\pi)\approx R(\theta)\diamond D(2)$. Thus (i) is proved.

	For items (ii)-(iii) and (v)-(vii), we first prove that $\td\xi_{\td\aa,\td\bb,e}(2\pi)\approx N_2(\om,b)$ is impossible by a method similar to that in
	the proof of Theorem 1.5 (ii) or (v). By Lemma 5.3, $\td\xi_{\td\aa,\td\bb,e}(2\pi)\approx M_1\dm M_2$ must hold.
	Then we use the information of $\pm1$-indices, null $\pm1$-indices, Lemma 5.6 and (\ref{eqn:split})
	to determine the basic normal forms of $M_1$ and $M_2$. Here the details are omitted.

	For (iv), if $\td\bb_{2n+1}(\td\aa,-1, e_0) \neq \td\bb_{2n+2}(\td\aa,-1,e_0)$ and $\td\bb_{2n+1}(\td\aa,-1, e_0) < \td\bb < \td\bb_{2n+2}(\td\aa,-1,e_0)$, by the definition of the degenerate curves, the index and nullity
	\bea
	i_1(\td\xi_{\td\aa,\td\bb,e_0}(2\pi))=2n+1, \quad \nu_1(\td\xi_{\td\aa,\td\bb,e_0}(2\pi))=0,\\
	i_{-1}(\td\xi_{\td\aa,\td\bb,e_0}(2\pi))=2n+1,\quad
	\nu_{-1}(\td\xi_{\td\aa,\td\bb,e_0}(2\pi))=0.
	\eea
	Therefore, by (\ref{eqn:split}) and the list of splitting number in Section 2, we have that
	$\td\xi_{\td\aa_0,\td\bb_0,e_0}(2\pi) \approx N_2(\om,b)$ for some $\om = e^{\sqrt{-1}\th}$ with $\th \in (0, \pi) \cup (\pi, 2\pi)$ or $\td\xi_{\td\aa_0,\td\bb_0,e_0}(2\pi) \approx D(\lm)\diamond D(\lm)$ with $\lm_1$, $\lm_2 \in \R\bs\{\pm 1\}$.

	If $\td\xi_{\td\aa_0,\td\bb_0,e_0}(2\pi) \approx N_2(\om,b)$, then we must have that $\nu_{\om}(\td\xi_{\td\aa_0,\td\bb_0,e_0}(2\pi)) \neq 0$.
	Then we can find one path $\ga(t) = (\td\aa,\td\bb(t), te_0)$ for $t\in [0,1]$ linking $\ga(1) = (\td\aa_0,\td\bb(0),e_0)$ with $\ga(0) = (\td\aa_0,\td\bb(1),0)$ such that  $\nu_{\om}(\td\xi_{\td\aa_0,\td\bb(t),e_0}(2\pi)) \neq 0$ by the continuity of $\td\ga(t)$.
  By Theorem \ref{thm:-1.om.nointersect}, we have  $\td\bb_{2n+1}(\td\aa,-1,te) < \td\bb(t) < \td\bb_{2n+2}(\td\aa,-1,te_0)$.
  However, we have that $\td\bb_{2n+1}(\td\aa,-1, 0) = \td\bb_{2n+2}(\td\aa,-1,0)$.
  Then there must exits a $t_0$ such that $\td\bb_{2n+1}(\td\aa,-1, t_0e_0) = \td\bb(t)$ or $\td\bb(t) = \td\bb_{2n+2}(\td\aa,-1,t_0e_0)$.
  Then this contradicts with Theorem \ref{thm:-1.om.nointersect}.

	This yields that $\td\xi_{\td\aa_0,\td\bb_0,e_0}(2\pi) \approx D(\lm_1)\diamond D(\lm_2)$ with $\lm_1$, $\lm_2 \in \R\bs\{\pm 1\}$.
  Following a similar argument, we have that $\td\xi_{\td\aa_0,\td\bb_0,e_0}(2\pi) \approx D(2)\diamond D(-2)$.
\end{proof}

Now we can give the proof of Theorem \ref{Th:stability.of.Pi_N}.

\begin{proof}[Proof of Theorem \ref{Th:stability.of.Pi_N}]
  (i) When $m_1=m_2$, by Proposition 6.2 of \cite{Leandro},
  the relative equilibria in $I_1'\subset\Pi_N$ are spectrally unstable.
  On the other hand, by (2.2) of \cite{HO},
  $R:=I_{2n-4}+\mathcal{D}$ can be considered as
  the regularized Hessian of the central configurations.
  In fact, for $a_0\in\mathscr{E}$ which is a central configurations,
  then $I(a_0) = {1\over2}$. With respect to the mass matrix $M$ inner product,
  the Hessian of the restriction of the potential to the inertia ellipsoid,
  is given by
  \be
  D^2U|_{\mathscr{E}}(a_0)=M^{-1}D^2U(a_0)+U(a_0),
  \ee
  and thus
  \be
  R = {1\over U(a_0)}P^{-1}D^2U|_{\mathscr{E}}(a_0)P|_{w\in\R^{2n-4}}
  \ee
  with $n=4$.
  By (\ref{linearized.system.sep_1}),
  (\ref{linearized.system.main}) and $\bb_{12,0}=0$, we have
  $\sg(R)=\left\{{3+\sqrt{9-\bb}\over2},{3+\sqrt{9-\bb}\over2},\lambda_3,\lambda_4\right\}$.
  Thus the eigenvalues of the Hessian at the central configuration
  \be\label{eig.of.Hessian}
  \sg(D^2U|_{\mathscr{E}}(a_0))=\left\{{3+\sqrt{9-\bb}\over2}\mu_0,{3-\sqrt{9-\bb}\over2}\mu_0,\lambda_3\mu_0,\lambda_4\mu_0\right\},
  \ee
  where $\mu_0>0$ is given by (\ref{mu0}).
  Since $\lambda_3$ is always positive, we must have $\lambda_4= -\tilde\bb < 0$
  when $m_4$ locate on $I_1'$.

  By Proposition 4.2 of \cite{Leandro},
  all central configurations for $m_4$ in $\Pi_N$ are non-degenerate.
  Thus $\lambda_4$ does not change its signature
  when $m_4$ change its position on $\Pi_N$,
  and hence $\lambda_4=-\tilde\bb<0$ there.
  Then we have $(\td\aa_0,\td\bb_0)\in\cR_{EH}$.
  By Theorem \ref{Th:stability.of.EH},
  when $m_4$ locates in $\Pi_N$,
  the ERE is always linearly unstable.

  (ii) When $m_1 = m_2 =m_4 = 0$ and $m_3 = 1$, the central configuration is given by $q_i$ are given by $q_1 = 0$, $q_2 = 1$, $q_3  = \frac{1}{2}+ \ii(\frac{\sqrt{3}}{2}+1) \in \pt\Pi_N$.
  In this case, we have that $\aa =\bb = \frac{1}{2}$, i.e., $\td{\aa} = \td{\bb} = 0$.
  By \ref{thm:bb=0.norm.aa=0} of Theorem \ref{thm:bb=0.norm}, for all $e\in[0,1)$, $\xi(2\pi) =\td\xi(2\pi) \approx I_2 \diamond N_1(1,1)$ and it is  spectrally stable but linearly unstable.
\end{proof}

\subsection{The $\om = \pm 1$ degenerate surfaces}
In this section, we will discuss that the bifurcation of degenerate curves at $e=0$.
Therefore, we will use the notation $\aa$ and $\bb$ instead of the $\td\aa$ and $\td\bb$ and suppose that $\aa\geq \bb>0$.

Recall when $\aa > 0$, $\cA(\aa,0,e)$ is positive definite on $\overline{D}(1,2\pi)$.
Now we set
\be
B(\aa,\om,e)=\cA(\aa, 0,e)^{-\frac{1}{2}}
\left(\frac{3S(t)}{2(1+e\cos t)}\right)
\cA(\aa, 0,e)^{-\frac{1}{2}}.\lb{eqn:B}
\ee
Because $\cA(\aa, 0,e)$ and $\frac{3S(t)}{2(1+e\cos t)}$ are self-adjoint, $B(\bb,e)$ is also self-adjoint.
Also, $\cA(\aa, 0,e)^{-\frac{1}{2}}$ is a compact operator
and $\frac{3S(t)}{2(1+e\cos t)}$ is a bounded operator,
hence one can apply Theorem 4.8 in p.158 of \cite{Ka} and conclude that $B(\aa,\om,e)$ is a compact operator.
Then we have
\begin{lemma}\label{L4.6}
	For $0\le e<1$, $\cA(\aa, \bb,e)$ is $\om$-degenerate if and only if $-\frac{1}{\bb}$ is an eigenvalue of $B(\aa,\om,e)$.
\end{lemma}

\begin{proof}
	Suppose $\cA(\aa,\bb,e)x=0$ holds for some $x\in\overline{D}(1,2\pi)$.
	Let $y=\cA(\aa,0,e)^{\frac{1}{2}}x$. Then by (\ref{eqn:B}) we obtain
	\bea
	\cA(\aa,0,e)^{\frac{1}{2}}\left(\frac{1}{\bb}+B(\aa,\bb,e)\right)y
	=\left(\frac{\cA(\aa,0,e)}{\bb}+\frac{3S(t)}{1+e\cos t}\right)x
	=\frac{1}{\bb}\cA(\aa,\bb,e)x =0  \label{A.B}
	\eea

	Conversely, if $(\frac{1}{\bb}+B(\aa,\om,e))y=0$,
	then $x=\cA(\aa,0,e)^{-\frac{1}{2}}y$ is an eigenfunction of $\cA(\aa,\bb,e)$ belonging to the eigenvalue $0$
	by our computations (\ref{A.B}).
\end{proof}

Although $e<0$ does not have physical meaning, we can extend the fundamental solution to the case $e\in(-1,1)$
mathematically and all the above results which hold for $e>0$ and also hold for $e<0$. By (\ref{eqn:td.bb.n}), the degenerate surface in $(\aa,\bb,e)$ can be given by $(\aa,\bb_n(\aa,1,e))^T = T^{-1}(\td\aa,\td\bb_n(\aa,1,e))$.
Then we have
\begin{theorem}\label{Th:analytic}
	For $\om\in\U$, $n \geq 0$ and $e\in(-1,1)$, there exist two analytic $\om$-degenerate surfaces $(\aa,h_i(\aa,\om,e),e)$ in  with $i = 1,2$ such that $(\aa,h_i(\aa,\om,e),e)$ is between $T^{-1}\Pi_{2n}(1,e)$ and $T^{-1}\Pi_{2n+1}(1,e)$. Specially, each $h_i(\aa,\om,e)$ is a really analytic function in $e\in(-1,1)$ and $\bb_{2n}(\aa,1,e) \leq h_i(\aa,\om,e) \leq \bb_{2n+1}(\aa,1,e)$ for given $\aa$ and $e$. Moreover, $\xi_{\aa,h_i(\aa,\om,e),e}(2\pi)$ is $\om$-degenerate for $i = 1,2$.
\end{theorem}

\begin{proof}
	For $(\aa,\bb)$ satisfies $T(\aa,\bb) =(\td\aa,\td\bb)$, we have
	\be
	i_1(\xi_{\aa,\bb,e})=2n+1,\quad \nu_1(\xi_{\aa,\bb,e})=0.
	\ee
	Moreover, from (i) of Theorem \ref{thm:REH.norm}, we have
	\be\label{boundary}
	\xi_{\aa,\bb,e}\approx I_2\diamond D(2),\quad \bb=\bb_{2n}(\aa,1,e)\;\;{\rm or}\;\;\bb_{2n+1}(\aa,1,e).
	\ee
	Then for $\om\in\U\backslash\{1\}$, we have
	\bea
	i_\om(\xi_{\aa, \bb_{2n}(\aa,1,e),e})=i_1(\xi_{\aa,\bb_{2n}(\aa,1,e),e})+S_{\xi_{\aa,\bb_{2n}(\aa,1,e),e}(2\pi)}^+(1)
	=2n-1+S_{I_2}^+(1)
	=2n. \label{left.om.index}
	\eea
	Similarly, we have
	\be
	i_\om(\xi_{\aa,\bb_{2n+1}(\aa,1,e),e})=2n+2.\label{right.om.index}
	\ee
	Therefore, by Lemma \ref{lem:A.mono}, it shows that, for fixed $e\in(-1,1)$,
	exactly two values $\bb=h_1(\aa,\om,e)$ and $h_2(\aa,\om,e)$ are in the interval $[\bb_{2n}(\aa,1,e),\bb_{2n+1}(\aa,1,e)]$
	at which (\ref{A.B}) is satisfied. Then $\cA(\aa,\bb,e)$ at these two values is $\om$-degenerate.
	Note that these two $\bb$ values are possibly equal to each other at some $\aa$ and $e$.
	Moreover, (\ref{boundary}) implies that $h_i(\aa,\om,e)\ne \bb_{2n}(\aa,1,e)$ or $\bb_{2n+1}(\aa,1,e)$ for $i=1, 2$.

	By (i) of Lemma \ref{lem:A.mono}, $-\frac{1}{h_i(\aa,\om,e)}$ is an eigenvalue of $B(\aa,e,\om)$.
	Note that $B(\aa,e,\om)$ is a compact operator and self-adjoint when $e$ is real.
	Moreover, it depends analytically on $\aa$ and $e$. By Theorem 3.9 of \cite{Ka},
	$-\frac{1}{h_i(\aa,\om,e)}$ is analytic in $e$ for each $i\in\N$. This implies that
	both $h_1(\aa,\om,e)$ and $h_2(\aa,\om,e)$ are real analytic functions in $\aa$ and $e$.
\end{proof}

\begin{theorem}\label{Th:analytic.om}
	For $\om \in \U$, every $\om$-degenerate curve $(\aa,\bb_n(\aa,\om,e),e)$ in $e\in (-1,1)$ is a piecewise analytic function.
	The set of $e\in(-1,1)$ such that $\bb_{2n+1}(\aa,\om,e)=\bb_{2n+2}(\aa, \om,e)$ is discrete or equal to the whole interval
	$(-1,1)$. In the first case, the functions $(\aa,e)\mapsto\bb_i(\aa,\om,e)$ with $i=2n+1$ and $2n+2$ are analytic
	for those $e$ when $\bb_{2n+1}(\aa, \om,e)<\bb_{2n+2}(\aa,\om,e)$. In the second case,
	$(\aa,e)\mapsto\bb_{2n+1}(\aa,\om,e)=\bb_{2n+2}(\aa,\om,e)$ are analytic everywhere.
\end{theorem}

By direct computations, we have Lemma \ref{lem:para.orth}.
\begin{lemma}\lb{lem:para.orth}
	By the definition (\ref{cA}) of $\cA(\aa,\bb,e)$, we have following results.
	\bea
	\left.\frac{\pt}{\pt\bb}\cA(\aa,\bb,e)\right|_{(\aa,\bb,e)=(\aa_0,\bb_0,0)}
	&=& \left.R(t)\frac{\pt}{\pt\bb}K_{\aa,\bb,e}(t)\right|_{(\aa,\bb,e)=(\aa_0,\bb_0,0)}R(t)^T,  \lb{7.22a}\\
	\left.\frac{\pt}{\pt e}\cA(\aa,\bb,e)\right|_{(\aa,\bb,e)=(\aa_0,\bb_0,0)}
	&=& \left.R(t)\frac{\pt}{\pt e}K_{\aa,\bb,e}(t)\right|_{(\aa,\bb,e)=(\aa_0,\bb_0,0)}R(t)^T,  \lb{7.23a}\eea
	where $R(t)$ is given in (\ref{cA}). If $e = 0$, we further have that
	\bea
	&& \frac{\pt}{\pt\bb}K_{\aa,\bb,e}(t)\left|_{(\aa,\bb,e)=(\aa_0,\bb_0,0)}
	= \begin{pmatrix}
		3 & 0\\ 0 &  -3
	\end{pmatrix},\right.   \lb{7.24a}\\
	&& \frac{\pt}{\pt e}K_{\aa,\bb,e}(t)\left|_{(\aa,\bb,e)=(\hat\bb_n,0)}
	= {-\cos t}
	\begin{pmatrix}
		1+\aa+3\bb& 0\\ 0 &  1+\aa-3\bb
	\end{pmatrix}.\right.   \lb{7.25a}\eea
\end{lemma}

\begin{theorem}\label{Th:orth}
	\begin{enumerate}[label=(\roman*)]
		\item For $n\geq 0$, every $1$-degenerate surface starts from the curve $\cR_{3,n}^*$ and satisfies
		\be
		T_p\Ga_n = \span\l\{\frac{\pt \bb_n(\aa,1,0)}{\pt \aa},0\r\},
		\ee
		where $p\in\cR_{3,n}^*$ and $\frac{\pt \bb(\aa,-1,0)}{\pt \aa}|_{(\aa,\bb)\in\cR_{3,n}^*} = \frac{\aa+n^2+1}{9\sqrt{(\aa+n^2+1)^2-1}}$.

		\item If $ n = 0$, the two degenerated surfaces $\Sigma_1$ and $\Sigma_2$ bifurcate from the curve $\cR_{3,\frac{1}{2}}^*$, with two different tangent spaces satisfy
		\be
		T_p\Sigma_1  = \span\l\{\l.\frac{\pt \bb_1(\aa,-1,0)}{\pt \aa}\r|_{e=0},  \frac{1}{24},
		\r\}\; \mbox{ and } \;
		T_p\Sigma_2  =\span \l\{\l.\frac{\pt \bb_1(\aa,-1,0)}{\pt \aa}\r|_{e=0},-\frac{1}{24}, \r\},
		\ee
		where $p\in \cR_{3,\frac{1}{2}}^*$ and $\frac{\pt \bb_1(\aa,-1,0)}{\pt \aa}|_{e=0} =\frac{\aa+5/4}{9\sqrt{(\aa+5/4)^2-1}}$.

		\item if $n \geq 1$, the two degenerated surfaces $\Sigma_n^+$ and $\Sigma_n^-$ start from $\cR_{3,n+\frac{1}{2}}^*$ and have the same degenerated tangent space
		\be
		T_p\Sigma_{2n} = T_p\Sigma_{2n+1}= \span\l\{\frac{\pt \bb_n(\aa,-1,0)}{\pt \aa},0\r\},
		\ee
		where $p\in\cR_{n+\frac{1}{2}}^*$ and $\frac{\pt \bb_n(\aa,-1,0)}{\pt \aa}|_{e=0}=\frac{\aa+(n+1/2)^2+1}{9\sqrt{(\aa+(n+1/2)^2+1)^2-4(n+1/2)^2}}$.
	\end{enumerate}
\end{theorem}
\begin{proof}
	(i) As (\ref{bb_s}), we have that the degenerate surfaces are given by $(\aa,\bb_n(\aa,1,e),e)$ and especially when $e =0$, $(\aa,\bb_n)$ satisfies $\aa =- (n^2+1) +\sqrt{9\bb_n^2+4n}$ and $\aa>\bb$.

	Let $(\aa, \bb_n(\aa,1,e),e)$ be the surface
	which  intersect with the plane $(\aa,\bb,0)$ with the line $\aa =- (n^2+1) +\sqrt{9\bb^2+4n}$. Let
	$e\in(-\epsilon,\epsilon)$ for some small $\epsilon>0$ and $x_e\in\ol{D}(1,2\pi)$ be the corresponding
	eigenvector, that is,
	\be  \cA(\aa,\bb(\aa,1,e),e)x_e=0.  \ee

	Recalling (\ref{A14}) and (\ref{A23}), in the plane $e = 0$, $\cA(\aa, \bb,0)$
	is degenerate when for $\aa =- (n^2+1) +\sqrt{9\bb^2+4n}$ and
	\begin{equation}
	\ker \cA(\aa,\bb,0)=\span\left\{
	R(t)(a_n\sin nt, \cos nt)^T,
	R(t)(a_n\sin nt, -\cos nt)^T
	\right\},
	\end{equation}
	with $a_n\in\R$.
	By (\ref{A23}), the equation system $\cA(\aa,\bb,0)R(t)(a_n\sin nt, \cos nt)^T=0$ reads
	\be\lb{a_n.equation}
	\left\{
	\begin{array}{cr}
		n^2a_n-2n+(1+\aa+3\bb)a_n&=0, \\
		n^2-2na_n-(1+\aa-3\bb)&=0. \end{array}
	\right.
	\ee
	Then by direct computations, $a_n = \frac{n^2+1+\aa-3\bb}{2n}$ and $\aa =- (n^2+1) +\sqrt{9\bb^2+4n}$.

	By (\ref{A0})-(\ref{A23}), if we set
	\bea
	\xi_1=(a_n\sin nt,\cos nt)^T \quad \text{and} \quad
	\xi_2=(a_n\cos nt,-\sin nt)^T, \nonumber
	\eea
	we can suppose when  $e= 0$, $x_e = x_0$ is given by
	\be
	x_0=\lambda_1R(t)\xi_1+\lambda_2R(t)\xi_2,\lb{7.20a}
	\ee
	where $\lambda_1,\lambda_2\in\R$ satisfy $\lambda_1^2+\lambda_2^2\ne0$.
	There holds
	\be
	\<\cA(\aa,\bb,e)x_e,x_e\>=0.\label{Axx1}
	\ee

	Differentiating both side of (\ref{Axx1}) with respect to $e$ yields
	\be  \frac{\pt \bb}{\pt e}\<\frac{\pt}{\pt \bb}\cA(\aa, \bb,e)x_e,x_e\> + \<\frac{\pt}{\pt e}\cA(\aa,\bb,e)x_e,x_e\>
	+ 2\<\cA(\aa,\bb,e)x_e,x'_e\> = 0,  \lb{eqn:A.diff.e}\ee
	where  $\aa'(e)$, $\bb'(e)$ and $x'_e$ denote the derivatives with respect to $e$. Then evaluating both sides at $e=0$ yields
	\be  \frac{\pt \bb}{\pt e}\<\frac{\pt}{\pt \bb} \cA(\aa, \bb,e)x_0,x_0\> + \<\frac{\pt}{\pt e}\cA(\aa, \bb,e)x_0,x_0\> = 0. \lb{7.21a}\ee

	By Lemma \ref{lem:para.orth}, we have
	\bea  \<\frac{\pt}{\pt\bb}\cA(\aa,\bb,0)R(t)\xi_1,R(t)\xi_1\>
	&=& \<\frac{\pt}{\pt\bb}K_{\aa,\bb,0}\xi_1,\xi_1\>    \nn\\
	&=& 3\int_0^{2\pi} \l( a_n^2\sin^2nt-\cos^2nt \r) \d t  \nn\\
	&=& 3\pi(a_n^2-1),  \lb{eqn:A.part.B,1}\\
	\<\frac{\pt}{\pt\bb}\cA(\aa,\bb,0)R(t)\xi_1,R(t)\xi_2\>
	&=& 3\int_0^{2\pi} \l(a_n^2\sin nt\cos nt
	-\sin nt\cos nt \r) \d t=0.  \lb{eqn:A.part.B,2}\\
	\<\frac{\pt}{\pt\bb}\cA(\aa,\bb,0)R(t)\xi_2,R(t)\xi_2\>
	&=& 3\int_0^{2\pi} \l( a_n^2\cos^2nt
	-\sin^2nt \r)\d t=3\pi(a_n^2-1), \lb{eqn:A.part.B,3}
	\eea
	and hence
	\bea
	\<\frac{\pt}{\pt\bb}\cA(\aa,\bb,0)x_0,R(t)x_0\>&=& \lm_1^2\<\frac{\pt}{\pt\bb}\cA(\aa,\bb,0)R(t)\xi_1,R(t)\xi_1\>
	+2\lm_1\lm_2\<\frac{\pt}{\pt\bb}\cA(\aa,\bb,0)R(t)\xi_1,R(t)\xi_2\>\nn\\
	&&\quad+\lm_2^2\<\frac{\pt}{\pt\bb}\cA(\aa,\bb,0)R(t)\xi_2,R(t)\xi_2\> \nn\\
	&=& 3(\lm_1^2+\lm_2^2)\pi(a_n^2-1). \lb{7.26a}
	\eea
	Similarly, we have
	\bea  \<\frac{\pt}{\pt e}\cA(\aa,\bb,0)R(t)\xi_1,R(t)\xi_1\>
	&=& \<\frac{\pt}{\pt e}K_{\aa,\bb,0}\xi_1,\xi_1\>    \nn\\
	&=& -\int_0^{2\pi} \l(\cos t(a_n^2(1+\aa+3\bb)\sin^2nt+(1+\aa-3\bb)\cos^2nt) \r)\d t  \nn\\
	&=& 0,  \lb{eqn:A.part.e,1}\\
	\<\frac{\pt}{\pt e}\cA(\aa,\bb,0)R(t)\xi_1,R(t)\xi_2\>
	&=& -\int_0^{2\pi}\l(\cos t(a_n^2(1+\aa+3\bb)+(1+\aa-3\bb) )\sin nt\cos nt\r) \d t \nn \\
	&=&0, \lb{eqn:A.part.e,2}\\
	\<\frac{\pt}{\pt e}\cA(\aa,\bb,0)R(t)\xi_2,R(t)\xi_2\>
	&=& -\int_0^{2\pi} \l( \cos t(a_n^2(1+\aa+3\bb)\cos^2nt  +(1+\aa-3\bb)\sin^2nt \r)\d t\nn \\
	&=&0, \lb{eqn:A.part.e,3}
	\eea
	Therefore,
	\bea  \<\frac{\pt}{\pt e}\cA(\aa,\bb,0)x_0,x_0\>
	= 0.  \lb{7.27a}\eea
	Therefore by (\ref{7.21a}) and (\ref{7.26a})-(\ref{7.27a}),
	together with (\ref{7.26a}), we have that
	\be
	3\frac{\pt \bb}{\pt e}(a_n^2+1) = 0.
	\ee
	Then we have that for any $\aa_n =- (n^2+1) +\sqrt{9\bb^2+4n}$,
	\be \l.\frac{\pt \bb(\aa,1,e)}{\pt e}\r|_{e=0} =0.\ee
	Then we have (i) of this theorem holds.

	(ii) As the discussion in (\ref{bb_s}), we have that the degenerate surface is given by $(\aa,\bb,e) =T^{-1}(\td\aa,\td\bb_n(\td\aa,-1,e),e)$ and especially when $e =0$, $(\aa,\bb) \in \cR_{n+\frac{1}{2}}^*$.

	Similar as the (i) of this proof, let
	$e\in(-\epsilon,\epsilon)$ for some small $\epsilon>0$ and $x_e\in\ol{D}(-1,2\pi)$ be the corresponding
	eigenvector, that is,
	\be  \cA(\aa,\bb(\aa,-1,e),e)x_e=0.  \ee

	Recalling (\ref{A14}) and (\ref{A23}), in the plane $e = 0$, $\cA(\aa, \bb,0)$ when $(\aa,\bb) \in \cR_{n+\frac{1}{2}}^*$
	is degenerate and the kernel is given by
	\begin{equation}
	\ker \cA(\aa,\bb,0)=\span\left\{
	R(t)(\tilde{a}_n\sin (n+\frac{1}{2})t, \cos (n+\frac{1}{2})t)^T,
	R(t)(\tilde{a}_n\sin (n+\frac{1}{2})t,
	-\cos (n+\frac{1}{2})t)^T
	\right\},\nn
	\end{equation}
	with $\td a_n\in\R$.
	By (\ref{A23}), the equation system $\cA(\aa,\bb,0)R(t)(a_n\sin (n+\frac{1}{2})t, \cos (n+\frac{1}{2})t)^T=0$ reads
	\be\lb{a_n+1/2.equation}
	\left\{
	\begin{array}{cr}
		(n+\frac{1}{2})^2\tilde{a}_n-2(n+\frac{1}{2})+(1++\aa+3\bb)\tilde{a}_n&=0, \\
		(n+\frac{1}{2})^2-2(n+\frac{1}{2})\tilde{a}_n-(1+\aa-3\bb)&=0. \end{array}
	\right.
	\ee
	Then by direct computations, $\tilde{a}_n = \frac{(n+\frac{1}{2})^2+1+\aa-3\bb}{2(n+\frac{1}{2})}$ and
	$\bb = 	\frac{\sqrt{(\aa+(n+1/2)^2+1)^2-4(n+1/2)^2}}{9}$.

	By (\ref{A0})-(\ref{A23}), we also set
	\bea
	\xi_1 = (\tilde{a}_n\sin (n+\frac{1}{2})t,\cos (n+\frac{1}{2})t)^T,  \quad
	\xi_2 = (\tilde{a}_n\cos (n+\frac{1}{2})t,-\sin (n+\frac{1}{2})t)^T, \nonumber
	\eea
    Therefore, (\ref{7.20a}) and (\ref{Axx1}) hold.
	Differentiating both side of (\ref{Axx1}) with respect to $e$ yields (\ref{eqn:A.diff.e}) and (\ref{7.21a}) hold.
 	Similiar with the calculation (\ref{eqn:A.part.B,1} - \ref{eqn:A.part.B,3}), we have
	\bea
	&&\<\frac{\pt}{\pt\bb}\cA(\aa,\bb,0)R(t)\xi_1,R(t)\xi_1\>
	= 3\pi(\tilde{a}_n^2-1),  \nn\\
 	&&\<\frac{\pt}{\pt\bb}\cA(\aa,\bb,0)R(t)\xi_1,R(t)\xi_2\>
 	=  0,  \nn\\
 	&&\<\frac{\pt}{\pt\bb}\cA(\aa,\bb,0)R(t)\xi_2,R(t)\xi_2\>
 	=  3\pi(\tilde{a}_n^2-1), \nn
	\eea
	and hence
	\be
	\<\frac{\pt}{\pt\bb}\cA(\aa,\bb,0)x_0,R(t)x_0\>= 3(\lm_1^2+\lm_2^2)\pi(\tilde{a}_n^2-1). \lb{eqn:A.part.bb}
	\ee
	When $ n = 0$, similarly (\ref{eqn:A.part.e,1}- \ref{eqn:A.part.e,3}), by direct computations, we have
	\bea  \<\frac{\pt}{\pt e}\cA(\aa,\bb,0)R(t)\xi_1,R(t)\xi_1\>
	&=& -\int_0^{2\pi}\cos t(\tilde{a}_0^2(1+\aa+3\bb)\sin^2\frac{t}{2}+(1+\aa-3\bb)\cos^2\frac{t}{2})\d t  \nn\\
	&=&\frac{\pi}{2}(\tilde{a}_0^2(1+\aa+3\bb) -(1+\aa-3\bb)) ,  \nn\\
	\<\frac{\pt}{\pt e}\cA(\aa,\bb,0)R(t)\xi_1,R(t)\xi_2\>&=& -\int_0^{2\pi}(\tilde{a}_0^2(1+\aa+3\bb)
	+(1+\aa-3\bb))\cos t \sin \frac{t}{2}\cos \frac{t}{2}\d t \nn \\
	&=& 0,  \nn\\
	\<\frac{\pt}{\pt e}\cA(\aa,\bb,0)R(t)\xi_2,R(t)\xi_2\>&=& -\int_0^{2\pi}\cos t(\tilde{a}_n^2(1+\aa+3\bb)\cos^2\frac{t}{2}
	+(1+\aa-3\bb)\sin^2\frac{t}{2})\d t\nn \\
	&=&-\frac{\pi}{2}(\tilde{a}_0^2(1+\aa+3\bb) -(1+\aa-3\bb)),
	\eea
	where $\tilde{a}_0 = \frac{5}{4}+ \aa-3\bb$ and $\bb =
	\frac{\sqrt{(\aa+5/4)^2-1}}{9}$.
	Therefore, we have that for $n = 0$, we have that
	\bea  \<\frac{\pt}{\pt e}\cA(\aa,\bb,0)x_0,x_0\>
	&=&\lambda_1^2\<\frac{\pt}{\pt e}K_{\aa,\bb,0}\xi_1,\xi_1\>
	+2\lambda_1\lambda_2\<\frac{\pt}{\pt e}K_{\aa,\bb,0}\xi_1,\xi_2\>
	+\lambda_2^2\<\frac{\pt}{\pt e}K_{\aa,\bb,0}\xi_2,\xi_2\>\nn\\
	&=&\frac{\pi}{2} (\lm_1^2-\lm_2^2)(\tilde{a}_0^2(1+\aa+3\bb) -(1+\aa-3\bb)).
  \eea
	By (\ref{7.21a}), for any $\lm_1$ and $\lm_2$ satisfying $\lm_1^2+\lm_2^2 \neq 0$,
	\be
	3\frac{\pt \bb}{\pt e}(\lm_1^2+\lm_1^2)(\tilde{a}_0^2-1) +\frac{1}{2}(\lm_1^2-\lm_2^2)(\tilde{a}_0^2(1+\aa+3\bb) -(1+\aa-3\bb)) = 0.
	\ee
	This yields that
	\bea
	3\frac{\pt \bb_1}{\pt e}(0)(\tilde{a}_0^2-1) +\frac{1}{2}(\tilde{a}_0^2(1+\aa+3\bb) -(1+\aa-3\bb)) = 0,\\
	3\frac{\pt \bb_2}{\pt e}(0)(\tilde{a}_0^2-1) -\frac{1}{2}(\tilde{a}_0^2(1+\aa+3\bb) -(1+\aa-3\bb)) = 0.
	\eea
	where for $\bb > 0$, $	\tilde{a}_0^2(1+\aa+3\bb) -(1+\aa-3\bb) = \frac{3\bb}{2}(\sqrt{9\bb^2+1}-3\bb) \neq 0$.

	Therefore, we have  $\frac{\pt \bb_1}{\pt e}$ and $\frac{\pt \bb_2}{\pt e}$ satisfy that
	\bea
	3\frac{\pt \bb_1}{\pt e}(\tilde{a}_0^2-1) +\frac{1}{2}(\tilde{a}_0^2(1+\aa+3\bb) -(1+\aa-3\bb)) = 0,\\
	3\frac{\pt \bb_2}{\pt e}(\tilde{a}_0^2-1) -\frac{1}{2}(\tilde{a}_0^2(1+\aa+3\bb) -(1+\aa-3\bb)) = 0.
	\eea
  When $n = 0$, the $-1$-degenerate curve bifurcates at the line $\cR_{3,\frac{1}{2}}^*$ when $e = 0$.
    By direct computations, when $\aa\geq \frac{1}{2}$, we have that $\tilde{a}_0^2-1 \neq 0$. Therefore, we have that
	\bea
	\l.\frac{\pt \bb_1(\aa,-1,e)}{\pt e}\r|_{\aa= \aa_0, e=0} = \frac{1}{24}, \quad
	\l.\frac{\pt \bb_2(\aa,-1,e)}{\pt e}\r|_{\aa= \aa_0, e=0} = -\frac{1}{24}.
	\eea

	Similarly, for $n\geq 1$,  by direct computations, we have for $i, j\in\{1,2\}$
		\bea
	\<\frac{\pt}{\pt e}\cA(\aa,\bb,0)R(t)\xi_i,R(t)\xi_j\>
=0,  \lb{eqn:6.73}
	\eea
	Therefore, we have that
	\bea  \<\frac{\pt}{\pt e}\cA(\aa,\bb,0)x_0,x_0\>
	= 0.  \lb{6.76}\eea
	By (\ref{7.21a}) and (\ref{6.76}),
	together with (\ref{eqn:A.part.bb}), we have that for $n\geq 1$,
	\be
	3\frac{\pt \bb}{\pt e}(\tilde{a}_n^2+1) = 0.
	\ee
	This yields that
	\be
	\l.\frac{\pt \bb(\aa,-1,e)}{\pt e}\right|_{e= 0} = 0.
	\ee
	Then we have (ii) and (iii) of this theorem.
\end{proof}

\setcounter{equation}{0}
\section{The Special Cases with Two Equal Masses}
In this section, we apply the discussion in Section 1-7 on the
linear stability of restricted 4-body elliptic solutions.
When $m_1= m_2 = m$ and $m_3 = 1-2m$. By the symmetry of the central configuration, we can assume that $q_{4,0} = \frac{1}{2} + \sqrt{-1} y $ and $y$ satisfies following central configuration equation.
\be
\frac{-m_1 y}{(\frac{1}{4}+y^2)^{\frac{3}{2}}} + \frac{-m_2 y}{(\frac{1}{4}+y^2)^{\frac{3}{2}}} + \frac{m_3 (\frac{\sqrt{3}}{2}-y)}{|\frac{\sqrt{3}}{2}-y|^3} + y -\frac{\sqrt{3}}{2}m_3 = 0,
\ee
where $y \in Y_1\cup Y_2 \cup Y_3$ with $ Y_1 \equiv(-\frac{\sqrt{3}}{2},\frac{\sqrt{3}}{2}-1)$, $Y_2 \equiv(0,\frac{\sqrt{3}}{2})$ and $Y_3 \equiv(\frac{\sqrt{3}}{2},\frac{\sqrt{3}}{2}+1)$ by \cite{Leandro}.
By $m_1 = m_2 = m$ and $m_3 =1-2m$, we have that
\be
m= \l(\frac{\sqrt{3}}{2}-y-\frac{\frac{\sqrt{3}}{2}-y}{\left|\frac{\sqrt{3}}{2}-y\right|^3}\r)\l(\sqrt{3}-\frac{2(\frac{\sqrt{3}}{2}-y)}{\left| \frac{\sqrt{3}}{2}-y\right| ^3}-\frac{2 y}{\left(y^2+\frac{1}{4}\right)^{3/2}}\r)^{-1}. \lb{eqn:equ.m.m}
\ee
By the definition of $\aa$ and $\bb$ by (\ref{aa.bb}), we have that in this case $\aa$ and $\bb$ are given by
\be
\aa = \frac{1}{2}\left(\frac{2m}{\left(y^2+\frac{1}{4}\right)^{3/2}}+\frac{1-2m}{|y -\frac{\sqrt{3}}{2}|^{3}}\right) \quad
\mbox{and} \quad
\bb = \frac{1}{2}\left|\frac{m(\frac{1}{2}-2y^2)}{\left(y^2+\frac{1}{4}\right)^{5/2}}-\frac{1-2m}{|y -\frac{\sqrt{3}}{2}|^{3}}\right|.\lb{eqn:equ.m.aa.bb}
\ee
Plugging (\ref{eqn:equ.m.m}) into (\ref{eqn:equ.m.aa.bb}), we have that $\aa(m.y)$ and $\bb(m,y)$ can by given by $\aa =\aa(y)$ and $\bb=\bb(y)$ with $y \in Y_1\cup Y_2 \cup Y_3$. When $y \in  Y_1$, $q_{4,0}\in \Pi_C$; when $y \in Y_2$, $q_{4,0}\in \Pi_I$;  when $y \in Y_3$, $q_{4,0}\in \Pi_N$.
Note that when $y \in Y_3$ has been discussed in Theorem \ref{Th:stability.of.Pi_N}. Then by the assistance of
numerical computations, we have following stability when $y\in Y_1 \cup Y_2\cup Y_3$.

\begin{theorem}\lb{thm:two.equal.numerical}
	(i) If $y \in Y_1$, $\aa>3\bb-1$.
	There exist $y_{1,1}\approx -0.6724$ and $y_{1,2} \approx -0.1590$ such that when $y \in [y_{1,1}, y_{1,2}]$ and $e\in [0,1)$,
	the essential part of the system possesses two pairs of hyperbolic eigenvalues.
	When $y\in Y_1\bs[y_{1,1}, y_{1,2}]$,
	$(\aa,\bb,e) \in \cR_{NH}$.
	Furthermore, when $e = 0$, if $y\in [y_0,\frac{\sqrt{3}}{2}-1)$ where $y_0 \approx-0.1355$,
	i.e. $0< m \leq m_0 \approx 0.00270963$, the essential part possesses two elliptic eigenvalues and $y \in (-\frac{\sqrt{3}}{2},y_0)$,
	we have that $(\aa,\bb)\in \cR_1$ the essential part possesses two hyperbolic eigenvalues.
	Then there exists an $e_*\in(0,1)$ such that for $e\in[0, e_*)$, $y\in Y_1\bs (-\frac{\sqrt{3}}{2},y_{2,2})$,
	the essential part of the system possesses two pairs of hyperbolic eigenvalues.

	(ii) If $y\in Y_2$, we have $y_{2,1} \approx 0.1403$, $y_{2,2} \approx 0.1796$, $y_{2,3} \approx 0.4224$ and $y_{2,4} \approx 0.4937$.
	When $y\in (0, y_{2,1}) \cup(y_{2,4},\frac{\sqrt{3}}{2})$, $(\aa,\bb,e) \in \cR_{EH}$ and the essential part possesses one pair of hyperbolic eigenvalues and one pair of elliptic eigenvalues.
	When $y\in (y_{2,1},y_{2,2}) \cup(y_{2,3},y_{2,4})$, $(\aa,\bb,e) \in \cR_{NH}$.
	When $y\in (y_{2,2},y_{2,3})$, we have $\aa-3\bb > 0$, and the essential part possesses two pairs of hyperbolic eigenvalues.
	Especially, when $e = 0$, we have that $\bar y_{2,1} \approx0.1548 $ and  $\bar y_{2,2} \approx 0.4679$.
	When $y \in (0,y_{2,1})\cup(0,y_{2,4})$, $(\aa,\bb,0) \in \cR_3$ and the essential part possesses one pair of hyperbolic eigenvalues one pair of elliptic eigenvalues.
	When $y \in (y_{2,1},\bar{y}_{2,1})\cup(\bar{y}_{2,2},y_{2,4})$,  $(\aa,\bb,0) \in \cR_{4}$ and the essential part possesses two pairs of hyperbolic eigenvalues.
	When $y \in (\bar{y}_{2,1},\bar{y}_{2,2})$, $(\aa,\bb,0) \in \cR_1$ and the essential part possesses two pairs of hyperbolic eigenvalues.

	(iii) If $y \in Y_3$ and $e\in[0,1)$, the essential part possesses one pair of elliptic eigenvalues and one pair of hyperbolic eigenvalues.
\end{theorem}

\begin{figure}[htbp]
	\centering
	\begin{minipage}[t]{0.48\textwidth}
		\centering
		\includegraphics[scale=0.65]{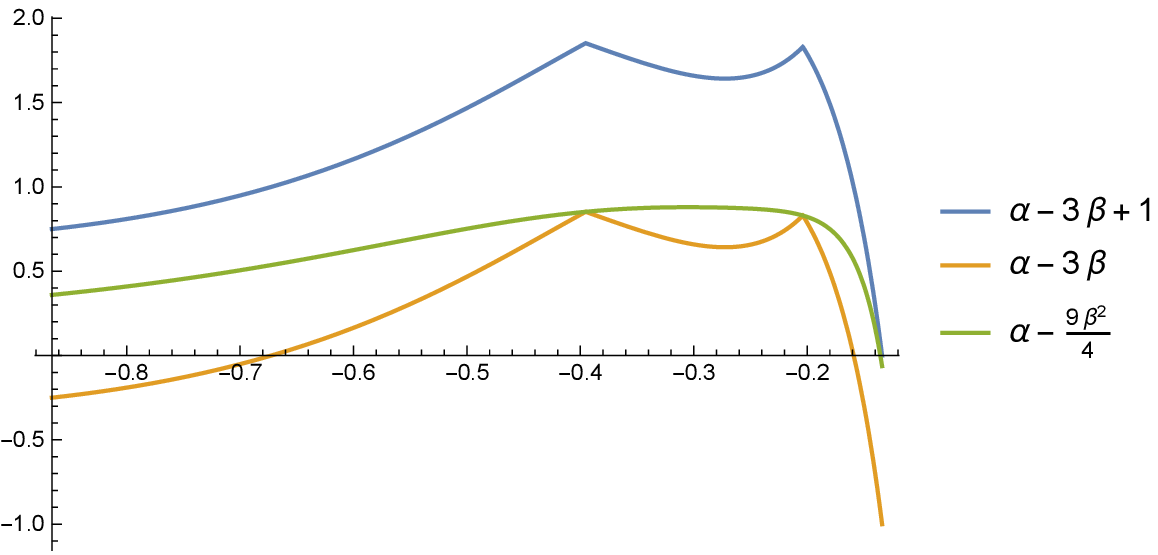}
		\captionof*{figure}{The cases of $y\in Y_1$.}
	\end{minipage}
	\begin{minipage}[t]{0.48\textwidth}
		\centering
		\includegraphics[scale=0.65]{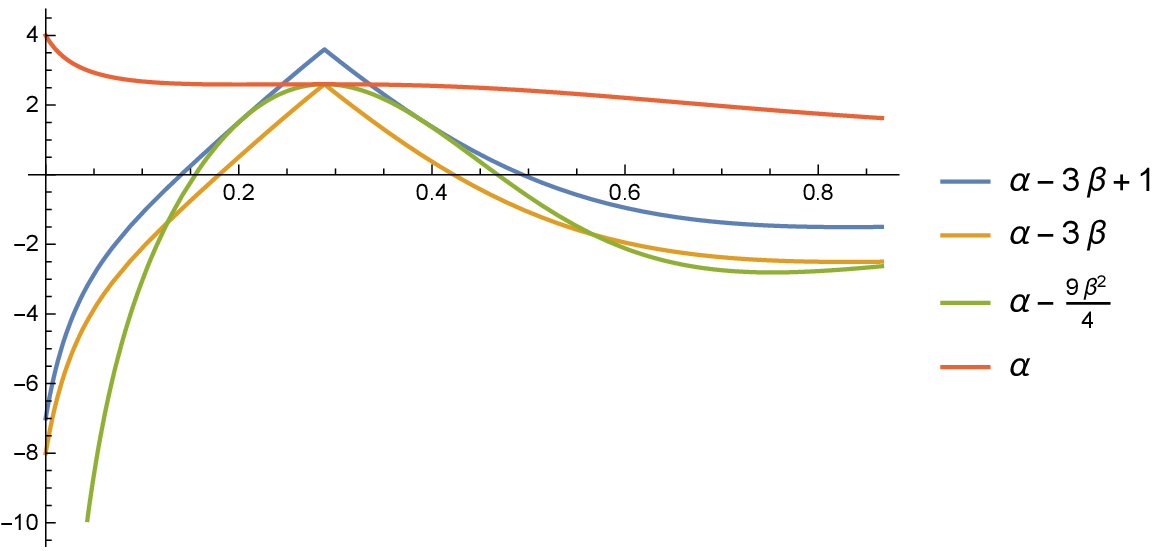}
		\captionof*{figure}{The cases of $y\in Y_2$.}
	\end{minipage}
	\\
\begin{minipage}[t]{0.48\textwidth}
	\centering
	\includegraphics[width=6cm]{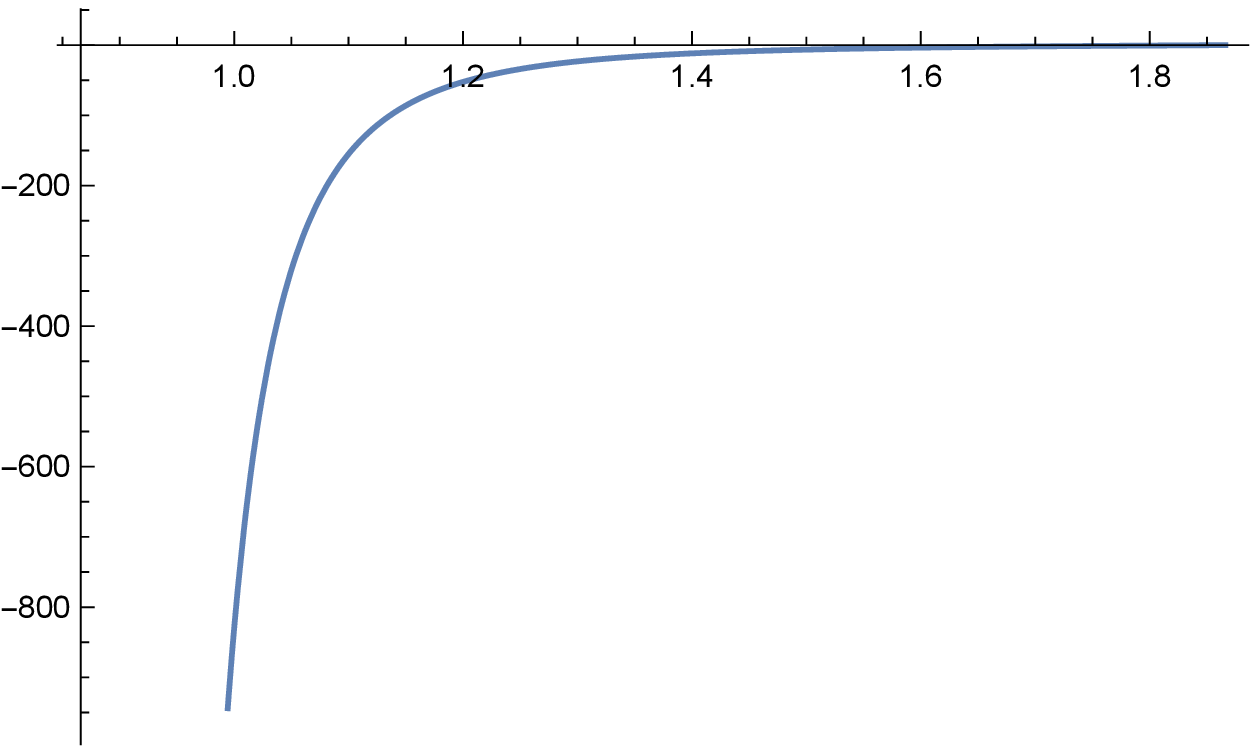}
	\captionof*{figure}{The case of $\aa-3\bb+1 < 0$ when $y\in Y_3$.}
\end{minipage}
\caption{These are the figures of $\aa$, $\aa-3\bb+1$, $\aa-3\bb$ and $\aa-\frac{9}{4}\bb^2$ when $y\in Y$.}
\end{figure}
Note that the case of (iii) of Theorem \ref{thm:two.equal.numerical} holds by Theorem \ref{Th:stability.of.Pi_N}. Then we only prove (i) and (ii) by the assistant of numerical computations
\begin{proof}
	(i)
	If $y \in Y_1$, the numerical computations show that $\aa>3\bb-1$.
	Also $\aa \geq 3\bb$ when $y \in [y_{1,1}, y_{1,2}]$.
	By \ref{thm:hyper.region} of Theorem \ref{thm:gen.ind.nul} , for $e\in [0,1)$, the essential part possesses two pairs of hyperbolic eigenvalues when $y \in [y_{1,1}, y_{1,2}]$.
	When $y\in Y_1\bs[y_{1,1}, y_{1,2}]$, we have $(\aa,\bb,e) \in \cR_{NH}$.

	For $\aa -3\bb+1 = 0$, there exist two roots which are $y_{1,1}\approx -0.6724$ and $y_{1,2} \approx -0.1590$ such that when $y \in [y_{1,1}, y_{1,2}]$, $\aa \geq 3\bb$ then for $e\in [0,1)$,
	the essential part of the system possesses two pairs of hyperbolic eigenvalues. When $y\in Y_1\bs[y_{1,1}, y_{1,2}]$,

	When $e = 0$, the numerical computations show that when $y \in (-\frac{\sqrt{3}}{2},y_0)$ where $y_0 \approx-0.1355$, $\aa >\frac{9}{4}\bb^2$.
	This yields that $(\aa,\bb,0) \in \cR_1$ when $y \in (-\frac{\sqrt{3}}{2},y_0)$.
	When $y \in [y_0,\frac{\sqrt{3}}{2}-1)$, $\aa\leq \frac{9}{4}\bb^2$ and $\aa> 3\bb-1$, this yields that $(\aa,\bb,0) \in \cR_2$.
	Especially, when $y = y_0$, $\aa <1$.
	By the continuous of the degenerate region, we have that there exists an $e_*\in (0,1)$ such that when $y \in (-\frac{\sqrt{3}}{2}, y_{2,2})$,
	the two pairs of the eigenvalues are always hyperbolic when $e\in [0,e_*)$.

	(ii)
	If $y\in Y_2$, there exist $y_{2,1} \approx 0.1403$, $y_{2,2} \approx 0.1796$, $y_{2,3} \approx 0.4224$ and $y_{2,4} \approx 0.4937$ such that
	when $y\in (0, y_{2,1}) \cup(y_{2,4},\frac{\sqrt{3}}{2})$, $\aa-3\bb+1 <0$. Then $(\aa,\bb,e) \in \cR_{EH}$ and the essential part possesses one pair of hyperbolic eigenvalues and one pair of elliptic eigenvalues.
	When $y\in (y_{2,1},y_{2,2}) \cup(y_{2,3},y_{2,4})$, $\aa-3\bb+1 >0$ and $\aa-3\bb < 0$.
	Then $(\aa,\bb,e) \in \cR_{NH}$. When $y\in (y_{2,2},y_{2,3})$, $\aa-3\bb > 0$.
	The essential part possesses two pairs of hyperbolic eigenvalues.
	Especially, when $e = 0$, we have that $\bar y_{2,1} \approx0.1548 $ and  $\bar y_{2,2} \approx 0.4679$.
	When $y \in (0,y_{2,1})\cup(0,y_{2,4})$,
	$\aa-3\bb+1 < 0$ and $\aa-\frac{9}{4}\bb^2< 0$,
	this yields that $(\aa,\bb,0) \in \cR_3$ and the essential part possesses one pair of hyperbolic eigenvalues one pair of elliptic eigenvalues.
	When $y \in (y_{2,1},\bar{y}_{2,1})\cup(\bar{y}_{2,2},y_{2,4})$, $\aa> 1$, $\aa-3\bb+1 > 0$ and $\aa-\frac{9}{4}\bb^2< 0$,
	this yields that $(\aa,\bb,0) \in \cR_4$ and the essential part possesses two pairs of hyperbolic eigenvalues.
	This yields that $(\aa,\bb,0) \in \cR_1$ and the essential part possesses two pairs of hyperbolic eigenvalues.
\end{proof}

\noindent {\bf Acknowledgements.}
Part of this work was done while B. Liu was visiting the Nankai University and Shandong University; he
sincerely thanks Professor Yiming Long and Professor Xijun Hu for their invitations and inspiring discussion.
The second author thank sincerely Professor Yiming Long for his precious help and valuable comments.

\end{document}